\documentclass[12pt]{article}%
\usepackage{amssymb}
\usepackage{amsfonts}
\usepackage{amsmath}
\usepackage{amsthm}
\usepackage{bbm}
\usepackage[nohead]{geometry}
\usepackage[singlespacing]{setspace}
\usepackage[bottom]{footmisc}
\usepackage{indentfirst}
\usepackage{minitoc}  
\usepackage{graphicx}%
\usepackage{rotating}
\usepackage{placeins}
\usepackage{enumerate}
\usepackage{booktabs}
\usepackage{subcaption}
\usepackage[flushleft]{threeparttable}
\usepackage{float}
\usepackage{verbatim}
\usepackage{epstopdf}
\usepackage{pdflscape}
\usepackage{multirow}
\usepackage{enumitem}
\usepackage{mathrsfs}

\usepackage{bibunits} 

\usepackage{siunitx}
\defaultbibliographystyle{plainnat}
\defaultbibliography{ref} 

\newtheorem{theorem}{Theorem}[section]
\newtheorem{lemma}{Lemma}

\newtheorem{proposition}{Proposition}
\newtheorem{corollary}{Corollary}
\newtheorem{assumption}{Assumption}

\newcounter{remark}
\newenvironment{remark}[1][Remark]{\refstepcounter{remark}\begin{trivlist}
		\item[\hskip \labelsep {\bfseries #1 \theremark.\ }]}{\end{trivlist}}
 
\newcommand{\mP}{\mathbb{P}}
\newcommand{\E}{\mathbb{E}}
\newcommand{\Var}{\mathrm{Var}}

\newcommand{\1}{\mathbf{1}}
\newcommand{\mG}{\mathbb{G}}

\usepackage[usenames,dvipsnames]{color}
\usepackage[round]{natbib}
\usepackage[hyperfootnotes=false]{hyperref}
\hypersetup{
	colorlinks,
	citecolor=Blue,
	linkcolor=Blue,
	urlcolor=Blue}

\numberwithin{equation}{section}

\begin{document}
	\begin{bibunit}

	\begin{spacing}{1.2}
		
		\title{Orthogonal Integrated Conditional Moment Tests for Treatment Effect Heterogeneity}
        
		\author{Haokun Lu\thanks{Department of Business Statistics and Econometrics,
				Guanghua School of Management, Peking University, Beijing, 100871, China. Email: \texttt{hklu25@stu.pku.edu.cn}.}\\Peking University \and Xiaojun Song\thanks{Department of Business Statistics and Econometrics, Guanghua School of Management, 
                Peking	University, Beijing, 100871, China. Email: \texttt{sxj@gsm.pku.edu.cn}. This work was supported by the National Natural Science Foundation of China [Grant Numbers 72373007 and 72333001]. The author also gratefully acknowledges the research support from the Center for Statistical Science of Peking University, China, and the Key Laboratory of Mathematical Economics and Quantitative Finance (Peking University) of the Ministry of Education, China.}\\Peking University}
		\maketitle

    \begin{abstract}
        We propose a nonparametric integrated conditional moment (ICM) test for treatment effect heterogeneity across subpopulations defined by a given covariate subvector. Under unconfoundedness, the null is recast as a conditional moment restriction based on a Neyman-orthogonal score, which reduces the first-order sensitivity of the empirical process to nuisance parameter estimation. The test statistics are constructed as continuous functionals of a marked empirical process. We establish a uniform feasible-to-oracle approximation and derive the asymptotic properties of these test statistics under the null and fixed alternatives. We further show that the test has nontrivial power against local alternatives converging to the null at the \(n^{-1/2}\) rate, and develop an easy-to-implement multiplier bootstrap for feasible inference. We also develop extensions to tests of parametric CATE specifications and to settings with endogenous treatment and a binary instrument. Finally, we apply the proposed testing approach to study whether the effect of maternal smoking during pregnancy on infant birth weight varies with maternal age.           \par\vspace{0.5em}
            
            \noindent\textbf{Keywords}: Doubly Robust; Empirical Processes; Integrated Conditional Moments; Neyman Orthogonality; Treatment Effect Heterogeneity.
	
	\noindent\textbf{JEL Classification Number:} C12; C14; C15.

		\end{abstract}
	\end{spacing}

	\thispagestyle{empty}
	
	\newpage
	
	\normalsize
	
	\section{Introduction}\label{sec:intro}

In program evaluation, it has long been recognized that treatment effects may vary systematically across individuals and subpopulations with different observed characteristics; see, for example, \citet{heckman1985alternative}, \citet{heckman1997making}, and \citet{imbens2009recent}. Such heterogeneity matters because an intervention that is beneficial on average may be less effective, ineffective, or even harmful for particular groups. Average treatment effects remain useful summaries, but they do not reveal how effects vary with economically meaningful characteristics such as age, income, education, or baseline risk. Understanding these differences is therefore important for interpreting empirical findings and assessing their external validity, and may also inform subsequent treatment-assignment decisions; see, for example, \citet{manski2004statistical}.

A natural way to study such heterogeneity is through the conditional average treatment effect (CATE), which characterizes how the average effect of a treatment varies across subpopulations defined by observed characteristics. In observational studies, however, the covariates needed for causal identification need not coincide with those along which treatment-effect heterogeneity is substantively of interest. Let \(X\) denote the full vector of pretreatment covariates used to adjust for treatment selection, and let \(Z\subseteq X\) denote the prespecified covariates of interest that define the relevant subpopulations. Thus, \(X\) is chosen to support a credible selection-on-observables strategy, whereas \(Z\) reflects the substantive heterogeneity question under study. This distinction is common in empirical work: researchers may require a relatively rich set of controls to account for treatment selection, while being primarily interested in how the treatment effect varies with a particular characteristic such as age, education, income, prior earnings, or baseline risk. In our empirical application, for example, maternal and pregnancy characteristics are used for selection adjustment, while the heterogeneity of interest is with respect to mother's age.

In this setting, a growing literature has developed methods for estimating and conducting inference on the CATE indexed by covariates of interest. \citet{abrevaya2015estimating} propose inverse-probability-weighted estimators and pointwise inference when the covariates defining the CATE form a subset of those used for selection adjustment. \citet{lee2017doubly} develop doubly robust estimators and uniform confidence bands, while \citet{fan2022estimation} accommodate high-dimensional first-step estimation through orthogonal scores and provide both pointwise and uniform inference. These methods allow researchers to estimate and visualize how treatment effects vary with \(Z\), together with the associated statistical uncertainty. A nonflat estimated CATE profile, however, need not by itself imply statistically significant heterogeneity. This observation motivates a direct specification question: whether the CATE is constant over \(Z\). Beyond constancy, one may further ask whether systematic variation in the CATE can be adequately summarized by a parsimonious functional form.

In this paper, we develop a nonparametric test for CATE heterogeneity with respect to prespecified covariates of interest. We first rewrite the homogeneity null as a conditional moment restriction. Following the integrated conditional moment (ICM) approach pioneered by \citet{bierens1982consistent}, we transform this restriction into a continuum of unconditional moment conditions indexed by the covariates of interest. The sample analog of these moments defines a marked empirical process, and we construct Kolmogorov--Smirnov and Cramér--von Mises test statistics as continuous functionals of this process. Importantly, the proposed procedure works directly with the moment restrictions and does not require nonparametrically estimating or smoothing the CATE function itself.

The main econometric difficulty arises because the conditional moment restriction depends on unknown nuisance functions, including the propensity score and the two conditional outcome regressions. We construct the test using an augmented inverse-probability-weighted, or doubly robust, score \citep{robins1994estimation,hahn1998role,bang2005doubly} and show that the resulting population moments are Neyman-orthogonal with respect to all three nuisance functions. They are therefore locally insensitive to first-order perturbations of these functions. We estimate the propensity score using the series logit estimator of \citet{hirano2003efficient} and the outcome regressions by sieve least squares. Under primitive smoothness and series-rate conditions, we establish a uniform feasible-to-oracle approximation: the feasible process based on the estimated nuisance functions is uniformly asymptotically equivalent to an oracle process constructed from the population doubly robust score. Building on this approximation, we derive the weak limit of the test process under the null, establish consistency against fixed alternatives, and obtain nontrivial power against alternatives approaching homogeneity at the \(n^{-1/2}\) rate. The same oracle representation also supports a simple and computationally efficient multiplier bootstrap. Because the effects of first-step estimation are uniformly asymptotically negligible, the bootstrap can be implemented by resampling the estimated score residuals while holding the nuisance estimates fixed, thereby avoiding repeated estimation of the propensity score and outcome regressions across bootstrap draws.

To accommodate related specification and identification questions, we consider two extensions. First, we develop tests of whether the CATE belongs to a prespecified parametric family, which may include, for example, linear or quadratic forms. Second, building on the local average treatment effect framework of \citet{imbens1994identification}, we adapt the proposed approach to test homogeneity of the conditional LATE when treatment is endogenous and a binary instrument is available. Monte Carlo experiments show that the tests have rejection frequencies close to nominal levels under the null and increasing power against a range of heterogeneous alternatives as the sample size grows. We then apply the method to North Carolina vital-statistics data to examine whether the effect of maternal smoking on infant birth weight varies with mother's age. The homogeneity null is not rejected for Black mothers but is strongly rejected for White mothers. For White mothers, a linear CATE in age is not rejected at the $5\%$ level, suggesting that the detected heterogeneity can be summarized approximately by a linear age profile.

The proposed procedures are related to a growing literature on formal tests of treatment-effect restrictions. A natural starting point is \citet{crump2008nonparametric}, an early and influential contribution that develops direct nonparametric tests of zero and homogeneous conditional average treatment effects. Their hypotheses concern the CATE conditional on the full covariate vector \(X\), and their statistics are constructed from series estimates of the treated and untreated conditional outcome regressions. Our procedure differs in both the conditioning structure and the construction of the test statistic. We use the full vector \(X\) for selection adjustment while studying heterogeneity with respect to prespecified covariates \(Z\subseteq X\), and construct a marked empirical process from a doubly robust score rather than directly from estimated outcome regression functions. Also working with restrictions conditional on the full covariate vector, \citet{chang2015nonparametric} develop kernel-based tests for a broad class of conditional equality and inequality restrictions, while \citet{delgado2013conditional} study bootstrap tests of conditional stochastic dominance.

When the heterogeneity of interest is indexed by selected covariates \(Z\subseteq X\), \citet{hsu2017consistent} tests whether the CATE is nonnegative for every value of \(Z\), formulating the null as a conditional moment inequality and drawing on \citet{andrews2013inference} to conduct inference through a collection of unconditional moment inequalities. \citet{sant2021nonparametric} develops tests of zero and homogeneous treatment effects for possibly right-censored duration outcomes, while \citet{cai2025nonparametric} study homogeneity of partially conditional quantile treatment effects using nonparametric estimators of the corresponding quantile-effect functions. More recently, \citet{dukes2026nonparametric} propose nonparametric tests of quantitative and qualitative CATE heterogeneity, motivated by applications in precision medicine and by the policy consequences of individualized treatment rules. Their framework evaluates heterogeneity through contrasts over prespecified classes of treatment rules, with the choice of rule class determining the alternatives to which the tests are most sensitive. We instead formulate constancy and more general prespecified functional restrictions on the mean CATE as conditional moment restrictions and use the ICM principle to construct Kolmogorov--Smirnov and Cramér--von Mises tests that are consistent against fixed alternatives.

The remainder of the paper is organized as follows. Section \ref{sec:cate} formulates the testing problem, introduces the orthogonal ICM process, and defines the KS and CvM statistics. Section \ref{sec:asy} establishes the feasible-to-oracle approximation and the null, fixed-alternative, and local-alternative limits. Section \ref{sec:mb} develops the multiplier bootstrap. Section \ref{sec:extension} presents the parametric-form and conditional-LATE extensions. Sections \ref{sec:simu} and \ref{sec:empirical} report the Monte Carlo study and empirical illustration, respectively. Proofs and additional simulation results are collected in the Online Appendix.

\textsc{Notation}. For an index set \(\mathcal A\), let \(\ell^\infty(\mathcal A)\) denote the Banach space of bounded real-valued functions on \(\mathcal A\), equipped with the supremum norm \(\|f\|_\infty:=\sup_{a\in\mathcal A}|f(a)|\). We write \(X_n\rightsquigarrow X\) for weak convergence of random elements in \(\ell^\infty(\mathcal A)\) in the Hoffmann--J{\o}rgensen sense; see, for example, Definition 1.3.3 of \citet{van1996weak}. Convergence in probability and in distribution are denoted by \(\xrightarrow{p}\) and \(\xrightarrow{d}\), respectively. For a deterministic sequence \(a_n>0\) and random elements \(X_n\) in a normed space, \(X_n=o_p(a_n)\) means that \(\|X_n\|/a_n\xrightarrow{p}0\), whereas \(X_n=O_p(a_n)\) means that \(\|X_n\|/a_n\) is bounded in probability.

\section{Testing Framework}
\label{sec:cate}

We first introduce the framework considered in this paper. We work with the potential-outcome setup of \citet{rubin1974estimating}. Let \(D\in\{0,1\}\) denote the treatment indicator, where \(D=1\) if an individual receives the treatment and \(D=0\) otherwise. Let \(Y(1)\) and \(Y(0)\) denote the potential outcomes under treatment and control, respectively. The observed outcome is \(Y=DY(1)+(1-D)Y(0)\). Let \(X\) be a vector of pretreatment covariates. Throughout the paper, we observe an i.i.d. sample \(W_i:=(Y_i,D_i,X_i)\), \(i=1,\ldots,n\).

Let \(Z\) denote a subvector of \(X\) containing the covariates of interest. In many applications, \(Z\) is low-dimensional even when \(X\) is relatively rich. Define the conditional average treatment effect given \(Z=z\) as $\tau(z):=\E[Y(1)-Y(0)\mid Z=z]$. In this section, we are interested in testing whether the treatment effect is heterogeneous with respect to \(Z\). Formally, we test
\begin{equation}
\label{eq:test_het}
    \mathbb H_0:\ \exists \tau\in\mathcal T \text{ such that } \tau(z)=\tau \text{ for a.e. } z\in\mathcal Z,
\quad \text{versus}\quad
\mathbb H_1:\ \mP\{\tau(Z)\neq \tau\}>0, \forall  \tau\in\mathbb R.
\end{equation}

To identify \(\tau(z)\) from the observed data, we impose the following standard conditions.

\begin{assumption}[Unconfoundedness]
\label{ass:unconfd}
\((Y(0),Y(1))\perp D\mid X\).
\end{assumption}

\begin{assumption}[Overlap]
\label{ass:overlap}
Let \(p(x):=\mP(D=1\mid X=x)\) denote the propensity score. There exists a constant \(\epsilon>0\) such that $\epsilon\le p(X)\le 1-\epsilon$ a.s.
\end{assumption}

Assumption \ref{ass:unconfd} is the selection-on-observables restriction introduced by \citet{rosenbaum1983central}: conditional on the vector of pre-treatment covariates \(X\), treatment assignment is independent of the potential outcomes. Assumption \ref{ass:overlap} rules out regions of the covariate support in which treatment status is deterministic. Together, these assumptions allow us to express the CATE in terms of observed-data nuisance functions.

Let \(\mu_d(x):=\E[Y\mid D=d,X=x]\), \(d=0,1\). 
To reduce the first-order impact of nuisance estimation on the subsequent testing procedure, we work with the doubly robust score representation (also known as the augmented inverse-probability-weighted form); see, e.g., \citet{robins1994estimation}, \citet{hahn1998role}, and \citet{bang2005doubly}. Its Neyman-orthogonality property will be formalized in Proposition \ref{prop:dr-gateaux} below. For a generic nuisance value \(\eta=(q,m_1,m_0)\), define
\[
\psi(W;\eta)
:=
m_1(X)-m_0(X)
+\frac{D\{Y-m_1(X)\}}{q(X)}
-\frac{(1-D)\{Y-m_0(X)\}}{1-q(X)}.
\]
Let \(\eta_0=(p,\mu_1,\mu_0)\) denote the true nuisance value. Under Assumptions \ref{ass:unconfd} and \ref{ass:overlap}, the CATE is identified by
\[
\tau(z)
=
\E[\psi(W;\eta_0)\mid Z=z].
\]
Consequently, the average treatment effect is identified as
\[
\tau_0:=\E[Y(1)-Y(0)]=\E[\tau(Z)]=\E[\psi(W;\eta_0)].
\]
Under the null in \eqref{eq:test_het}, the constant must equal \(\tau_0\) by the law of iterated expectations. Therefore, the testing problem in \eqref{eq:test_het} can be equivalently represented as the following conditional moment restriction:
\begin{equation}
    \label{eq:cate_cmr}
    \mathbb H_0:\E[\psi(W;\eta_0)-\tau_0\mid Z]=0
\quad \text{a.s.},\quad \text{versus}\quad \mathbb H_1:\mP\!\left(
\E[\psi(W;\eta_0)-\tau_0\mid Z]=0
\right)<1 .
\end{equation}

Following \citet{bierens1982consistent}, \citet{bierens1997asymptotic}, and \citet{stute1997nonparametric}, we adopt an integrated conditional moment (ICM) approach. Specifically, we test the conditional moment restriction in \eqref{eq:cate_cmr} through the family of unconditional moments
\begin{equation}
\E\!\left[\{\psi(W;\eta_0)-\tau_0\}w(Z,z)\right]=0,
\qquad z\in\Pi,
\label{eq:cate_umr}
\end{equation}
where \(\Pi\) is an indexing set and \(w(\cdot,z)\) is a measurable weighting function indexed by \(z\). Under suitable conditions on the class \(\{w(\cdot,z):z\in\Pi\}\), the family of unconditional moments in \eqref{eq:cate_umr} is equivalent to the conditional moment restriction in \eqref{eq:cate_cmr}; see \citet{bierens1997asymptotic} and \citet{escanciano2006goodness}. In this paper, following \citet{stute1997nonparametric}, we use the lower-orthant indicator weight \(w(Z,z)=\1\{Z\le z\}\) with \(\Pi=\mathcal Z\).\footnote{If \(Z\) is vector-valued, the inequality is understood componentwise, i.e., \(\1\{Z\le z\}=\1\{Z_1\le z_1,\ldots,Z_{d_z}\le z_{d_z}\}\).}

We now explain why we use the doubly robust score in the ICM moments. The key motivation is Neyman orthogonality. Moment conditions with this property have reduced first-order sensitivity to nuisance perturbations. Formally, Neyman orthogonality means that the moment condition has zero G\^ateaux derivative with respect to the nuisance functions. In the present setting, define, for each \(z\in\mathcal Z\),
$M(z;\eta,\tau)
:=
\E\!\left[
\{\psi(W;\eta)-\tau\}\1\{Z\le z\}
\right].$
The result is presented in the following proposition.

\begin{proposition}[Neyman orthogonality of the ICM moments]
\label{prop:dr-gateaux}
Suppose Assumptions \ref{ass:unconfd} and \ref{ass:overlap} hold. Let
\(\eta=(q,m_1,m_0)\) be an arbitrary nuisance value such that \(q\) is bounded away from zero and one. For \(r\in[0,1]\), define the path
\(\eta^r:=\eta_0+r(\eta-\eta_0)\), where \(\eta_0=(p,\mu_1,\mu_0)\). Suppose that, for each \(z\in\mathcal Z\),
\(\sup_{r\in[0,1]}\left|\partial[\{\psi(W;\eta^r)-\tau\}\1\{Z\le z\}]/\partial r\right|\) is integrable. Then
\[
\left.
\frac{\partial}{\partial r}
M(z;\eta^r,\tau)
\right|_{r=0}
=0,
\qquad z\in\mathcal Z,
\]
where \(\tau\) need not equal its true value. Moreover, for any fixed nuisance value \(\eta\) and any scalar direction \(h_\tau\in\mathbb R\),
\[
\left.
\frac{d}{dt}
M(z;\eta,\tau_0+t h_\tau)
\right|_{t=0}
=
-h_\tau F_Z(z),
\qquad z\in\mathcal Z.
\]
\end{proposition}

The first part of Proposition \ref{prop:dr-gateaux} gives Neyman orthogonality with respect to the infinite-dimensional nuisance functions. Thus, the ICM moments based on the doubly robust score are locally insensitive to first-order perturbations of \(p\), \(\mu_1\), and \(\mu_0\). The second part shows that estimation of the scalar parameter \(\tau_0\) contributes the projection term \(F_Z(z)\). This term will appear in the oracle asymptotic representation, although the test itself is constructed using the uncentered lower-orthant indicator.

We now construct the feasible test process. Let \(\hat\eta=(\hat p,\hat\mu_1,\hat\mu_0)\) be estimators of \(\eta_0=(p,\mu_1,\mu_0)\). Define the feasible doubly robust score
\begin{equation}
    \label{eq:hat_psi}
    \hat\psi_i
:=
\psi(W_i;\hat\eta)
=
\hat\mu_1(X_i)-\hat\mu_0(X_i)
+\frac{D_i\{Y_i-\hat\mu_1(X_i)\}}{\hat p(X_i)}
-\frac{(1-D_i)\{Y_i-\hat\mu_0(X_i)\}}{1-\hat p(X_i)}.
\end{equation}
The average treatment effect is estimated by
\begin{equation}
    \label{eq:hat_tau}
    \hat\tau:=\frac1n\sum_{i=1}^n \hat\psi_i.
\end{equation}
Multiplying the sample analog of \eqref{eq:cate_umr} by \(\sqrt n\), we obtain the feasible ICM process
\begin{equation}
\widehat R_n(z)
:=
\frac1{\sqrt n}\sum_{i=1}^n
(\hat\psi_i-\hat\tau)\1\{Z_i\le z\},
\qquad z\in\mathcal Z.
\label{eq:cate_process}
\end{equation}
Because \(n^{-1}\sum_{i=1}^n(\hat\psi_i-\hat\tau)=0\), \eqref{eq:cate_process} is algebraically equivalent to the empirical-centered form
\begin{equation}
    \label{eq:equiv}
    \widehat R_n(z)
    =\frac1{\sqrt n}\sum_{i=1}^n
    \hat\psi_i
    \{\1\{Z_i\le z\}-\widehat F_{Z,n}(z)\}=
    \frac1{\sqrt n}\sum_{i=1}^n
    (\hat\psi_i-\hat\tau)
    \{\1\{Z_i\le z\}-\widehat F_{Z,n}(z)\},
\end{equation}
where \(\widehat F_{Z,n}(z):=n^{-1}\sum_{i=1}^n\1\{Z_i\le z\}\). Hence the process is implemented with the intuitive uncentered indicator, while the centering induced by \(\hat\tau\) is built in automatically.

By the equivalence between the conditional moment restriction in \eqref{eq:cate_cmr} and the ICM moments in \eqref{eq:cate_umr}, departures from \(\mathbb H_0\) are reflected in systematic deviations of \(\widehat R_n\) from zero. We therefore use the Kolmogorov--Smirnov and Cram\'er--von Mises functionals
\[
KS_n:=\sup_{z\in\mathcal Z}|\widehat R_n(z)|,
\qquad
CvM_n:=\int_{\mathcal Z}\widehat R_n(z)^2\,d\widehat F_{Z,n}(z)
=
\frac1n\sum_{i=1}^n \widehat R_n(Z_i)^2.
\]
The null hypothesis is rejected for sufficiently large values of either statistic. 

The construction above is stated for generic first-step estimators. In the following section, we specify the sieve estimators used to construct \(\hat\eta\), substitute them into \eqref{eq:hat_psi}--\eqref{eq:cate_process}, and establish the feasible-to-oracle approximation that underlies the asymptotic theory.

\section{Asymptotic Theory}
\label{sec:asy}

We now specify the first-step estimators and derive the asymptotic properties of the feasible ICM process. The propensity score is estimated by the series logit estimator of \citet{hirano2003efficient}. This estimator has been widely used in semiparametric treatment-effect analysis, both for estimation and for testing. It combines flexible nonparametric approximation with the logistic link, thereby keeping the estimated propensity score in the unit interval; see also \citet{firpo2007efficient}, \citet{firpo2016identification}, \citet{hsu2017consistent}, \citet{sant2021nonparametric}, and \citet{cai2025nonparametric}. The outcome regressions are estimated by sieve least squares on the same series space. Formally, let \(r\) denote the dimension of \(X\), and let
\(\varphi=(\varphi_1,\ldots,\varphi_r)'\in\mathbb Z_+^r\) be a vector of nonnegative integers. 
Write \(|\varphi|=\sum_{j=1}^r\varphi_j\), and let \(\{\varphi(k)\}_{k=1}^\infty\) be a sequence containing all distinct elements of \(\mathbb Z_+^r\), ordered so that \(|\varphi(k)|\) is nondecreasing in \(k\). For \(x^\varphi=\prod_{j=1}^r x_j^{\varphi_j}\), define $R^K(x):=\bigl(x^{\varphi(1)},\ldots,x^{\varphi(K)}\bigr)'$,
where \(K\) denotes the number of terms in the series. Let \(\Lambda(a)=\exp(a)/(1+\exp(a))\). The series logit estimator of the propensity score is $\hat p(x)=\Lambda\{R^K(x)'\hat\pi_K\}$, where
\[
\hat\pi_K
=
\arg\max_{\pi_K}
\frac1n\sum_{i=1}^n
\left[
D_i\log\Lambda\{R^K(X_i)'\pi_K\}
+
(1-D_i)\log\{1-\Lambda(R^K(X_i)'\pi_K)\}
\right].
\]
Given \(\hat p\), we estimate the outcome regressions by sieve least squares. Specifically,
\[
\hat\mu_1(x)
=
R^K(x)'
\left(\sum_{i=1}^n R^K(X_i)R^K(X_i)'\right)^{-1}
\left(\sum_{i=1}^n \frac{D_iY_i}{\hat p(X_i)}R^K(X_i)\right),
\]
and
\[
\hat\mu_0(x)
=
R^K(x)'
\left(\sum_{i=1}^n R^K(X_i)R^K(X_i)'\right)^{-1}
\left(\sum_{i=1}^n \frac{(1-D_i)Y_i}{1-\hat p(X_i)}R^K(X_i)\right).
\]
These estimators define \(\hat\eta=(\hat p,\hat\mu_1,\hat\mu_0)\), and the feasible score, ATE estimator, and ICM process are then obtained from \eqref{eq:hat_psi}--\eqref{eq:cate_process}.

To establish the asymptotic properties of the feasible ICM process under the sieve estimators above, we impose the following additional regularity conditions.

\begin{assumption}[Distribution and outcome smoothness]
\label{ass:dist}
\quad
\begin{itemize}
    \setlength{\itemsep}{0pt} 
    \setlength{\parskip}{0pt} 
    \setlength{\parsep}{0pt}
    \item[(i)] The support of the \(r\)-dimensional covariate vector \(X\) is a Cartesian product of compact intervals, \(\mathcal X=\prod_{j=1}^{r}[x_{lj},x_{uj}]\). Moreover, \(X\) admits a density \(f\) on \(\mathcal X\) such that \(0<\underline c\le f(x)\le \bar c<\infty\) for all \(x\in\mathcal X\).
    
    \item[(ii)] For \(d=0,1\), \(\E[|Y(d)|^2]<\infty\), and \(\sup_{x\in\mathcal X}\Var(Y(d)\mid X=x)<\infty\). 
    
    \item[(iii)] The conditional mean functions \(\mu_d(x):=\E[Y\mid D=d,X=x]\), \(d=0,1\), are \(s_\mu\)-times continuously differentiable on \(\mathcal X\), with \(s_\mu>r/2\).
\end{itemize}
\end{assumption}

\begin{assumption}[Propensity-score smoothness]
\label{ass:pscore}
The propensity score \(p(x)\) is \(s\)-times continuously differentiable on \(\mathcal X\), with \(s>5r\).
\end{assumption}

\begin{assumption}[Rates for the series estimator]
\label{ass:SLE}
The number of terms in the series satisfies \(K=a\cdot n^\nu\) for some \(a>0\) and $\max\left\{
{r}/{(2s-4r)},
{r}{/(s+2s_\mu)}
\right\}
<\nu<1/6 $.
\end{assumption}

Similar assumptions have been adopted by \citet{hirano2003efficient}, \citet{crump2008nonparametric}, \citet{hsu2017consistent}, and \citet{sant2021nonparametric}, among others. Assumption \ref{ass:dist} requires the covariates in \(X\) to be continuous. This restriction is imposed mainly to keep the notation and arguments simple. At the expense of additional notation, the framework can accommodate covariates with both continuous and discrete components. In that case, the same procedure can be applied within cells defined by the discrete covariates, or equivalently by interacting the continuous power-series basis with indicators for the discrete covariate values. Similarly, if the covariate of interest \(Z\) contains discrete components, the lower-orthant instrument class can be augmented by indicators for the corresponding discrete cells. See Remark 2 of \citet{hirano2003efficient} and Remark 3.2 of \citet{hsu2017consistent} for related discussions.

The smoothness requirements for \(\mu_d(\cdot)\) and \(p(\cdot)\) in Assumptions \ref{ass:dist} and \ref{ass:pscore} reflect the doubly robust structure of the feasible process. In \citet{hirano2003efficient}, the estimator is IPW-based, and the primitive smoothness requirement is concentrated on the propensity score: \(p(\cdot)\) is assumed to be \(s\)-times continuously differentiable with \(s\ge 7r\), while the outcome regressions are only required to be continuously differentiable. By contrast, \citet{crump2008nonparametric} construct their tests from series estimators of the two outcome regression functions; they require \(\mu_w(\cdot)\) to be \(s\)-times continuously differentiable with \(s/r>25/4\), while the propensity score enters only through overlap. Our feasible ICM process uses both \(\hat p\) and \(\hat\mu_d\) through the doubly robust score. By Proposition \ref{prop:dr-gateaux}, the DR score removes their leading first-order effects in the ICM moments. The remaining plug-in effects are controlled uniformly through the convergence rates of the first-step estimators, which motivates the smoothness conditions \(s>5r\) for \(p(\cdot)\) and \(s_\mu>r/2\) for \(\mu_d(\cdot)\). Finally, the restrictions in Assumptions \ref{ass:dist} and \ref{ass:pscore} guarantee the existence of a \(\nu\) satisfying the conditions in Assumption \ref{ass:SLE}.

The next proposition establishes the key feasible-to-oracle approximation: under the maintained regularity conditions, the feasible ICM process based on estimated nuisance functions is uniformly asymptotically equivalent to its oracle counterpart based on the population doubly robust score.

\begin{proposition}
\label{prop:ICM-process-uniform}
Suppose Assumptions \ref{ass:unconfd}-\ref{ass:SLE} hold. Then, 
$$
\sup_{z\in\mathcal Z}
\left|
\widehat R_n(z)
-
\frac{1}{\sqrt n}\sum_{i=1}^n
\left(\psi(W_i;\eta_0)-\tau_0\right)\bigl(\1\{Z_i\le z\}-F_Z(z)\bigr)
\right|
=o_p(1).
$$
\end{proposition}

\begin{remark}
\label{rem:orthogonal-ipw}
Proposition \ref{prop:ICM-process-uniform} is the uniform sample-level realization of the population orthogonality in Proposition \ref{prop:dr-gateaux}. The latter shows that the ICM moments based on the doubly robust score are locally insensitive to first-order perturbations of \(p,\mu_1,\mu_0\); the former shows that this orthogonality removes the leading effect of nuisance estimation uniformly over \(z\in\mathcal Z\). Consequently, \(\widehat R_n\) has the same first-order representation as the oracle process based on \(\psi(W;\eta_0)\). The centering term \(F_Z(z)\) has a different source: it reflects the estimation of \(\tau_0\) by \(\hat\tau\). Although \(\widehat R_n\) is defined with the uncentered lower-orthant indicator, \eqref{eq:equiv} shows that it is equivalently an empirical-centered process; hence the population-centered weight \(\1\{Z\le z\}-F_Z(z)\) appears naturally in the oracle representation.

It is also useful to relate the DR formulation to the IPW representation used in \citet{hirano2003efficient}. Under unconfoundedness and overlap, the CATE is also identified by
\[
\tau(z)=\E[\phi(W,p)\mid Z=z],
\qquad
\phi(W,p):=\frac{Y\{D-p(X)\}}{p(X)\{1-p(X)\}}.
\]
However, if one constructs the plug-in process using \(\phi(W,\hat p)\), the estimation error in \(\hat p\) contributes a non-negligible first-order term. Its linearization produces the correction
\(-\{\mu_1(X)/p(X)+\mu_0(X)/(1-p(X))\}\{D-p(X)\}\), which is exactly the augmentation component embedded in the DR score. Consequently,
\[
\sup_{z\in\mathcal Z}
\left|
\frac1{\sqrt n}\sum_{i=1}^n
\{\phi(W_i,\hat p)-\hat\tau_{\mathrm{IPW}}\}\1\{Z_i\le z\}
-
\frac1{\sqrt n}\sum_{i=1}^n
\{\psi(W_i;\eta_0)-\tau_0\}\{\1\{Z_i\le z\}-F_Z(z)\}
\right|
=o_p(1),
\]
where \(\hat\tau_{\mathrm{IPW}}=n^{-1}\sum_{i=1}^n\phi(W_i,\hat p)\). Thus the feasible DR process and the appropriately linearized IPW process share the same oracle limit.

A formal statement and proof of this uniform IPW linearization are provided in Appendix \ref{sec:app_add_HIR}. This result may be of independent interest: it extends the estimated-propensity-score expansion of \citet{hirano2003efficient} from the ATE estimator to the present ICM process, uniformly over \(z\in\mathcal Z\). It also shows that the first-order propensity-score correction required by the plug-in IPW process coincides with the augmentation component of the DR score. We therefore use the DR formulation as our baseline, since it incorporates the required correction directly into the score.
\end{remark}

Building on Proposition \ref{prop:ICM-process-uniform}, the next theorem gives the weak limit of the feasible ICM process under the null hypothesis.

\begin{theorem}
\label{thm:ICM-Gaussian}
Suppose Assumptions \ref{ass:unconfd}-\ref{ass:SLE} hold. Then, under $\mathbb H_0$,
$$
\widehat R_n \rightsquigarrow \mG
\qquad \text{in } \ell^\infty(\mathcal Z),
$$
where $\mG$ is a mean-zero Gaussian process with covariance function
$$
K(z_1,z_2)
=
\E\!\left[
\sigma_\psi^2(Z)
\bigl(\1\{Z\le z_1\}-F_Z(z_1)\bigr)
\bigl(\1\{Z\le z_2\}-F_Z(z_2)\bigr)
\right],
\qquad z_1,z_2\in\mathcal Z,
$$
with $\sigma_\psi^2(Z)
:=
\Var\!\left(\psi(W;\eta_0)\mid Z\right).$
\end{theorem}

Theorem \ref{thm:ICM-Gaussian}, together with the continuous mapping theorem, gives the asymptotic null distributions of the KS and CvM statistics. For the CvM statistic, the replacement of \(F_Z\) by \(\widehat F_{Z,n}\) is asymptotically negligible; the argument is provided in Appendix.
\begin{corollary}
\label{cor:null-dist}
Suppose Assumptions \ref{ass:unconfd}--\ref{ass:SLE} hold. Then, under $\mathbb H_0$,
$$
KS_n \xrightarrow{d} \sup_{z\in\mathcal Z} |\mathbb G(z)|,
\qquad
CvM_n \xrightarrow{d} \int_{\mathcal Z} \mathbb G(z)^2\,d F_Z(z),
$$
where $\mathbb G(\cdot)$ is the centered Gaussian process defined in Theorem \ref{thm:ICM-Gaussian}.
\end{corollary}

We next examine the behavior of the test under fixed alternatives. The following theorem shows that the feasible process has a deterministic drift of order \(\sqrt n\), which yields consistency of the proposed tests.
\begin{theorem}
\label{thm:fixed-alt}
Suppose Assumptions \ref{ass:unconfd}--\ref{ass:SLE} hold. Then, under the fixed alternative $\mathbb H_1$,
$$
\sup_{z\in\mathcal Z}
\left|
\frac{1}{\sqrt n}\widehat R_n(z)-\Gamma(z)
\right|
=o_p(1),
$$
where $\Gamma(z):=\E\left[\left(\psi(W;\eta_0)-\tau_0\right)\1\{Z\le z\}\right]$.
\end{theorem}

Under \(\mathbb H_1\), the conditional moment restriction in \eqref{eq:cate_cmr} fails. By the equivalence between the conditional moment restriction and the ICM moments, this implies \(\Gamma\not\equiv0\). Hence \(\sup_{z\in\mathcal Z}|\Gamma(z)|>0\), and, for the CvM statistic, \(\int_{\mathcal Z}\Gamma(z)^2\,dF_Z(z)>0\). Consequently, \(KS_n\) and \(CvM_n\) diverge in probability under fixed alternatives, and the proposed tests are consistent.

We next consider local alternatives that approach the null at the \(n^{-1/2}\) rate. Specifically, let \(\mathbb H_{1n}\) be a sequence of alternatives of the form
\[
\tau_n(z)=\tau_0+n^{-1/2}\lambda(z),
\qquad z\in\mathcal Z,
\]
where \(\lambda:\mathcal Z\to\mathbb R\) is a measurable, non-a.s.-constant function satisfying
\(\E[\lambda(Z)^2]<\infty\).\footnote{If \(\lambda(Z)\) is a.s. constant, then
\(\tau_n(z)\) remains constant in \(z\), so \(\mathbb H_{1n}\) collapses to the null.}

\begin{theorem}
\label{thm:local-alt}
Suppose Assumptions \ref{ass:unconfd}--\ref{ass:SLE} hold. Under \(\mathbb H_{1n}\),
\[
\widehat R_n \rightsquigarrow \mathbb G+\Lambda
\qquad \text{in } \ell^\infty(\mathcal Z),
\]
where \(\mathbb G\) is the centered Gaussian process in Theorem \ref{thm:ICM-Gaussian} and $\Lambda(z):=
\E\!\left[\lambda(Z)\{\1\{Z\le z\}-F_Z(z)\}\right]$. Moreover, $KS_n \xrightarrow{d}\sup_{z\in\mathcal Z}|\mathbb G(z)+\Lambda(z)|$ and $CvM_n \xrightarrow{d}
\int_{\mathcal Z}\{\mathbb G(z)+\Lambda(z)\}^2\,dF_Z(z)$.
\end{theorem}

Since \(\lambda(Z)\) is not a.s. constant, \(\Lambda\not\equiv0\). Thus, under alternatives drifting toward the null at the parametric rate \(n^{-1/2}\), the test statistics have nondegenerate shifted limits, and the proposed tests exhibit nontrivial local power against \(\mathbb H_{1n}\).

From the above theorems, the limiting null distributions of the continuous functionals of \(\widehat R_n\), including \(KS_n\) and \(CvM_n\), depend on the unknown covariance structure of \(\mathbb G\). The limiting distributions are nonpivotal and therefore do not yield closed-form critical values. To address this issue, we use an easy-to-implement multiplier bootstrap procedure, described next.

\section{Computation of Critical Values}
\label{sec:mb}

In this section, we introduce the multiplier bootstrap used to approximate the null distributions of \(KS_n\) and \(CvM_n\). The procedure exploits the feasible first-order representation in Proposition \ref{prop:ICM-process-uniform}; it is computationally simple and avoids re-estimating the first-step nuisance functions across bootstrap replications. The procedure is as follows:
\begin{enumerate}
    \setlength{\itemsep}{0pt}
    \setlength{\parskip}{0pt}
    \setlength{\parsep}{0pt}

    \item Compute the feasible doubly robust scores $\hat\psi_i$, the average $\hat\tau=n^{-1}\sum_{i=1}^n\hat\psi_i$, and the empirical distribution function $\widehat F_{Z,n}(z)=n^{-1}\sum_{i=1}^n\1\{Z_i\le z\}$.

    \item Generate i.i.d. multipliers $\{V_i\}_{i=1}^n$, independent of the data, with zero mean, unit variance, and bounded support. A popular choice is the Mammen two-point distribution, under which $V_i$ takes the values $1-\kappa$ and $\kappa$ with probabilities $\kappa/\sqrt{5}$ and $1-\kappa/\sqrt{5}$, respectively, where $\kappa=(\sqrt{5}+1)/2$, as suggested by \citet{mammen1993bootstrap}.

    \item Compute the bootstrap process
    $$
    \widehat R_n^*(z)
    :=
    \frac1{\sqrt n}\sum_{i=1}^n
    V_i(\hat\psi_i-\hat\tau)\{\1\{Z_i\le z\}-\widehat F_{Z,n}(z)\},
    \qquad z\in\mathcal Z,
    $$
    together with
    $$
    KS_n^*:=\sup_{z\in\mathcal Z}|\widehat R_n^*(z)|,
    \qquad
    CvM_n^*:=\int_{\mathcal Z}\widehat R_n^*(z)^2\,d\widehat F_{Z,n}(z)
    =
    \frac1n\sum_{i=1}^n\widehat R_n^*(Z_i)^2.
    $$

    \item Repeat Steps 2--3 for $B$ times to obtain $\{KS_{n,b}^*,CvM_{n,b}^*\}_{b=1}^B$.

    \item For a given significance level $\alpha\in(0,1)$, let $c_{KS,1-\alpha}^*$ and $c_{CvM,1-\alpha}^*$ denote the empirical $(1-\alpha)$-quantiles of $\{KS_{n,b}^*\}_{b=1}^B$ and $\{CvM_{n,b}^*\}_{b=1}^B$, respectively. Reject $\mathbb H_0$ whenever $KS_n>c_{KS,1-\alpha}^*$ or $CvM_n>c_{CvM,1-\alpha}^*$.
\end{enumerate}

Denote by ``$\rightsquigarrow^*$ in probability'' the weak convergence under the bootstrap law, that is, conditional on the original sample $\{W_i\}_{i=1}^n$. The next theorem establishes the asymptotic validity of the proposed multiplier bootstrap procedure.

\begin{theorem}
\label{thm:bootstrap}
Suppose Assumptions \ref{ass:unconfd}--\ref{ass:SLE} hold. Under $\mathbb H_0$ or the sequence of local alternatives $\mathbb H_{1n}$,
$$
\widehat R_n^{*}\rightsquigarrow^{*}\mathbb G
\qquad
\text{in probability, in } \ell^\infty(\mathcal Z),
$$
where $\mathbb G$ is the centered Gaussian process in Theorem \ref{thm:ICM-Gaussian}.

Under the fixed alternative $\mathbb H_1$,
$$
\widehat R_n^{*}\rightsquigarrow^{*}\mathbb G_1
\qquad
\text{in probability, in } \ell^\infty(\mathcal Z),
$$
where $\mathbb G_1$ is a centered Gaussian process with covariance function
$$
K_1(z_1,z_2)
=
\E\!\left[
\{\psi(W;\eta_0)-\tau_0\}^2
\{\1\{Z\le z_1\}-F_Z(z_1)\}
\{\1\{Z\le z_2\}-F_Z(z_2)\}
\right].
$$
\end{theorem}

On the one hand, under $\mathbb H_0$, Theorem \ref{thm:bootstrap} shows that the multiplier bootstrap process $\widehat R_n^*$ converges weakly to the same Gaussian process $\mathbb G$ as the observed process $\widehat R_n$. Hence, the bootstrap critical values deliver asymptotically correct size for tests based on $KS_n$ and $CvM_n$. Moreover, under the local alternatives $\mathbb H_{1n}$, the bootstrap process still mimics the null Gaussian component, while Theorem \ref{thm:local-alt} shows that the observed process is shifted by the deterministic drift $\Lambda$. Therefore, the multiplier bootstrap preserves the local power properties characterized in Theorem \ref{thm:local-alt}.

On the other hand, under fixed alternatives $\mathbb H_1$, Theorem \ref{thm:bootstrap} implies that $\widehat R_n^*$ remains stochastically bounded, although its covariance kernel generally differs from the null kernel because $\E[\psi(W;\eta_0)-\tau_0\mid Z]\neq0$. By contrast, Theorem \ref{thm:fixed-alt} shows that $KS_n$ and $CvM_n$ diverge in probability. Consequently, tests based on multiplier-bootstrap critical values are consistent against fixed alternatives. Therefore, the multiplier bootstrap procedure is asymptotically valid.

\section{Extensions}
\label{sec:extension}

The preceding sections focus on testing whether the CATE is constant with respect to the covariate of interest. We now discuss two extensions. The first concerns parametric restrictions on the CATE function. When the constant-effect null is rejected, researchers may further ask whether the detected heterogeneity can be summarized by a parsimonious and interpretable functional form, such as a linear or quadratic dependence on the covariate of interest. Such information may be useful for subsequent policy evaluation and statistical decision problems, where policy design may depend on how treatment effects vary across observable characteristics; see, for example, \citet{manski2004statistical}, \citet{kitagawa2018should}, and \citet{athey2021policy}. The second extension considers settings in which treatment assignment may be endogenous but a binary instrument is available. In this case, building on the local average treatment effect (LATE) framework pioneered by \citet{imbens1994identification} and its covariate-adjusted extensions, the proposed ICM approach can be adapted to test heterogeneity of a conditional local average treatment effect.

\subsection{Testing Parametric CATE Specifications}
\label{sec:ex-para}

Our main analysis tests whether the CATE function \(\tau(z)\) is constant in \(Z\). Once treatment-effect homogeneity is rejected, a natural follow-up question is whether \(\tau(z)\) can be described by a prespecified low-dimensional form. This subsection develops such a specification test. Let $\mathcal H:=\{h(\cdot,\theta):\theta\in\Theta\}$, where \(\Theta\subset\mathbb R^{d_\theta}\) is the parameter space and \(h:\mathcal Z\times\Theta\to\mathbb R\) is known up to \(\theta\). We consider the null hypothesis
\[
\mathbb H_0^\dagger:\ \exists\,\theta_0\in\Theta\ \text{such that}\ \tau(z)=h(z,\theta_0)\ \text{for all }z\in\mathcal Z.
\]
The null \(\mathbb H_0\) in the main text is the special case \(h(z,\theta)=\theta\).

Under Assumptions \ref{ass:unconfd} and \ref{ass:overlap}, $\mathbb H_0^\dagger$ is equivalently expressed as
$$
\E\!\left[\psi(W;\eta_0)-h(Z,\theta_0)\mid Z\right]=0
\quad\text{a.s.}
$$
Using the feasible doubly robust scores $\hat\psi_i$ defined in Section \ref{sec:cate}, we estimate $\theta_0$ by nonlinear least squares:
$$
Q_n(\theta)
:=
\frac1n\sum_{i=1}^n
\{\hat\psi_i-h(Z_i,\theta)\}^2,
\qquad
\hat\theta:=\arg\min_{\theta\in\Theta}Q_n(\theta).
$$
This leads to the empirical process
$$
\widehat R_n^\dagger(z)
:=
\frac1{\sqrt n}\sum_{i=1}^n
\{\hat\psi_i-h(Z_i,\hat\theta)\}\1\{Z_i\le z\},
\qquad z\in\mathcal Z.
$$
Large values of suitable functionals of \(\widehat R_n^\dagger\) provide evidence against \(\mathbb H_0^\dagger\). Analogously to Proposition \ref{prop:ICM-process-uniform}, we next derive a uniform first-order representation for \(\widehat R_n^\dagger\). For a function \(h(z,\theta)\) differentiable with respect to the finite-dimensional parameter \(\theta\), we write \(\dot h_\theta(z,\theta)\) and \(\ddot h_{\theta\theta}(z,\theta)\) for its first and second derivatives, respectively. We impose the following additional conditions.
\begin{assumption}[Parameter space and smoothness]\label{ass:param-space}
The parameter space $\Theta\subset\mathbb R^{d_\theta}$ is compact. For each $z\in\mathcal Z$, the map $\theta\mapsto h(z,\theta)$ is twice continuously differentiable on a neighborhood of $\Theta$.
\end{assumption}

\begin{assumption}[Identification]\label{ass:param-ident}
Under \(\mathbb H_0^\dagger\), there exists a unique \(\theta_0\) in the interior of \(\Theta\) such that \(\tau(z)=h(z,\theta_0)\) for all \(z\in\mathcal Z\).
\end{assumption}

\begin{assumption}[Moment bounds and nonsingularity]\label{ass:param-moment}
There exist measurable envelopes $H_0(Z)$, $H_1(Z)$, and $H_2(Z)$ such that $\sup_{\theta\in\Theta}|h(Z,\theta)|\le H_0(Z)$, $\sup_{\theta\in\Theta}\|\dot h_\theta(Z,\theta)\|\le H_1(Z)$, and $\sup_{\theta\in\Theta}\|\ddot h_{\theta\theta}(Z,\theta)\|\le H_2(Z)$ almost surely, with $\E[H_0(Z)^2]<\infty$, $\E[H_1(Z)^2]<\infty$, and $\E[H_2(Z)^2]<\infty$. In addition, $H:=\E[\dot h_\theta(Z,\theta_0)\dot h_\theta(Z,\theta_0)']$ is finite and nonsingular.
\end{assumption}

Assumptions \ref{ass:param-space}--\ref{ass:param-moment} are standard regularity conditions for finite-dimensional extremum estimation; see, for example, \citet{newey1994large}. They ensure that the parametric projection is locally identified, sufficiently smooth, and nonsingular, so that the effect of estimating \(\theta_0\) can be accounted for by the usual first-order projection term. The next proposition gives the resulting uniform feasible-to-oracle approximation for \(\widehat R_n^\dagger\).

\begin{proposition}
\label{prop:parametric-process-uniform}
Suppose Assumptions \ref{ass:unconfd}--\ref{ass:param-moment} hold. Under $\mathbb H_0^\dagger$,
$$
\sup_{z\in\mathcal Z}
\left|
\widehat R_n^\dagger(z)
-
\frac1{\sqrt n}\sum_{i=1}^n
\{\psi(W_i;\eta_0)-h(Z_i,\theta_0)\}
\Bigl\{\1(Z_i\le z)-G(z)'H^{-1}\dot h_\theta(Z_i,\theta_0)\Bigr\}
\right|
=o_p(1),
$$
where
$G(z):=\E[\dot h_\theta(Z,\theta_0)\1\{Z\le z\}]$, $H:=\E[\dot h_\theta(Z,\theta_0)\dot h_\theta(Z,\theta_0)']$.
\end{proposition}

Proposition \ref{prop:parametric-process-uniform} implies that, under \(\mathbb H_0^\dagger\),
\[
\widehat R_n^\dagger \rightsquigarrow \mathbb G^\dagger
\qquad\text{in }\ell^\infty(\mathcal Z),
\]
where \(\mathbb G^\dagger\) is a centered Gaussian process with covariance kernel
\[
K^\dagger(z_1,z_2)=
\E\Big[
\sigma_\psi^2(Z)
\Bigl\{\1(Z\le z_1)-G(z_1)'H^{-1}\dot h_\theta(Z,\theta_0)\Bigr\}
\Bigl\{\1(Z\le z_2)-G(z_2)'H^{-1}\dot h_\theta(Z,\theta_0)\Bigr\}
\Big],
\]
with \(\sigma_\psi^2(Z)=\Var\{\psi(W;\eta_0)\mid Z\}\). The KS and CvM statistics constructed from \(\widehat R_n^\dagger\) therefore converge to the corresponding continuous functionals of \(\mathbb G^\dagger\).

Critical values can be computed by the multiplier bootstrap process
\begin{equation}
    \label{eq:bootstrap_process_para}
    \widehat R_n^{\dagger,*}(z)
    :=
    \frac1{\sqrt n}\sum_{i=1}^n
    V_i\{\hat\psi_i-h(Z_i,\hat\theta)\}
    \Bigl\{\1(Z_i\le z)-\hat G(z)'\hat H^{-1}\dot h_\theta(Z_i,\hat\theta)\Bigr\},
\end{equation}
where \(\hat G(z):=n^{-1}\sum_{i=1}^n \dot h_\theta(Z_i,\hat\theta)\1\{Z_i\le z\}\) and \(\hat H:=n^{-1}\sum_{i=1}^n\dot h_\theta(Z_i,\hat\theta)\dot h_\theta(Z_i,\hat\theta)'\). The term \(\hat G(z)'\hat H^{-1}\dot h_\theta(Z_i,\hat\theta)\) is the projection adjustment for estimating \(\theta_0\). In the constant-effect case \(h(z,\theta)=\theta\), \(\dot h_\theta(Z_i,\theta)=1\), \(\hat H=1\), and \(\hat G(z)=\widehat F_{Z,n}(z)\); hence this adjustment reduces to the empirical centering in \eqref{eq:equiv}.

Under \(\mathbb H_0^\dagger\) and \(\mathbb H_{1n}^\dagger\),
\[
\widehat R_n^{\dagger,*}\rightsquigarrow^* \mathbb G^\dagger
\qquad
\text{in probability, in }\ell^\infty(\mathcal Z).
\]
Therefore, under \(\mathbb H_0^\dagger\), the bootstrap critical values yield asymptotically correct size for the KS and CvM tests. Under \(n^{-1/2}\)-local alternatives \(\mathbb H_{1n}^\dagger\), the bootstrap continues to approximate the null Gaussian component, whereas the observed test process is shifted by a nonzero deterministic projected drift; hence the tests have nontrivial local power. Under fixed alternatives, the observed statistics diverge while the bootstrap process remains stochastically bounded, so the tests are consistent.

\subsection{Testing without Unconfoundedness}
\label{sec:ex-clate}

The ICM framework can also be adapted to settings in which treatment assignment may be endogenous but a binary instrument is available. We follow the local average treatment effect framework of \citet{imbens1994identification}; see also \citet{abadie2003semiparametric} and \citet{frolich2007nonparametric} for covariate-adjusted and nonparametric formulations. Let \(W=(Y,D,A,X)\), where \(D\in\{0,1\}\) is the treatment, \(A\in\{0,1\}\) is a binary instrument, \(X\) is the full covariate vector, and \(Z\) is the covariate of interest, with \(Z\) a subvector of \(X\).

Let \(D(a)\) denote the potential treatment status when the instrument is set to \(a\), and let \(Y(d,a)\) denote the potential outcome when the treatment and instrument are set to \((d,a)\). Under the exclusion restriction below, we write \(Y(d)\equiv Y(d,1)=Y(d,0)\). Let \(\mathcal C:=\{D(1)>D(0)\}\) denote the complier group. The conditional local average treatment effect with respect to \(Z\) is
\[
\tau_L(z):=\E[Y(1)-Y(0)\mid \mathcal C,Z=z].
\]
We test whether the conditional LATE is homogeneous in \(Z\):
\begin{equation}
\label{eq:test_late}
\mathbb H_0^L:\ \exists \tau\in\mathcal T \text{ such that } \tau_L(z)=\tau \text{ for a.e. } z\in\mathcal Z,
\quad
\mathbb H_1^L:\ \mP\{\tau_L(Z)\neq \tau\}>0,\ \forall \tau\in\mathcal T.
\end{equation}

We impose the following standard conditional LATE assumptions.

\begin{assumption}[Conditional LATE model]
\label{ass:clate}
\quad
\begin{itemize}
    \setlength{\itemsep}{0pt}
    \setlength{\parskip}{0pt}
    \setlength{\parsep}{0pt}
    \item[(i)] Instrument independence: \(A\perp \bigl(Y(1,1),Y(1,0),Y(0,1),Y(0,0),D(1),D(0)\bigr)\mid X\).
    \item[(ii)] Exclusion restriction: \(Y(d,1)=Y(d,0)\) for \(d=0,1\).
    \item[(iii)] Monotonicity: \(\mP\{D(1)\ge D(0)\mid X\}=1\) a.s.
    \item[(iv)] First stage: for \(q(x):=\mP(A=1\mid X=x)\), there exists \(\epsilon>0\) such that \(\epsilon\le q(X)\le 1-\epsilon\) a.s.; moreover, \(\mP(\mathcal C\mid Z=z)>0\) for \(z\) in the region of interest.
\end{itemize}
\end{assumption}

Define \(\mu_{Y,a}(x):=\E[Y\mid A=a,X=x]\) and \(\mu_{D,a}(x):=\E[D\mid A=a,X=x]\), \(a=0,1\). The reduced-form and first-stage effects are identified by the doubly robust scores
\[
\psi_Y^L(W;\eta_0)
:=
\mu_{Y,1}(X)-\mu_{Y,0}(X)
+
\frac{A\{Y-\mu_{Y,1}(X)\}}{q(X)}
-
\frac{(1-A)\{Y-\mu_{Y,0}(X)\}}{1-q(X)},
\]
and
\[
\psi_D^L(W;\eta_0)
:=
\mu_{D,1}(X)-\mu_{D,0}(X)
+
\frac{A\{D-\mu_{D,1}(X)\}}{q(X)}
-
\frac{(1-A)\{D-\mu_{D,0}(X)\}}{1-q(X)}.
\]
Under Assumption \ref{ass:clate}, the conditional LATE is identified by the conditional Wald ratio
\[
\tau_L(z)
=
\frac{\E[\psi_Y^L(W;\eta_0)\mid Z=z]}
{\E[\psi_D^L(W;\eta_0)\mid Z=z]}.
\]

Under \(\mathbb H_0^L\), the common value of \(\tau_L(z)\) is invariant to the distribution used to average over \(Z\). For identification and interpretation, we average over the distribution of \(Z\) among compliers, rather than the marginal distribution of \(Z\). This yields the global LATE
\[
\tau_0
:=
\E[Y(1)-Y(0)\mid \mathcal C]
=
\frac{\E[\psi_Y^L(W;\eta_0)]}{\E[\psi_D^L(W;\eta_0)]}
=
\int_{\mathcal Z}\tau_L(z)\,dF_{Z\mid\mathcal C}(z).
\]
Thus the testing problem in \eqref{eq:test_late} is equivalently written as the conditional moment restriction
\begin{equation}
\label{eq:clate-cmr}
\E\!\left[
\psi_Y^L(W;\eta_0)-\tau_0\psi_D^L(W;\eta_0)
\mid Z
\right]
=0
\quad\text{a.s.}
\end{equation}

To construct the feasible process, estimate \(q\), \(\mu_{Y,a}\), and \(\mu_{D,a}\) by the same sieve procedures as in the main analysis, with the instrument \(A\) replacing the treatment indicator $D$. Let \(\widehat\psi_Y^L\) and \(\widehat\psi_D^L\) denote the resulting feasible reduced-form and first-stage scores, and estimate \(\tau_0\) by the sample Wald ratio
\[
\widehat\tau_L
:=
\frac{n^{-1}\sum_{i=1}^n \widehat\psi_Y^L(W_i)}
     {n^{-1}\sum_{i=1}^n \widehat\psi_D^L(W_i)}.
\]
The corresponding empirical process is
\begin{equation}
\label{eq:clate-process}
\widehat R_n^L(z)
:=
\frac1{\sqrt n}\sum_{i=1}^n
\left\{
\widehat\psi_Y^L(W_i)-\widehat\tau_L\widehat\psi_D^L(W_i)
\right\}
\1\{Z_i\le z\},
\qquad z\in\mathcal Z.
\end{equation}
Large values of KS or CvM functionals of \(\widehat R_n^L\) provide evidence against \(\mathbb H_0^L\). The next proposition gives the uniform feasible-to-oracle approximation for this process.

\begin{proposition}
\label{prop:clate-process-uniform}
Suppose Assumption \ref{ass:clate} holds, together with the analogues of Assumptions \ref{ass:dist}--\ref{ass:SLE} for the instrument propensity score \(q\) and the nuisance regression functions entering \(\psi_Y^L\) and \(\psi_D^L\). Then, under \(\mathbb H_0^L\),
\[
\sup_{z\in\mathcal Z}
\left|
\widehat R_n^L(z)
-
\frac1{\sqrt n}\sum_{i=1}^n
\left\{
\psi_Y^L(W_i;\eta_0)
-
\tau_0\psi_D^L(W_i;\eta_0)
\right\}
\left\{
\1(Z_i\le z)-F_{Z\mid \mathcal C}(z)
\right\}
\right|
=o_p(1),
\]
where \(F_{Z\mid \mathcal C}(z):=\mP(Z\le z\mid \mathcal C)\).
\end{proposition}

Proposition \ref{prop:clate-process-uniform} is the IV analogue of Proposition \ref{prop:ICM-process-uniform}. The reduced-form residual $\psi_Y^L(W;\eta_0)-\tau_0\psi_D^L(W;\eta_0)$ plays the role of the oracle score. The centering term is \(F_{Z\mid\mathcal C}(z)\), rather than \(F_Z(z)\), because \(\tau_0\) is the global LATE averaged over the complier distribution.

Under \(\mathbb H_0^L\), Proposition \ref{prop:clate-process-uniform} implies
\[
\widehat R_n^L \rightsquigarrow \mathbb G^L
\qquad\text{in } \ell^\infty(\mathcal Z),
\]
where \(\mathbb G^L\) is a centered Gaussian process with covariance kernel
\[
K^L(z_1,z_2)
=
\E\!\left[
\sigma_L^2(Z)
\left\{\1(Z\le z_1)-F_{Z\mid\mathcal C}(z_1)\right\}
\left\{\1(Z\le z_2)-F_{Z\mid\mathcal C}(z_2)\right\}
\right],
\]
with $\sigma_L^2(Z)
:=
\Var\!\left(
\psi_Y^L(W;\eta_0)-\tau_0\psi_D^L(W;\eta_0)
\mid Z
\right)$. Therefore, the KS and CvM statistics constructed from \(\widehat R_n^L\) converge to the corresponding continuous functionals of \(\mathbb G^L\). In implementation, critical values can be computed by the multiplier bootstrap process
\begin{equation}
    \label{eq:bootstrap_process_clate}
    \widehat R_n^{L,*}(z)
:=
\frac1{\sqrt n}\sum_{i=1}^n
V_i
\{\widehat\psi_Y^L(W_i)-\widehat\tau_L\widehat\psi_D^L(W_i)\}
\{\1\{Z_i\le z\}-\widehat F_{Z\mid\mathcal C}(z)\},
\end{equation}
where $\widehat F_{Z\mid\mathcal C}(z):=\sum_i \widehat\psi_D^L(W_i)\1\{Z_i\le z\}/\sum_i \widehat\psi_D^L(W_i)$. The validity of this bootstrap procedure follows by the same arguments as in the unconfounded and parametric-extension cases. We omit the details.
\section{Monte Carlo Simulation Study}\label{sec:simu}

In this section, we conduct a Monte Carlo study of the finite-sample performance of the proposed heterogeneity test. The designs examine empirical size when the CATE is constant and power when the CATE varies with the covariate of interest. Throughout, the covariate of interest is \(Z=X_1\), while the full covariate vector is \(X=(X_1,\ldots,X_r)'\) with \(r\in\{3,5,10\}\).

We consider sample sizes \(n\in\{500,1000,2000\}\). All results are based on \(5{,}000\) Monte Carlo replications, and bootstrap \(p\)-values are computed using \(B=2{,}000\) multiplier bootstrap draws. We report empirical rejection frequencies at the 1\%, 5\%, and 10\% nominal levels. In each replication, the propensity score is estimated by series logit, and the nuisance regressions \(\mu_1(X)\) and \(\mu_0(X)\) are estimated by sieve least squares using the same type of basis. To keep the number of series terms moderate as the dimension increases, we use a total-degree cubic polynomial basis when \(r=3\), a total-degree quadratic basis when \(r=5\), and an additive quadratic basis when \(r=10\). The additive quadratic basis contains an intercept, all linear terms, and all squared terms, but excludes interactions.

Across all designs, outcomes are generated by
\[
Y=\mu_0(X)+\tau(X_1)D+\varepsilon,
\qquad \tau_0=0.5.
\]
Under the null, \(\tau(X_1)\equiv\tau_0\). Under the alternatives, \(\tau(X_1)\) varies with \(X_1\). We consider both homoskedastic disturbances, \(\varepsilon\sim \mathcal N(0,\sigma_u^2)\) with \(\sigma_u=0.25\), and heteroskedastic disturbances,
\(\varepsilon\sim \mathcal N(0,\sigma_u^2\{1+|X_1-0.5|\}^2)\). The three data-generating processes differ in the covariate distribution, the treatment assignment mechanism, the baseline outcome function, and the form of treatment-effect heterogeneity, as described below.

\paragraph{Case I: Smooth logit design.}
Let \(X_j\stackrel{iid}{\sim}\mathrm{Unif}[0,1]\), \(j=1,\ldots,r\), and write \(U_j=X_j-0.5\). Treatment is generated from
\(D\mid X\sim\mathrm{Bernoulli}(\Lambda\{m(X)\})\). When \(r=3\),
\[
m(X)= -0.2+1.2U_1+0.8U_2+0.5U_3+0.6U_1U_2+0.4U_1^2+1.5U_1^3+0.4e^{-6U_2^2},
\]
and $\mu_0(X)=0.5U_2+0.3U_2U_3+0.2U_1^2+0.15U_1^3$.

For \(r>3\), we add \(0.2\bar U+0.2\bar E\) to \(m(X)\) and replace the last term in \(\mu_0(X)\) by \(0.15\bar U\), where
\(\bar U=(r-3)^{-1}\sum_{j=4}^r U_j\) and
\(\bar E=(r-3)^{-1}\sum_{j=4}^r e^{-6U_j^2}\). Under the alternative, $\tau(X_1)=\tau_0+0.2(X_1-0.5)$.

\paragraph{Case II: Threshold assignment design.}
Let \(X_j\stackrel{iid}{\sim}\mathrm{Unif}[0,1]\), \(j=1,\ldots,r\). Treatment is generated by
\(D=\1\{U<p(X)\}\), where \(U\sim\mathrm{Unif}[0,1]\) is independent of \(X\),
\(p(X)=0.1+0.4S(X)\), and
\[
S(X)=
\begin{cases}
(X_1+X_2+X_3)/3, & r=3,\\
0.8(X_1+X_2+X_3)/3+0.2(r-3)^{-1}\sum_{j=4}^r X_j, & r>3.
\end{cases}
\]
The baseline outcome is \(\mu_0(X)=0.5(X_2-0.5)+0.5(X_3-0.5)\) when \(r=3\), and adds \(0.15\bar U\) when \(r>3\), where \(\bar U=(r-3)^{-1}\sum_{j=4}^r(X_j-0.5)\). Under the alternative, $\tau(X_1)=\tau_0+0.4(X_1-0.5)\1\{X_1>0.5\}$.

\paragraph{Case III: Fractional latent-index design.}
Let \(X_j\stackrel{iid}{\sim}\mathrm{Beta}(2,2)\), \(j=1,\ldots,r\), and write \(U_j=X_j-0.5\). Treatment is generated by \(D=\1\{m(X)-e>0\}\), where \(e\sim \mathcal N(0,1)\), so that \(p(X)=\Phi\{m(X)\}\). The latent index is
\[
m(X)=
\frac{-0.10+0.10\sum_{j=1}^r X_j^2+0.50U_1+0.40U_2+0.30U_1U_2}
{\exp\{-0.10\sum_{j=1}^r X_j^2\}},
\]
and the baseline outcome is $\mu_0(X)=0.4U_1+0.3U_2+0.2U_3^2+0.10r^{-1}\sum_{j=1}^r U_j$. Under the alternative, $\tau(X_1)=\tau_0+0.8(X_1-0.4)^2$.

The main findings are as follows. Table \ref{tab:size} reports empirical rejection frequencies under the constant-CATE null. The results show accurate size control across a broad range of designs. This includes different treatment assignment mechanisms, from a smooth logit propensity score to a threshold assignment rule and a nonlinear latent-index model; different baseline outcome functions; dimensions \(r\in\{3,5,10\}\); and both homoskedastic and heteroskedastic disturbances. Across these settings, the rejection frequencies of both KS and CvM statistics remain close to the nominal 1\%, 5\%, and 10\% levels. Thus, the multiplier bootstrap appears to provide a stable finite-sample approximation to the null distribution, even when the treatment allocation rule and the outcome equation differ substantially across designs.

\begin{landscape}
\begingroup
\centering
\captionof{table}{Empirical size under homoskedastic and heteroskedastic errors.}
\label{tab:size}
\scriptsize
\setlength{\tabcolsep}{2.5pt}
\renewcommand{\arraystretch}{0.65}
\begin{tabular*}{\linewidth}{@{\extracolsep{\fill}} ll*{12}{c} @{}}
\toprule
& & \multicolumn{6}{c}{Homoskedastic errors} & \multicolumn{6}{c}{Heteroskedastic errors} \\
\cmidrule(lr){3-8}\cmidrule(lr){9-14}
& & \multicolumn{2}{c}{$n=500$} & \multicolumn{2}{c}{$n=1000$} & \multicolumn{2}{c}{$n=2000$}
& \multicolumn{2}{c}{$n=500$} & \multicolumn{2}{c}{$n=1000$} & \multicolumn{2}{c}{$n=2000$} \\
\cmidrule(lr){3-4}\cmidrule(lr){5-6}\cmidrule(lr){7-8}
\cmidrule(lr){9-10}\cmidrule(lr){11-12}\cmidrule(lr){13-14}
$r$ & $\alpha$ & KS & CvM & KS & CvM & KS & CvM & KS & CvM & KS & CvM & KS & CvM \\
\midrule
\multicolumn{14}{l}{\textit{Case I: Smooth logit design}} \\
\midrule
\multirow{3}{*}{3}
& 1\% & 1.02 & 0.94 & 1.18 & 1.32 & 1.26 & 1.24 & 1.04 & 1.08 & 1.14 & 1.24 & 1.30 & 1.22 \\
& 5\% & 5.42 & 4.98 & 5.40 & 5.18 & 5.98 & 5.72 & 5.50 & 5.20 & 5.28 & 4.98 & 5.86 & 5.80 \\
& 10\% & 10.52 & 10.36 & 10.06 & 10.40 & 11.50 & 11.00 & 10.90 & 10.64 & 10.12 & 10.58 & 11.20 & 10.86 \\
\addlinespace
\multirow{3}{*}{5}
& 1\% & 1.18 & 1.34 & 1.32 & 1.12 & 1.24 & 1.00 & 1.18 & 1.36 & 1.22 & 1.14 & 1.06 & 0.90 \\
& 5\% & 5.30 & 5.16 & 5.38 & 5.28 & 5.30 & 5.22 & 5.28 & 5.40 & 5.66 & 5.06 & 5.32 & 5.00 \\
& 10\% & 10.50 & 10.32 & 10.62 & 10.56 & 10.46 & 10.14 & 10.86 & 10.46 & 10.76 & 10.48 & 10.50 & 10.34 \\
\addlinespace
\multirow{3}{*}{10}
& 1\% & 1.42 & 1.46 & 0.96 & 1.12 & 1.16 & 0.88 & 1.46 & 1.54 & 1.06 & 1.16 & 1.06 & 0.84 \\
& 5\% & 5.88 & 5.62 & 5.64 & 5.50 & 4.98 & 5.22 & 5.68 & 5.46 & 5.40 & 5.42 & 5.06 & 5.20 \\
& 10\% & 11.08 & 10.46 & 10.52 & 10.42 & 10.60 & 10.40 & 10.82 & 10.52 & 10.72 & 10.58 & 10.40 & 10.24 \\
\midrule
\multicolumn{14}{l}{\textit{Case II: Threshold-assignment design}} \\
\midrule
\multirow{3}{*}{3}
& 1\% & 0.94 & 0.88 & 0.84 & 0.80 & 0.96 & 0.80 & 1.06 & 1.00 & 0.74 & 0.86 & 0.82 & 0.66 \\
& 5\% & 4.34 & 4.38 & 4.28 & 4.54 & 4.66 & 4.18 & 4.78 & 4.40 & 4.32 & 4.46 & 4.34 & 4.12 \\
& 10\% & 9.30 & 9.50 & 9.64 & 9.42 & 9.82 & 9.74 & 9.06 & 9.68 & 9.12 & 8.96 & 9.10 & 9.40 \\
\addlinespace
\multirow{3}{*}{5}
& 1\% & 0.62 & 0.66 & 0.98 & 0.88 & 1.02 & 1.14 & 0.70 & 0.88 & 1.02 & 1.00 & 1.00 & 1.00 \\
& 5\% & 3.94 & 3.94 & 4.38 & 4.52 & 5.04 & 5.00 & 4.34 & 4.56 & 4.72 & 4.82 & 4.94 & 5.00 \\
& 10\% & 8.70 & 8.72 & 9.34 & 9.70 & 10.18 & 9.78 & 9.20 & 9.08 & 9.68 & 9.50 & 9.74 & 9.66 \\
\addlinespace
\multirow{3}{*}{10}
& 1\% & 1.08 & 0.92 & 1.22 & 1.22 & 0.98 & 1.10 & 0.96 & 0.86 & 1.24 & 1.24 & 1.00 & 1.16 \\
& 5\% & 5.34 & 5.34 & 5.34 & 5.42 & 5.52 & 5.10 & 4.84 & 4.98 & 5.46 & 5.24 & 5.26 & 4.86 \\
& 10\% & 10.96 & 11.42 & 11.26 & 10.44 & 10.48 & 9.94 & 10.42 & 10.62 & 10.94 & 10.40 & 10.16 & 9.70 \\
\midrule
\multicolumn{14}{l}{\textit{Case III: Fractional latent-index design}} \\
\midrule
\multirow{3}{*}{3}
& 1\% & 0.94 & 0.94 & 1.06 & 1.16 & 1.08 & 0.92 & 0.94 & 1.06 & 1.20 & 1.12 & 1.00 & 0.94 \\
& 5\% & 4.96 & 4.96 & 5.04 & 5.22 & 4.80 & 4.46 & 4.92 & 5.04 & 5.20 & 5.26 & 4.86 & 4.60 \\
& 10\% & 10.80 & 10.32 & 10.32 & 10.34 & 9.30 & 9.72 & 10.78 & 10.52 & 10.50 & 10.34 & 9.88 & 9.94 \\
\addlinespace
\multirow{3}{*}{5}
& 1\% & 1.08 & 1.06 & 0.98 & 1.08 & 1.24 & 1.24 & 1.08 & 1.12 & 1.00 & 1.08 & 1.20 & 1.26 \\
& 5\% & 5.48 & 5.42 & 5.30 & 4.92 & 4.94 & 4.88 & 5.44 & 5.30 & 5.26 & 5.34 & 4.90 & 5.10 \\
& 10\% & 10.72 & 10.30 & 10.58 & 10.54 & 9.94 & 10.16 & 10.70 & 10.76 & 10.68 & 10.52 & 9.98 & 9.78 \\
\addlinespace
\multirow{3}{*}{10}
& 1\% & 1.32 & 1.28 & 1.00 & 0.78 & 1.02 & 1.08 & 1.22 & 1.26 & 0.94 & 0.72 & 1.10 & 1.10 \\
& 5\% & 5.90 & 6.38 & 4.96 & 4.74 & 5.18 & 5.26 & 5.88 & 6.10 & 4.90 & 4.78 & 5.50 & 5.12 \\
& 10\% & 11.70 & 11.56 & 10.16 & 9.66 & 10.56 & 10.14 & 11.48 & 11.58 & 9.68 & 9.76 & 10.36 & 10.28 \\
\bottomrule
\end{tabular*}

\vspace{0.35em}
\begin{minipage}{0.96\linewidth}
\scriptsize
\textit{Notes.} Entries report rejection frequencies (in percent) over 5,000 Monte Carlo replications under the constant-CATE null. The nominal significance levels are 1\%, 5\%, and 10\%. Bootstrap critical values are based on 2,000 multiplier bootstrap draws.
\end{minipage}
\par
\endgroup
\newpage
\begingroup
\centering
\captionof{table}{Empirical power under homoskedastic and heteroskedastic errors.}
\label{tab:power}
\scriptsize
\setlength{\tabcolsep}{2.5pt}
\renewcommand{\arraystretch}{0.65}
\begin{tabular*}{\linewidth}{@{\extracolsep{\fill}} ll*{12}{c} @{}}
\toprule
& & \multicolumn{6}{c}{Homoskedastic errors} & \multicolumn{6}{c}{Heteroskedastic errors} \\
\cmidrule(lr){3-8}\cmidrule(lr){9-14}
& & \multicolumn{2}{c}{$n=500$} & \multicolumn{2}{c}{$n=1000$} & \multicolumn{2}{c}{$n=2000$}
& \multicolumn{2}{c}{$n=500$} & \multicolumn{2}{c}{$n=1000$} & \multicolumn{2}{c}{$n=2000$} \\
\cmidrule(lr){3-4}\cmidrule(lr){5-6}\cmidrule(lr){7-8}
\cmidrule(lr){9-10}\cmidrule(lr){11-12}\cmidrule(lr){13-14}
$r$ & $\alpha$ & KS & CvM & KS & CvM & KS & CvM & KS & CvM & KS & CvM & KS & CvM \\
\midrule
\multicolumn{14}{l}{\textit{Case I: Smooth logit design}} \\
\midrule
\multirow{3}{*}{3}
& 1\% & 28.32 & 34.88 & 60.42 & 68.94 & 91.94 & 95.88 & 15.96 & 18.96 & 36.38 & 40.88 & 70.54 & 76.04 \\
& 5\% & 52.88 & 59.14 & 81.20 & 86.24 & 98.12 & 98.98 & 36.34 & 38.80 & 60.80 & 65.20 & 87.86 & 90.36 \\
& 10\% & 65.22 & 70.62 & 88.54 & 91.94 & 99.22 & 99.52 & 47.76 & 51.52 & 71.68 & 74.98 & 93.10 & 94.30 \\
\addlinespace
\multirow{3}{*}{5}
& 1\% & 32.74 & 39.72 & 64.56 & 72.90 & 93.96 & 96.78 & 18.66 & 21.16 & 38.72 & 43.68 & 73.34 & 77.96 \\
& 5\% & 56.78 & 62.86 & 84.20 & 88.42 & 98.70 & 99.38 & 38.60 & 42.00 & 63.26 & 67.52 & 89.20 & 92.10 \\
& 10\% & 69.46 & 73.86 & 90.68 & 93.68 & 99.44 & 99.70 & 50.54 & 53.94 & 73.60 & 77.46 & 93.76 & 95.54 \\
\addlinespace
\multirow{3}{*}{10}
& 1\% & 35.58 & 42.32 & 68.58 & 76.34 & 95.46 & 97.66 & 18.92 & 21.56 & 41.60 & 45.98 & 76.84 & 80.68 \\
& 5\% & 60.28 & 66.42 & 86.78 & 90.44 & 99.14 & 99.62 & 39.56 & 43.52 & 65.20 & 70.22 & 91.14 & 93.08 \\
& 10\% & 72.22 & 77.06 & 92.08 & 94.34 & 99.64 & 99.84 & 52.00 & 56.16 & 76.06 & 79.44 & 94.82 & 96.08 \\
\midrule
\multicolumn{14}{l}{\textit{Case II: Threshold-assignment design}} \\
\midrule
\multirow{3}{*}{3}
& 1\% & 17.30 & 17.54 & 46.72 & 46.62 & 87.70 & 85.86 & 10.84 & 11.08 & 29.12 & 28.26 & 67.34 & 63.42 \\
& 5\% & 36.94 & 38.88 & 70.86 & 71.08 & 96.12 & 95.72 & 26.22 & 27.50 & 52.48 & 53.44 & 86.10 & 83.90 \\
& 10\% & 49.04 & 51.58 & 80.64 & 81.42 & 98.24 & 97.86 & 37.82 & 39.72 & 64.74 & 65.68 & 92.04 & 91.46 \\
\addlinespace
\multirow{3}{*}{5}
& 1\% & 16.68 & 18.52 & 49.06 & 50.18 & 90.32 & 89.94 & 10.40 & 11.48 & 30.94 & 31.54 & 69.78 & 68.32 \\
& 5\% & 35.16 & 38.10 & 70.96 & 73.02 & 97.06 & 97.16 & 25.72 & 27.80 & 54.00 & 54.54 & 87.88 & 86.94 \\
& 10\% & 46.48 & 50.28 & 80.14 & 82.42 & 98.64 & 98.60 & 36.44 & 38.84 & 65.42 & 67.26 & 92.98 & 92.92 \\
\addlinespace
\multirow{3}{*}{10}
& 1\% & 25.90 & 28.38 & 63.72 & 66.70 & 95.40 & 96.20 & 15.46 & 16.36 & 38.66 & 38.86 & 76.82 & 76.32 \\
& 5\% & 48.98 & 52.98 & 84.10 & 86.14 & 98.84 & 99.08 & 33.96 & 36.14 & 63.18 & 64.78 & 91.62 & 91.60 \\
& 10\% & 61.96 & 65.40 & 90.48 & 92.02 & 99.46 & 99.60 & 46.32 & 48.56 & 74.46 & 76.00 & 95.08 & 95.76 \\
\midrule
\multicolumn{14}{l}{\textit{Case III: Fractional latent-index design}} \\
\midrule
\multirow{3}{*}{3}
& 1\% & 13.44 & 14.48 & 33.32 & 35.36 & 72.68 & 75.02 & 9.44 & 8.94 & 22.20 & 21.00 & 53.96 & 49.80 \\
& 5\% & 32.54 & 34.88 & 59.50 & 62.14 & 90.04 & 91.58 & 24.70 & 24.26 & 45.64 & 44.24 & 76.84 & 76.52 \\
& 10\% & 45.72 & 48.62 & 72.14 & 75.54 & 94.98 & 96.02 & 36.68 & 36.52 & 58.24 & 58.28 & 85.86 & 85.84 \\
\addlinespace
\multirow{3}{*}{5}
& 1\% & 16.72 & 17.78 & 37.42 & 39.04 & 76.22 & 76.38 & 11.14 & 10.94 & 25.32 & 22.96 & 55.98 & 51.22 \\
& 5\% & 37.76 & 40.08 & 64.04 & 66.02 & 92.20 & 93.44 & 28.10 & 28.14 & 47.92 & 47.12 & 78.30 & 76.56 \\
& 10\% & 50.98 & 53.04 & 75.54 & 77.96 & 96.42 & 97.24 & 40.48 & 40.68 & 61.22 & 60.92 & 87.02 & 86.74 \\
\addlinespace
\multirow{3}{*}{10}
& 1\% & 21.54 & 21.92 & 42.14 & 42.44 & 77.26 & 76.90 & 13.64 & 13.00 & 26.58 & 23.76 & 55.88 & 50.68 \\
& 5\% & 43.66 & 45.60 & 68.92 & 69.64 & 92.22 & 92.54 & 32.02 & 30.94 & 50.76 & 48.20 & 77.88 & 75.30 \\
& 10\% & 57.40 & 58.98 & 79.52 & 80.92 & 95.88 & 96.90 & 44.04 & 43.48 & 63.88 & 63.06 & 85.92 & 85.62 \\
\bottomrule
\end{tabular*}

\vspace{0.35em}
\begin{minipage}{0.96\linewidth}
\scriptsize
\textit{Notes.} Entries report rejection frequencies (in percent) over 5,000 Monte Carlo replications under heterogeneous alternatives. The nominal significance levels are 1\%, 5\%, and 10\%. Bootstrap critical values are based on 2,000 multiplier bootstrap draws.
\end{minipage}
\par
\endgroup
\end{landscape}
\FloatBarrier

Table \ref{tab:power} reports rejection frequencies under heterogeneous alternatives. The alternatives differ not only in the assignment and outcome equations but also in the form of treatment-effect heterogeneity: the smooth logit design uses a linear departure from homogeneity, the threshold-assignment design uses a one-sided threshold-type departure, and the fractional latent-index design uses a quadratic departure. In all cases, power increases clearly with the sample size, consistent with the fixed-alternative theory. Under homoskedastic errors, the tests have high power by \(n=2000\) in all three designs. Under heteroskedastic errors, power is lower, reflecting the reduced signal-to-noise ratio, but it still rises substantially with \(n\). Hence the tests remain informative across different forms of treatment allocation, outcome response, and treatment-effect heterogeneity. Comparing the two statistics, CvM is often slightly more powerful than KS for smooth alternatives, whereas the two statistics behave similarly in the threshold-type design. Overall, the results support the finite-sample reliability of the bootstrap approximation and the ability of the proposed tests to detect heterogeneous treatment effects across a range of allocation rules, outcome models, and heterogeneity patterns.

\section{Empirical Illustration}\label{sec:empirical}
In this section, we apply the proposed tests to examine whether the effect of maternal smoking during pregnancy on infant birth weight varies with mother's age. We use North Carolina vital-statistics records from 1988 to 2002. These data have been widely used to study treatment-effect heterogeneity in this setting; see, among others, \citet{abrevaya2015estimating}, \citet{lee2017doubly}, \citet{fan2022estimation}, and \citet{cai2025nonparametric}. Following the existing literature, we restrict the sample to first-time mothers and analyze Black and White mothers separately. We further follow \citet{cai2025nonparametric} in restricting attention to mothers aged 20--30, yielding samples of 79,412 Black mothers and 274,495 White mothers.

Low infant birth weight is associated with adverse health and human-capital outcomes throughout the life cycle \citep{black2007cradle,almond2011human}, and maternal smoking during pregnancy is widely regarded as an important preventable cause of low birth weight \citep{kramer1987intrauterine}. These concerns have motivated a growing literature on whether the effect of maternal smoking varies with mother's age. \citet{abrevaya2015estimating} and \citet{lee2017doubly} estimate age-specific CATEs, with their estimated profiles generally suggesting that the adverse effect of smoking becomes more pronounced at older ages, although the magnitude and statistical precision of the heterogeneity differ across studies. \citet{fan2022estimation} revisit the same question using high-dimensional methods for nuisance estimation and also reports evidence of age-related variation. Focusing instead on partially conditional quantile treatment effects, \citet{cai2025nonparametric} find significant age variation for White mothers but not for Black mothers. Motivated by these findings, we apply the proposed tests to assess, separately by race, whether the CATE is constant over mother's age. A rejection of the homogeneity null naturally raises the further question of whether the resulting age variation can be summarized by a parsimonious functional form. We therefore also test whether the CATE is linear in mother's age. In this way, our analysis complements the existing estimation and confidence-band evidence by providing direct tests of both CATE homogeneity and the linear specification.

Here, the outcome $Y$ is infant birth weight measured in grams, the treatment indicator $D$ equals one if the mother smoked during pregnancy, and the covariate of interest $Z$ is mother's age. Following \citet{abrevaya2015estimating}, the covariate vector $X$ contains 13 observed characteristics: mother's age and education, the month of the first prenatal visit, the number of prenatal visits, previous terminated pregnancies, and indicators for a male infant, marital status, missing father's age, gestational diabetes, hypertension, amniocentesis, ultrasound examinations, and alcohol use. We maintain unconfoundedness conditional on $X$, so that potential birth-weight outcomes are independent of maternal smoking status given these covariates; see \citet{abrevaya2015estimating} for a detailed discussion of this identifying assumption and the choice of controls.

We estimate the propensity score by series logit and the two conditional outcome regressions by sieve least squares. Our baseline nuisance basis contains an intercept, the 13 covariates, squared age, and interactions between age and the remaining covariates, for a total of 27 terms. This covariate expansion follows the empirical propensity-score implementations in \citet{abrevaya2015estimating} and \citet{cai2025nonparametric}. As a robustness check, we also use a 59-term age-mixed basis that adds interactions between the eight binary variables and the five continuous or count variables, together with additional centered-age terms. For each implementation, the same collection of basis terms is used in the propensity-score and outcome regressions.

Overlap is limited in this application. Among mothers aged 20--30, smokers account for only 8.3\% of the Black sample and 15.5\% of the White sample. Under the baseline basis, 40.3\% and 28.1\% of observations, respectively, have estimated propensity scores below 0.05. Following \citet{abrevaya2015estimating}, we therefore apply the data-driven trimming rule of \citet{crump2009dealing}. For each race and basis, we first estimate the propensity score, select the cutoff $\widehat{\alpha}$, and retain observations satisfying $\widehat{\alpha}\leq \widehat{p}(X)\leq 1-\widehat{\alpha}$. The baseline basis yields cutoffs of 0.036 for Black mothers and 0.070 for White mothers, close to those reported by \citet{abrevaya2015estimating}. We then re-estimate all nuisance functions on the retained sample before constructing the doubly robust scores and test statistics. As summarized in Table \ref{tab:empirical_sample}, the two bases retain approximately 55,000--57,000 Black mothers and 166,000--182,000 White mothers, corresponding to about 70--72\% and 61--66\% of the respective pre-trimming samples.

\begin{table}[!htbp]
\centering
\caption{Analysis samples and overlap trimming}
\label{tab:empirical_sample}
\begin{threeparttable}
\small
\begin{tabular}{llrrrr}
\toprule
Race & Nuisance basis & Analysis sample & Smokers (\%)
     & $\widehat\alpha$ & Retained sample \\
\midrule
Black & Baseline 27  & 79,412  & 8.3  & 0.036 & 56,838 \\
Black & Age-mixed 59 & 79,412  & 8.3  & 0.037 & 55,309 \\
White & Baseline 27  & 274,495 & 15.5 & 0.070 & 182,146 \\
White & Age-mixed 59 & 274,495 & 15.5 & 0.077 & 166,233 \\
\bottomrule
\end{tabular}
\begin{tablenotes}[flushleft]
\footnotesize
\item \textit{Notes.} The analysis sample contains first-time mothers aged 20--30 before overlap trimming. The smoking share is computed before trimming. The retained sample contains observations with estimated propensity scores in $[\widehat\alpha,1-\widehat\alpha]$. All nuisance functions are re-estimated after trimming.
\end{tablenotes}
\end{threeparttable}
\end{table}

Within each trimmed sample, we construct the doubly robust score and implement the KS and CvM tests using the Mammen multiplier bootstrap described in Section \ref{sec:mb}. We begin with the CATE homogeneity null \(\mathbb H_0\). Panel A of Table \ref{tab:empirical_tests} reports the bootstrap \(p\)-values. Under the baseline basis, the homogeneity test yields KS and CvM \(p\)-values of 0.392 and 0.621 for Black mothers, so the null of a constant CATE is not rejected. The conclusion is unchanged under the age-mixed basis, for which the corresponding \(p\)-values are 0.292 and 0.443. For White mothers, by contrast, both \(p\)-values equal 0.0005 under the baseline basis, and the same result is obtained under the age-mixed basis. The homogeneity null is therefore strongly rejected for White mothers under both bases. These separate-sample results should not be interpreted as a formal test that the Black and White CATE functions differ. The qualitative pattern is nevertheless consistent with \citet{cai2025nonparametric}, who find significant age variation in partially conditional quantile treatment effects for White mothers but not for Black mothers, although their estimand differs from the CATE considered here.

\begin{table}[!htbp]
\centering
\caption{Bootstrap \(p\)-values for CATE heterogeneity and linearity tests}
\label{tab:empirical_tests}
\begin{threeparttable}
\small
\renewcommand{\arraystretch}{1.08}
\begin{tabular*}{0.96\textwidth}{
    @{\extracolsep{\fill}}
    llrcccc
    @{}
}
\toprule
Race & Nuisance basis & \(n\) & KS & CvM
     & \(\widehat{\theta}_0\) & \(\widehat{\theta}_1\) \\
\midrule
\multicolumn{7}{l}{\textit{Panel A: CATE heterogeneity test}} \\
\addlinespace[2pt]
\multirow{2}{*}{Black}
  & Baseline 27  & 56,838  & 0.392  & 0.621  & -- & -- \\
  & Age-mixed 59 & 55,309  & 0.292  & 0.443  & -- & -- \\
\addlinespace[2pt]
\multirow{2}{*}{White}
  & Baseline 27  & 182,146 & 0.0005 & 0.0005 & -- & -- \\
  & Age-mixed 59 & 166,233 & 0.0005 & 0.0005 & -- & -- \\
\midrule
\multicolumn{7}{l}{
    \textit{Panel B: Linear CATE test and coefficient estimates}
} \\
\addlinespace[2pt]
\multirow{2}{*}{Black}
  & Baseline 27  & 56,838  & 0.238 & 0.382
  & -85.74 & -2.31 \\
  & Age-mixed 59 & 55,309  & 0.119 & 0.157
  & -58.71 & -3.37 \\
\addlinespace[2pt]
\multirow{2}{*}{White}
  & Baseline 27  & 182,146 & 0.088 & 0.120
  & -68.36 & -6.29 \\
  & Age-mixed 59 & 166,233 & 0.137 & 0.154
  & -70.64 & -6.05 \\
\bottomrule
\end{tabular*}

\begin{tablenotes}[flushleft]
\footnotesize
\item \textit{Notes.} Panel A reports bootstrap \(p\)-values for the
CATE homogeneity null, \(\tau(z)=\tau_0\). Panel B reports coefficient
estimates and bootstrap \(p\)-values for the linear CATE restriction,
\(\tau(z)=\theta_0+\theta_1 z\). The \(p\)-values are based on 2,000
Mammen multiplier bootstrap draws; 0.0005 is the smallest attainable
value.
\end{tablenotes}
\end{threeparttable}
\end{table}

The rejection of CATE homogeneity for White mothers motivates a more structured question: whether the age variation can be summarized by a linear function. We therefore apply the projected ICM test developed in Section \ref{sec:ex-para} to $\mathbb H_0^{\dagger}:\tau(z)=\theta_0+\theta_1 z$. Panel B of Table \ref{tab:empirical_tests} reports the corresponding coefficient estimates and bootstrap \(p\)-values. For White mothers, the baseline KS and CvM \(p\)-values are 0.088 and 0.120, respectively, while the age-mixed basis yields \(p\)-values of 0.137 and 0.154. Thus, the linear CATE restriction is not rejected at the 5\% level under either basis, although the baseline KS result is marginal at the 10\% level. Together with the strong rejection of homogeneity, these results are consistent with an approximately linear age profile over the retained age range. The estimated slopes are \(-6.29\) and \(-6.05\) grams per year under the baseline and age-mixed bases, respectively, implying that the estimated adverse effect of smoking becomes about 6 grams larger in magnitude for each additional year of mother's age. For completeness, the linear restriction is also not rejected for Black mothers, for whom the homogeneity null was already not rejected.

\begin{figure}[H]
    \centering
    \begin{minipage}{0.88\textwidth}
        \centering

        \begin{subfigure}[t]{0.485\linewidth}
            \centering
            \includegraphics[width=\linewidth]{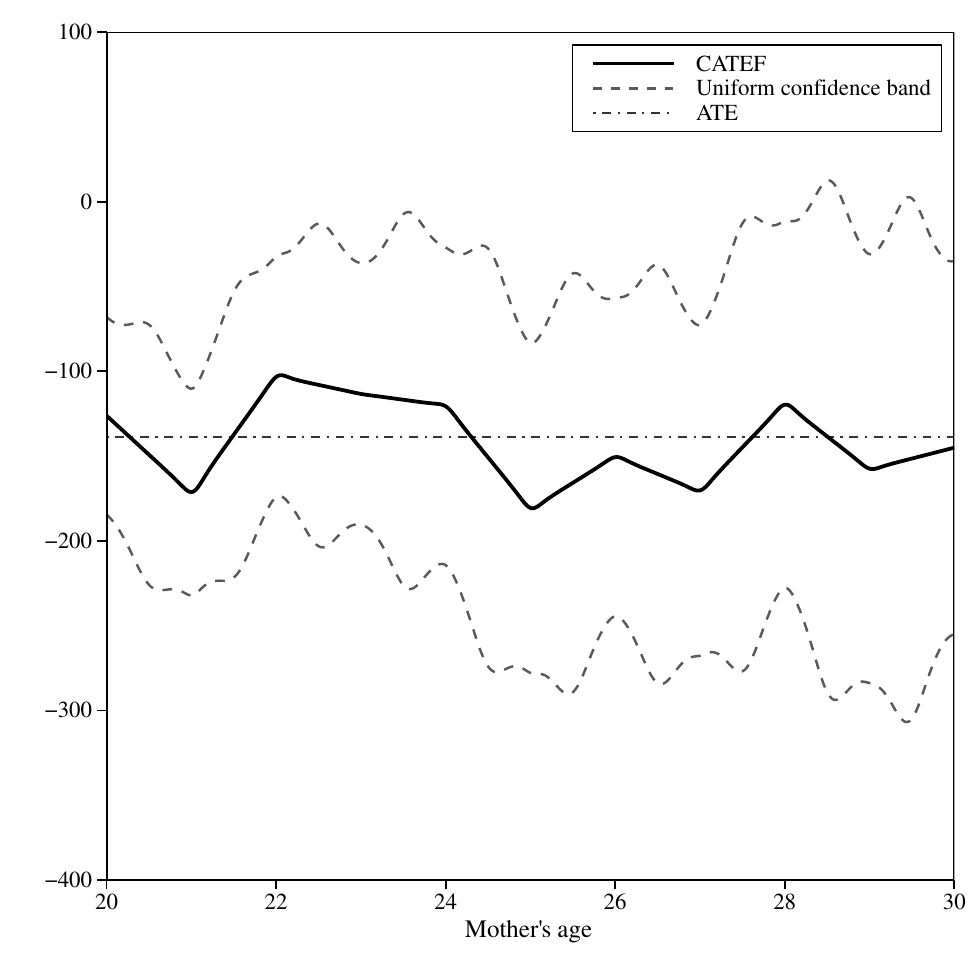}
            \caption{Black mothers}
            \label{fig:catef_black}
        \end{subfigure}
        \hfill
        \begin{subfigure}[t]{0.485\linewidth}
            \centering
            \includegraphics[width=\linewidth]{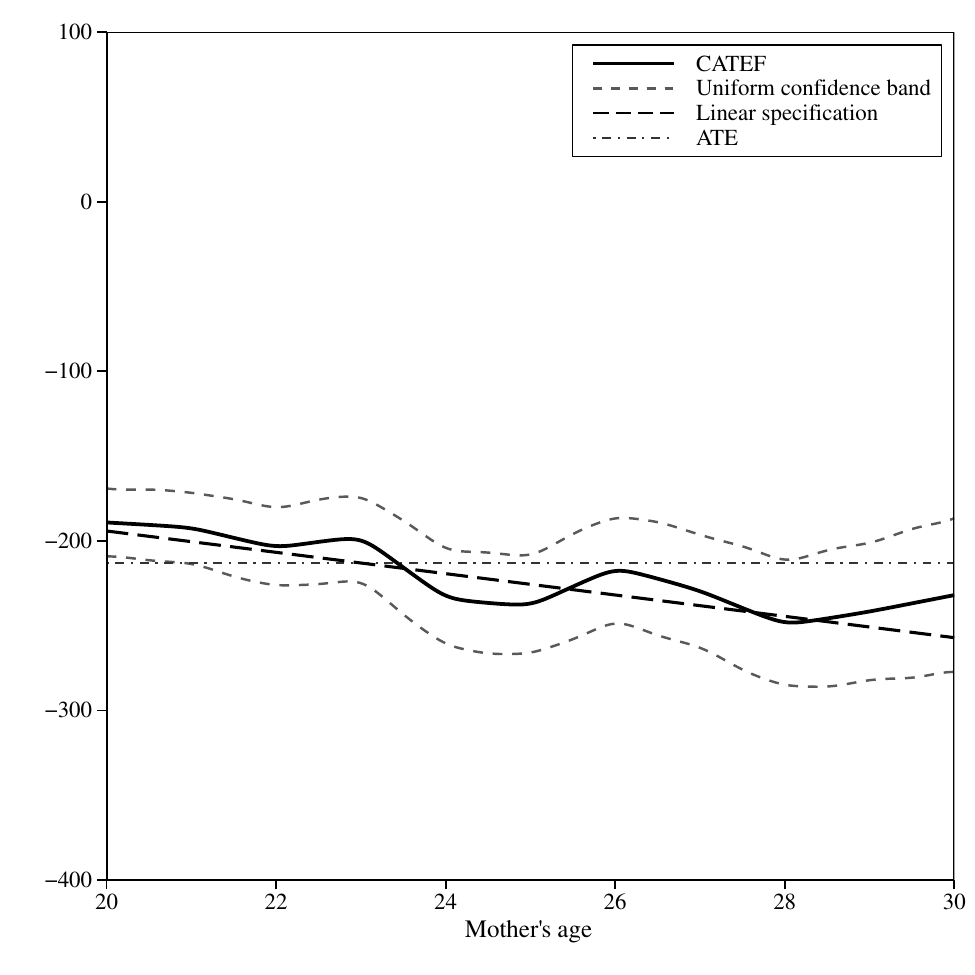}
            \caption{White mothers}
            \label{fig:catef_white}
        \end{subfigure}

    \end{minipage}
    \caption{Doubly robust local-linear CATEF estimates of the effect of maternal smoking on birth weight, with 95\% uniform confidence bands following \citet{lee2017doubly}. The horizontal line in each panel denotes the estimated ATE; the additional line for White mothers is the fitted linear CATE profile from the projected ICM procedure, corresponding to a specification not rejected at the 5\% level.}
    \label{fig:lee_catef_ucb_race}
\end{figure}

To complement the formal tests, Figure \ref{fig:lee_catef_ucb_race} plots doubly robust local-linear estimates of the CATEF, together with 95\% uniform confidence bands constructed using the method proposed by \citet{lee2017doubly}. The estimates use the baseline nuisance basis and the same overlap-trimmed samples as the formal tests. For Black mothers, the estimated CATEF fluctuates around the ATE without a clear systematic age profile, and the horizontal ATE line remains within the uniform confidence band over the displayed age range. For White mothers, the estimated effect becomes more negative with age, while the horizontal ATE line does not capture this movement and lies outside the uniform confidence band over parts of the age range. The additional line in the White-mother panel is the fitted linear CATE profile obtained from our projected ICM procedure. It corresponds to the linear restriction that is not rejected at the 5\% level and closely tracks the overall downward pattern of the nonparametric estimate. These graphical patterns are consistent with the non-rejection of CATE homogeneity for Black mothers and, for White mothers, the rejection of homogeneity together with the non-rejection of the linear restriction.

\section{Conclusion}\label{sec:con}

This paper develops orthogonal integrated conditional moment tests for CATE heterogeneity with respect to covariates of interest, while allowing a richer set of controls to be used for treatment-selection adjustment. Based on an augmented inverse-probability-weighted score, we construct Kolmogorov--Smirnov and Cramér--von Mises tests. Our main technical result establishes a uniform feasible-to-oracle approximation for the indexed process when the propensity score and outcome regressions are estimated by series logit and sieve least squares. This approximation yields the null asymptotic distributions, consistency against fixed alternatives, nontrivial power against \(n^{-1/2}\)-local alternatives, and a computationally simple multiplier bootstrap. The framework also accommodates projected specification tests of prespecified parametric CATE restrictions and conditional-LATE heterogeneity.

The simulations show accurate size and increasing power against heterogeneous alternatives. In the maternal-smoking application, CATE homogeneity with respect to mother's age is not rejected for Black mothers but is strongly rejected for White mothers. For White mothers, a linear age profile is not rejected at the 5\% level. These findings illustrate how direct specification tests can complement CATE estimation and visualization by assessing whether apparent heterogeneity is statistically significant and whether it admits a parsimonious representation.
\putbib 

\end{bibunit}
\clearpage

\begin{bibunit}
\appendix
\setstretch{1.0}
\renewcommand{\thesection}{ \Alph{section}}
\renewcommand{\thesubsection}{\thesection.\arabic{subsection}}
\counterwithin{equation}{section}
\counterwithin{figure}{section}
\counterwithin{table}{section}
\counterwithin{lemma}{section}

\setcounter{page}{1}
\pagenumbering{arabic}

\begin{center}
    {\LARGE Online Appendix for}\\[1ex]
    {\Large ``Orthogonal Integrated Conditional Moment Tests for Treatment Effect Heterogeneity''}\\[3ex]
    Haokun Lu \quad Xiaojun Song\\[1ex]
\end{center}

The Online Appendix is organized into three sections. Section \ref{sec:proofs} provides the proofs of all main results. Section \ref{sec:app_add_HIR} establishes the uniform IPW linearization result discussed in Remark \ref{rem:orthogonal-ipw}, which connects the plug-in IPW process with the doubly robust oracle process, and Section \ref{sec:app-parametric-sim} reports additional Monte Carlo results for the parametric-form extensions discussed in Section \ref{sec:ex-para}.

\vspace{1em}
\section{Proofs}
\label{sec:proofs}
\subsection{Proofs for Section \ref{sec:cate}}
\paragraph{Proof of Proposition \ref{prop:dr-gateaux}}
\begin{proof}
Fix \(z\in\mathcal Z\) and write \(I_z:=\1\{Z\le z\}\). Let
\[
\Delta_q(X):=q(X)-p(X),\qquad
\Delta_1(X):=m_1(X)-\mu_1(X),\qquad
\Delta_0(X):=m_0(X)-\mu_0(X).
\]
By the domination condition in the proposition, differentiation and expectation can be interchanged. Since \(\tau\) is fixed in the derivative with respect to \(r\),
\[
\left.
\frac{\partial}{\partial r}
M(z;\eta^r,\tau)
\right|_{r=0}
=
\E\!\left[
\left.
\frac{\partial}{\partial r}\psi(W;\eta^r)
\right|_{r=0}
I_z
\right].
\]
A direct differentiation along the joint path \(\eta^r=(q^r,m_1^r,m_0^r)\) gives
\[
\begin{aligned}
\left.
\frac{\partial}{\partial r}\psi(W;\eta^r)
\right|_{r=0}
&=
-\left\{
\frac{\Delta_1(X)}{p(X)}
+
\frac{\Delta_0(X)}{1-p(X)}
\right\}\{D-p(X)\} \\
&\quad
-\Delta_q(X)
\left[
\frac{D\{Y-\mu_1(X)\}}{p(X)^2}
+
\frac{(1-D)\{Y-\mu_0(X)\}}{\{1-p(X)\}^2}
\right].
\end{aligned}
\]
Therefore,
\[
\begin{aligned}
\left.
\frac{\partial}{\partial r}
M(z;\eta^r,\tau)
\right|_{r=0}
&=
-\E\!\left[
\left\{
\frac{\Delta_1(X)}{p(X)}
+
\frac{\Delta_0(X)}{1-p(X)}
\right\}\{D-p(X)\}I_z
\right] \\
&\quad
-\E\!\left[
\Delta_q(X)
\left\{
\frac{D\{Y-\mu_1(X)\}}{p(X)^2}
+
\frac{(1-D)\{Y-\mu_0(X)\}}{\{1-p(X)\}^2}
\right\}I_z
\right].
\end{aligned}
\]
Since \(Z\) is a subvector of \(X\), \(I_z\) is \(X\)-measurable. The first expectation is zero because
\[
\E[D-p(X)\mid X]=0.
\]
For the second expectation, by the definition of \(\mu_d(X)=\E[Y\mid D=d,X]\),
\[
\E[D\{Y-\mu_1(X)\}\mid X]
=
p(X)\E[Y-\mu_1(X)\mid D=1,X]
=0,
\]
and
\[
\E[(1-D)\{Y-\mu_0(X)\}\mid X]
=
\{1-p(X)\}\E[Y-\mu_0(X)\mid D=0,X]
=0.
\]
Because \(\Delta_q(X)I_z\) is \(X\)-measurable, the second expectation is also zero. Hence
\[
\left.
\frac{\partial}{\partial r}
M(z;\eta^r,\tau)
\right|_{r=0}
=0.
\]

Finally, for the scalar parameter \(\tau_0\), fix any nuisance value \(\eta\). Since \(\psi(W;\eta)\) does not depend on \(\tau\),
\[
\begin{aligned}
\left.
\frac{d}{dt}
M(z;\eta,\tau_0+t h_\tau)
\right|_{t=0}
&=
\left.
\frac{d}{dt}
\E[
\{\psi(W;\eta)-\tau_0-t h_\tau\}I_z]
\right|_{t=0} \\
&=
-h_\tau\E[I_z]
=
-h_\tau F_Z(z).
\end{aligned}
\]
\end{proof}

\subsection{Proofs for Section \ref{sec:asy}}

We next introduce the series notation used throughout the appendix. Following
\citet{hirano2003efficient} (hereafter HIR), we work with a power-series basis.
Let $\varphi=(\varphi_1,\ldots,\varphi_{r})'\in\mathbb Z_+^{r}$ be a
$r$-dimensional vector of nonnegative integers, and let
$|\varphi|=\sum_{j=1}^{r}\varphi_j$. Let
$\{\varphi(k)\}_{k=1}^\infty$ be an ordering of all distinct
$\varphi\in\mathbb Z_+^{r}$ such that $|\varphi(k)|$ is nondecreasing in
$k$. For $x^\varphi=\prod_{j=1}^{r}x_j^{\varphi_j}$, define the raw
power-series vector
$$
r^K(x):=\bigl(x^{\varphi(1)},\ldots,x^{\varphi(K)}\bigr)'.
$$
Equivalently, the sequence of powers begins with
$$
1,\ x_1,\ldots,x_{r},\ x_1^2,\ x_1x_2,\ldots,\ x_{r}^2,\ldots .
$$

As in HIR, we use a normalized version of this basis. Let $A_K$ be a
nonsingular $K\times K$ matrix and define
$$
R^K(x):=A_K r^K(x),
\qquad
\E[R^K(X)R^K(X)']=I_K.
$$
Since $A_K$ is nonsingular, $r^K$ and $R^K$ generate the same sieve space.
Finally, define
$$
\zeta(K):=\sup_{x\in\mathcal X}\|R^K(x)\|.
$$
For orthonormal polynomial bases, \cite{newey1997convergence} show that $\zeta(K)\lesssim K$.

Throughout the appendix, let
$\|f\|_{2,n}:=\{n^{-1}\sum_{i=1}^n f(X_i)^2\}^{1/2}$ denote the empirical
$L_2$ norm. For two nonnegative sequences $a_n$ and $b_n$, we write
$a_n\lesssim b_n$ if there exists a constant $C<\infty$, independent of $n$,
such that $a_n\le Cb_n$ for all sufficiently large $n$. We write
$a_n\asymp b_n$ if $a_n\lesssim b_n$ and $b_n\lesssim a_n$. In what follows,
$C$ denotes a generic positive constant that may vary from line to line. We next introduce several useful lemmas.

\begin{lemma}[Centered lower-orthant class]
\label{lem:centered-lower-orthant-vc}
Let $Z\in\mathbb R^{d_z}$, and for $z\in\mathcal Z\subset\mathbb R^{d_z}$, write
$$
F_Z(z):=\mathbb P(Z\le z)
=
\mathbb P(Z_1\le z_1,\ldots,Z_{d_z}\le z_{d_z}),
$$
where the inequality is understood componentwise. Define
$$
\mathcal H:=\{h_z:z\in\mathcal Z\},
\qquad
h_z(w):=\1\{z_w\le z\}-F_Z(z),
$$
where $w=(y,d,x,z_w)$. Then $\mathcal H$ is a VC-subgraph class with envelope $1$. Consequently, if $\chi(W)$ is measurable and $\E[\chi(W)^2]<\infty$, then
$$
\chi\mathcal H:=\{\chi(w)h_z(w):z\in\mathcal Z\}
$$
is $P$-Donsker with envelope $|\chi|$.
\end{lemma}

\begin{proof}
Let $A_z:=\{w:z_w\le z\}$. Since $Z\in\mathbb R^{d_z}$, the class
$$
\mathcal A:=\{A_z:z\in\mathcal Z\}
$$
is the class of lower orthants in $\mathbb R^{d_z}$, and hence is VC. For each $z\in\mathcal Z$, the subgraph of $h_z$ is
$$
\{(w,t):t<h_z(w)\}
=
\bigl(A_z\times(-\infty,1-F_Z(z))\bigr)
\cup
\bigl(A_z^c\times(-\infty,-F_Z(z))\bigr).
$$
The class of lower half-lines on $\mathbb R$ is VC, and VC classes are stable under complements, Cartesian products, and finite unions. Hence $\mathcal H$ is VC-subgraph. Since $|h_z|\le1$, $1$ is an envelope for $\mathcal H$.

For the weighted class, a standard change-of-measure argument shows that multiplying a uniformly bounded VC-subgraph class by a fixed measurable function $\chi$ preserves the VC-type entropy bound, with envelope $|\chi|$. Since $P\chi^2<\infty$, the standard Donsker theorem for VC-type classes with square-integrable envelopes implies that $\chi\mathcal H$ is $P$-Donsker.
\end{proof}

\begin{lemma}[Rates for the series logit propensity score]
\label{lem:p-rate}
Suppose that the conditions of Lemmas 1 and 2 in Appendix A of
\citet{hirano2003efficient} hold for the series logit estimator. In particular,
the support of $X$ is compact, the propensity score $p(x)$ is $s$-times
continuously differentiable and bounded away from zero and one, the density of
$X$ is bounded away from zero on its support, and the series dimension satisfies
$K\to\infty$ and $K^4/n\to0$. Then
\begin{equation}
\label{eq:p-sup-rate-summary}
\sup_{x\in\mathcal X}|\hat p(x)-p(x)|
=
O_p\!\left(\sqrt{\frac{K^3}{n}}\right)
+
O\!\left(K^{1-s/(2r)}\right).
\end{equation}
Moreover,
\begin{equation}
\label{eq:p-l2-rate-summary}
\|\hat p-p\|_{2,n}
=
O_p\!\left(\sqrt{\frac{K}{n}}\right)
+
O_p\!\left(K^{-s/(2r)}\right).
\end{equation}
\end{lemma}

\begin{proof}
The maintained conditions imply the assumptions of Lemmas 1 and 2 in Appendix A of \citet{hirano2003efficient}. Hence we may use their
approximation and convergence results for the series logit estimator. 
Let $\hat p(x)=\Lambda\{R^K(x)'\hat\pi_K\}$, where
$\Lambda(a)=\exp(a)/(1+\exp(a))$. We follow the notation in Appendix A of HIR.
Since the log-odds ratio $\log\{p(x)/(1-p(x))\}$ is $s$-times continuously
differentiable on $\mathcal X$, the power-series approximation argument in HIR
(see equation (24) in their Appendix A) implies that there exists a coefficient
vector $\pi_K$ such that
$$
\sup_{x\in\mathcal X}
\left|
\log\frac{p(x)}{1-p(x)}
-
R^K(x)'\pi_K
\right|
=
O(K^{-s/r}).
$$
This $\pi_K$ is the deterministic approximation coefficient for the population
log-odds ratio.

Let $\pi_K^*$ denote the population pseudo-true logit coefficient,
\begin{equation}
    \label{eq:piKstar}
    \pi_K^*
:=
\arg\max_{\pi\in\mathbb R^K}
\E\left[
p(X)\log\Lambda\{R^K(X)'\pi\}
+
\{1-p(X)\}\log\{1-\Lambda(R^K(X)'\pi)\}
\right],
\end{equation}
and define the corresponding pseudo-true propensity score as
$$
p_K^*(x):=\Lambda\{R^K(x)'\pi_K^*\}.
$$

The sup-norm bound in \eqref{eq:p-sup-rate-summary} follows directly from
Appendix A of HIR. In particular, their Lemma 1 gives
$\|\pi_K-\pi_K^*\|=O(K^{-s/(2r)})$ and
$\sup_{x\in\mathcal X}|p(x)-p_K^*(x)|
=O(\zeta(K)K^{-s/(2r)})$, while their Lemma 2 gives
$\|\hat\pi_K-\pi_K^*\|=O_p(\sqrt{K/n})$. Therefore, by the boundedness of
$\Lambda'$ and $\zeta(K)\lesssim K$,
$$
\begin{aligned}
\sup_{x\in\mathcal X}|\hat p(x)-p(x)|
&\le
C\zeta(K)\|\hat\pi_K-\pi_K^*\|
+
\sup_{x\in\mathcal X}|p_K^*(x)-p(x)|  \\
&=
O_p\!\left(\sqrt{\frac{K^3}{n}}\right)
+
O\!\left(K^{1-s/(2r)}\right),
\end{aligned}
$$
which proves \eqref{eq:p-sup-rate-summary}.

It remains to prove the empirical $L_2$ bound. By the triangle inequality,
$$
\|\hat p-p\|_{2,n}
\le
\|\hat p-p_K^*\|_{2,n}
+
\|p_K^*-p\|_{2,n}.
$$
For the stochastic component, the mean value theorem and the boundedness of
$\Lambda'$ imply
$|\hat p(X_i)-p_K^*(X_i)|
\le C|R^K(X_i)'(\hat\pi_K-\pi_K^*)|$. Hence
$$
\begin{aligned}
\|\hat p-p_K^*\|_{2,n}^2
&\le
C(\hat\pi_K-\pi_K^*)'
\left\{
\frac1n\sum_{i=1}^n R^K(X_i)R^K(X_i)'
\right\}
(\hat\pi_K-\pi_K^*) \\
&=
C(\hat\pi_K-\pi_K^*)'\hat S_K(\hat\pi_K-\pi_K^*),
\end{aligned}
$$
where $\hat S_K:=n^{-1}\sum_{i=1}^n R^K(X_i)R^K(X_i)'$. By equation (33) in Appendix A of HIR, $\|\hat S_K-I_K\|=O_p(\zeta (K)\sqrt{K/n})=o_p(1)$ under the maintained Assumption \ref{ass:SLE} on $K$. Hence $\lambda_{\max}(\hat S_K)=O_p(1)$, and therefore
$$
\|\hat p-p_K^*\|_{2,n}
=
O_p(\|\hat\pi_K-\pi_K^*\|)
=
O_p\!\left(\sqrt{\frac{K}{n}}\right).
$$

For the approximation component, again by the triangle inequality,
$$
\|p_K^*-p\|_{2,n}
\le
\|p_K^*-\Lambda(R^K(\cdot)'\pi_K)\|_{2,n}
+
\|\Lambda(R^K(\cdot)'\pi_K)-p\|_{2,n}.
$$
For the first term, the mean value theorem gives
$$
\begin{aligned}
\|p_K^*-\Lambda(R^K(\cdot)'\pi_K)\|_{2,n}^2
&\le
C(\pi_K^*-\pi_K)'\hat S_K(\pi_K^*-\pi_K) \\
&=
O_p(\|\pi_K^*-\pi_K\|^2)
=
O_p\!\left(K^{-s/r}\right).
\end{aligned}
$$
 For the second term, the log-odds approximation and the
boundedness of $\Lambda'$ imply
$$
\|\Lambda(R^K(\cdot)'\pi_K)-p\|_{2,n}
\le
\sup_{x\in\mathcal X}
|\Lambda(R^K(x)'\pi_K)-p(x)|
=
O(K^{-s/r}).
$$
Therefore,
$\|p_K^*-p\|_{2,n}=O_p(K^{-s/(2r)})$. Combining the stochastic and approximation
components yields
$$
\|\hat p-p\|_{2,n}
=
O_p\!\left(\sqrt{\frac{K}{n}}\right)
+
O_p\!\left(K^{-s/(2r)}\right),
$$
which proves \eqref{eq:p-l2-rate-summary}.
\end{proof}
\begin{lemma}[Rates for the sieve outcome regressions]
\label{lem:mu-rate}
Suppose that the conditions of Lemma \ref{lem:p-rate} hold. Suppose further that, for $d=0,1$, $\mu_d(x):=\E[Y(d)\mid X=x]$ is $s_\mu$-times continuously differentiable on $\mathcal X$, and
$\sup_{x\in\mathcal X}\E[Y(d)^2\mid X=x]<\infty$. Then
\begin{equation}
\label{eq:mu-rate-summary}
\|\hat\mu_d-\mu_d\|_{2,n}
=
O_p\!\left(\sqrt{\frac{K}{n}}\right)
+
O_p\!\left(K^{-s/(2r)}\right)
+
O\!\left(K^{-s_\mu/r}\right),
\qquad d=0,1.
\end{equation}
\end{lemma}

\begin{proof}
We prove the result for $\hat\mu_1$; the proof for $\hat\mu_0$ is identical. The argument follows the same series-regression consistency argument used in the proof of Theorem 2 of \citet{hirano2003efficient}, but we work with an empirical $L_2$ norm rather than a sup norm. Define
$$
\hat S_K:=\frac1n\sum_{i=1}^n R^K(X_i)R^K(X_i)',
\qquad
\widehat\Phi_K:=\frac1n\sum_{i=1}^n \frac{D_iY_i}{\hat p(X_i)}R^K(X_i),
\qquad
\Phi_K:=\frac1n\sum_{i=1}^n \frac{D_iY_i}{p(X_i)}R^K(X_i).
$$
Then
$$
\hat\mu_1(x)=\widehat\Phi_K'\hat S_K^{-1}R^K(x),
\qquad
\tilde\mu_1(x):=\Phi_K'\hat S_K^{-1}R^K(x).
$$
By the triangle inequality,
\begin{equation}
\label{eq:mu-l2-decomp}
\|\hat\mu_1-\mu_1\|_{2,n}
\le
\|\hat\mu_1-\tilde\mu_1\|_{2,n}
+
\|\tilde\mu_1-\mu_1\|_{2,n}.
\end{equation}

We first consider the ideal term. Since $\E[D_iY_i/p(X_i)\mid X_i=x]=\mu_1(x)$, $\tilde\mu_1$ is the fitted value in a least-squares series regression of $D_iY_i/p(X_i)$ on $R^K(X_i)$. By the standard empirical $L_2$ rate for series least squares,
\begin{equation}
\label{eq:ideal-mu-l2-rate}
\|\tilde\mu_1-\mu_1\|_{2,n}
=
O_p\!\left(\sqrt{\frac{K}{n}}\right)
+
O\!\left(K^{-s_\mu/r}\right).
\end{equation}
It remains to control the feasible-minus-ideal term. Let $\Delta_K:=\widehat\Phi_K-\Phi_K$. Then
$$
\hat\mu_1(x)-\tilde\mu_1(x)
=
\Delta_K'\hat S_K^{-1}R^K(x).
$$
Therefore,
$$
\begin{aligned}
\|\hat\mu_1-\tilde\mu_1\|_{2,n}^2
&=
\frac1n\sum_{i=1}^n
\{\Delta_K'\hat S_K^{-1}R^K(X_i)\}^2  \\
&=
\Delta_K'\hat S_K^{-1}
\left\{\frac1n\sum_{i=1}^nR^K(X_i)R^K(X_i)'\right\}
\hat S_K^{-1}\Delta_K  \\
&=
\Delta_K'\hat S_K^{-1}\Delta_K.
\end{aligned}
$$
Since $\|\hat S_K-I_K\|=o_p(1)$, $\|\hat S_K^{-1}\|=O_p(1)$, and
\begin{equation}
\label{eq:mu-fi-l2-delta}
\|\hat\mu_1-\tilde\mu_1\|_{2,n}
=
O_p(\|\Delta_K\|).
\end{equation}

Now
$$
\Delta_K
=
\frac1n\sum_{i=1}^n
D_iY_i
\left\{
\frac1{\hat p(X_i)}-\frac1{p(X_i)}
\right\}R^K(X_i).
$$
Under overlap and Lemma \ref{lem:p-rate}, $p(X_i)$ and $\hat p(X_i)$ are bounded away from zero uniformly in $i$ with probability approaching one. Hence
$|1/\hat p(X_i)-1/p(X_i)|\le C|\hat p(X_i)-p(X_i)|$. For any vector $b\in\mathbb R^K$ with $\|b\|=1$,
$$
\begin{aligned}
|b'\Delta_K|
&\le
C\frac1n\sum_{i=1}^n
D_i|Y_i|\,|\hat p(X_i)-p(X_i)|\,|b'R^K(X_i)| \\
&\le
C
\left\{
\frac1n\sum_{i=1}^n
D_iY_i^2|\hat p(X_i)-p(X_i)|^2
\right\}^{1/2}
\left\{
\frac1n\sum_{i=1}^n
(b'R^K(X_i))^2
\right\}^{1/2}.
\end{aligned}
$$
The second factor is bounded uniformly over $\|b\|=1$ because
$$
\sup_{\|b\|=1}\frac1n\sum_{i=1}^n(b'R^K(X_i))^2
=
\lambda_{\max}(\hat S_K)
=
O_p(1).
$$
For the first factor, since $\hat p$ is estimated from $(D_i,X_i)_{i=1}^n$ and $\sup_{x\in\mathcal X}\E[Y^2\mid D=1,X=x]<\infty$, conditional Markov's inequality gives
$$
\frac1n\sum_{i=1}^n
D_iY_i^2|\hat p(X_i)-p(X_i)|^2
=
O_p\!\left(
\frac1n\sum_{i=1}^n|\hat p(X_i)-p(X_i)|^2
\right).
$$
It follows that
\begin{equation}
\label{eq:DeltaK-l2-rate}
\|\Delta_K\|
=
O_p\!\left(\|\hat p-p\|_{2,n}\right).
\end{equation}
Combining \eqref{eq:mu-fi-l2-delta} and \eqref{eq:DeltaK-l2-rate} yields
$
\|\hat\mu_1-\tilde\mu_1\|_{2,n}
=
O_p\!\left(\|\hat p-p\|_{2,n}\right).
$
By Lemma \ref{lem:p-rate},
$$
\|\hat p-p\|_{2,n}
=
O_p\!\left(\sqrt{\frac{K}{n}}\right)
+
O_p\!\left(K^{-s/(2r)}\right).
$$
Therefore,
\begin{equation}
\label{eq:mu-fi-l2-rate-final}
\|\hat\mu_1-\tilde\mu_1\|_{2,n}
=
O_p\!\left(\sqrt{\frac{K}{n}}\right)
+
O_p\!\left(K^{-s/(2r)}\right).
\end{equation}

Finally, combining \eqref{eq:mu-l2-decomp}, \eqref{eq:ideal-mu-l2-rate}, and \eqref{eq:mu-fi-l2-rate-final} gives
$$
\|\hat\mu_1-\mu_1\|_{2,n}
=
O_p\!\left(\sqrt{\frac{K}{n}}\right)
+
O_p\!\left(K^{-s/(2r)}\right)
+
O\!\left(K^{-s_\mu/r}\right).
$$
The same argument applies to $\hat\mu_0$, replacing $D_iY_i/\hat p(X_i)$ by $(1-D_i)Y_i/\{1-\hat p(X_i)\}$ and using the overlap bound for $1-p(X_i)$. This proves the lemma.
\end{proof}

\begin{lemma}
\label{lem:An111a-control}
Suppose that the conditions of Lemmas \ref{lem:p-rate} and \ref{lem:mu-rate} hold. Define
$$
r_{p,n}:=\sqrt{\frac K n}+K^{-s/(2r)},
\qquad
r_{\mu,n}:=\sqrt{\frac K n}+K^{-s/(2r)}+K^{-s_\mu/r}.
$$
Suppose further that $\sqrt n\,r_{\mu,n}r_{p,n}\to0$, $
\sqrt K\,r_{\mu,n}\to0$. Then
$$
\sup_{z\in\mathcal Z}
\left|
\frac1{\sqrt n}\sum_{i=1}^n
\frac{\hat\mu_1(X_i)-\mu_1(X_i)}{\hat p(X_i)}
\{D_i-p(X_i)\}
\{\1\{Z_i\le z\}-F_Z(z)\}
\right|
=o_p(1).
$$
The same conclusion holds with $\mu_0$ and $1-\hat p$ in place of $\mu_1$ and $\hat p$.
\end{lemma}

\begin{proof}
We prove the result for the treated outcome regression. Write
$p_i:=p(X_i)$, $\hat p_i:=\hat p(X_i)$, and $I_i(z):=\1\{Z_i\le z\}$. Let
$$
A_{n,111}^a(z)
:=
\frac1{\sqrt n}\sum_{i=1}^n
\frac{\hat\mu_1(X_i)-\mu_1(X_i)}{\hat p_i}
(D_i-p_i)\{I_i(z)-F_Z(z)\}.
$$
By overlap and Lemma \ref{lem:p-rate}, $\hat p_i$ and $p_i$ are bounded away from zero uniformly in $i$ with probability approaching one. On this event,
$
1/{\hat p_i}
=
1/{p_i}
+
\left(1/{\hat p_i}-1/{p_i}\right).
$
Accordingly,
$$
A_{n,111}^a(z)=A_{n,111,0}^a(z)+A_{n,111,p}^a(z),
$$
where
$$
A_{n,111,0}^a(z)
:=
\frac1{\sqrt n}\sum_{i=1}^n
\frac{\hat\mu_1(X_i)-\mu_1(X_i)}{p_i}
(D_i-p_i)\{I_i(z)-F_Z(z)\},
$$
and
$$
A_{n,111,p}^a(z)
:=
\frac1{\sqrt n}\sum_{i=1}^n
\{\hat\mu_1(X_i)-\mu_1(X_i)\}
\left(\frac1{\hat p_i}-\frac1{p_i}\right)
(D_i-p_i)\{I_i(z)-F_Z(z)\}.
$$

We first control the denominator-replacement remainder. Since
$|\hat p_i^{-1}-p_i^{-1}|\le C|\hat p_i-p_i|$, $|D_i-p_i|\le1$, and
$|I_i(z)-F_Z(z)|\le1$,
$$
\begin{aligned}
\sup_{z\in\mathcal Z}|A_{n,111,p}^a(z)|
&\le
\frac C{\sqrt n}\sum_{i=1}^n
|\hat\mu_1(X_i)-\mu_1(X_i)|\,|\hat p_i-p_i|  \\
&\le
C\sqrt n\,
\|\hat\mu_1-\mu_1\|_{2,n}\|\hat p-p\|_{2,n}.
\end{aligned}
$$
By Lemmas \ref{lem:p-rate} and \ref{lem:mu-rate}, the last display is
$O_p(\sqrt n\,r_{\mu,n}r_{p,n})=o_p(1)$.

It remains to control $A_{n,111,0}^a(z)$. Let $\mu_{1K}$ be a sieve approximation to $\mu_1$ in the span of $R^K$ such that
$$
\|\mu_{1K}-\mu_1\|_{2,n}=O_p(K^{-s_\mu/r})
$$
and, equivalently in population $L_2$ norm, $\|\mu_{1K}-\mu_1\|_2=O(K^{-s_\mu/r})$. Decompose
$$
A_{n,111,0}^a(z)
=
A_{n,111,s}^a(z)+A_{n,111,b}^a(z),
$$
where
$$
A_{n,111,s}^a(z)
:=
\frac1{\sqrt n}\sum_{i=1}^n
\frac{\hat\mu_1(X_i)-\mu_{1K}(X_i)}{p_i}
(D_i-p_i)\{I_i(z)-F_Z(z)\},
$$
and
$$
A_{n,111,b}^a(z)
:=
\frac1{\sqrt n}\sum_{i=1}^n
\frac{\mu_{1K}(X_i)-\mu_1(X_i)}{p_i}
(D_i-p_i)\{I_i(z)-F_Z(z)\}.
$$

We first consider the stochastic sieve component. Since both $\hat\mu_1$ and $\mu_{1K}$ belong to the span of $R^K$, there exists a random vector $\hat\gamma_K\in\mathbb R^K$ such that
$$
\hat\mu_1(x)-\mu_{1K}(x)=R^K(x)'\hat\gamma_K.
$$
Therefore,
$
A_{n,111,s}^a(z)=\hat\gamma_K'V_n(z),
$
where
$$
V_n(z)
:=
\frac1{\sqrt n}\sum_{i=1}^n
R^K(X_i)\frac{D_i-p_i}{p_i}\{I_i(z)-F_Z(z)\}.
$$

We claim that $\sup_{z\in\mathcal Z}\|V_n(z)\|=O_p(\sqrt K)$. To see this, let $V_{n,j}(z)$ denote the $j$-th coordinate of $V_n(z)$. 
For each $j$, Lemma \ref{lem:centered-lower-orthant-vc} implies that the class indexing
$V_{n,j}(z)$ is centered and of VC type with envelope
$$
\left|R_j^K(X)\frac{D-p(X)}{p(X)}\right|.
$$
Hence, by Pollard's maximal inequality (see Theorem 2.14.1 in \cite{van1996weak}),
$$
\E\left[\sup_{z\in\mathcal Z}|V_{n,j}(z)|^2\right]
\lesssim
\E\left[
\left\{R_j^K(X)\frac{D-p(X)}{p(X)}\right\}^2
\right].
$$
By overlap, the right-hand side is bounded by $C\E[R_j^K(X)^2]$. Hence
$$
\E\left[\sup_{z\in\mathcal Z}\|V_n(z)\|^2\right]
\le
\sum_{j=1}^K
\E\left[\sup_{z\in\mathcal Z}|V_{n,j}(z)|^2\right]
\le
C\sum_{j=1}^K\E[R_j^K(X)^2].
$$
Under the HIR normalization, $\E[R^K(X)R^K(X)']=I_K$, so
$\sum_{j=1}^K\E[R_j^K(X)^2]=K$. Thus
$
\sup_{z\in\mathcal Z}\|V_n(z)\|=O_p(\sqrt K).
$

Next, note that
$$
\|\hat\mu_1-\mu_{1K}\|_{2,n}^2
=
\hat\gamma_K'\hat S_K\hat\gamma_K,
\qquad
\hat S_K:=\frac1n\sum_{i=1}^nR^K(X_i)R^K(X_i)'.
$$
Since $\lambda_{\min}(\hat S_K)$ is bounded away from zero with probability approaching one,
$
\|\hat\gamma_K\|
=
O_p(\|\hat\mu_1-\mu_{1K}\|_{2,n}).
$
Moreover, by the triangle inequality,
$
\|\hat\mu_1-\mu_{1K}\|_{2,n}
\le
\|\hat\mu_1-\mu_1\|_{2,n}
+
\|\mu_1-\mu_{1K}\|_{2,n}.
$
Lemma \ref{lem:mu-rate} gives
$\|\hat\mu_1-\mu_1\|_{2,n}=O_p(r_{\mu,n})$. Moreover, the sieve approximation satisfies
$\|\mu_1-\mu_{1K}\|_{2,n}=O_p(K^{-s_\mu/r})=O_p(r_{\mu,n})$. Hence
$$
\sup_{z\in\mathcal Z}|A_{n,111,s}^a(z)|
\le
\|\hat\gamma_K\|\sup_{z\in\mathcal Z}\|V_n(z)\|
=
O_p(\sqrt K\,r_{\mu,n})
=
o_p(1),
$$
by the condition $\sqrt K\,r_{\mu,n}\to0$.

It remains to control the approximation component. Let
$$
h_K(X):=\frac{\mu_{1K}(X)-\mu_1(X)}{p(X)}.
$$
Then
$$
A_{n,111,b}^a(z)
=
\frac1{\sqrt n}\sum_{i=1}^n
h_K(X_i)(D_i-p_i)\{I_i(z)-F_Z(z)\}.
$$
Since $\E[D_i-p_i\mid X_i]=0$, the process $A_{n,111,b}^a(z)$ is centered conditional on $X_i$. Applying the same Pollard's maximal inequality with envelope $|h_K(X)|$ gives
$$
\sup_{z\in\mathcal Z}|A_{n,111,b}^a(z)|
=
O_p(\|h_K\|_2).
$$
By overlap and the approximation property of $\mu_{1K}$,
$\|h_K\|_2\le C\|\mu_{1K}-\mu_1\|_2=O(K^{-s_\mu/r})=o(1)$.
Thus
$$
\sup_{z\in\mathcal Z}|A_{n,111,b}^a(z)|=o_p(1).
$$
Together with the bounds for $A_{n,111,p}^a(z)$ and $A_{n,111,s}^a(z)$, this proves
$$
\sup_{z\in\mathcal Z}|A_{n,111}^a(z)|=o_p(1).
$$
The proof for the control outcome term is identical after replacing $\mu_1$, $p$, and $\hat p$ by $\mu_0$, $1-p$, and $1-\hat p$, respectively.
\end{proof}

\begin{lemma}
\label{lem:An112a-control}
Suppose that the conditions of Lemmas \ref{lem:p-rate} and \ref{lem:mu-rate} hold. Define
$$
r_{\infty,n}:=\sqrt{\frac{K^3}{n}}+K^{1-s/(2r)}.
$$
Suppose further that $\sqrt n\,r_{\infty,n}^2\to0$,
then
$$
\sup_{z\in\mathcal Z}
\left|
\frac1{\sqrt n}\sum_{i=1}^n
\mu_1(X_i)
\left\{\frac1{\hat p(X_i)}-\frac1{p(X_i)}\right\}
\{D_i-p(X_i)\}
\{\1(Z_i\le z)-F_Z(z)\}
\right|
=o_p(1).
$$
The same conclusion holds with $\mu_0$ and $1-\hat p$ in place of $\mu_1$ and $\hat p$.
\end{lemma}

\begin{proof}
We prove the first assertion; the second is identical after replacing $p$ by
$1-p$. Write $p_i:=p(X_i)$, $\hat p_i:=\hat p(X_i)$,
$e_i:=\hat p_i-p_i$, and $q_i(z):=\1\{Z_i\le z\}-F_Z(z)$. Define
$$
A_{n,112}^a(z)
:=
\frac1{\sqrt n}\sum_{i=1}^n
\mu_1(X_i)
\left(\frac1{\hat p_i}-\frac1{p_i}\right)
(D_i-p_i)q_i(z).
$$
By overlap, Lemma \ref{lem:p-rate}, and the boundedness of $\mu_1$ on the compact support, all denominators below are bounded away from zero with probability approaching one and
$m_1(x):=\mu_1(x)/p(x)^2$ is bounded.

A Taylor expansion of $u\mapsto 1/u$ around $p_i$ gives
$$
\frac1{\hat p_i}-\frac1{p_i}
=
-\frac{e_i}{p_i^2}+r_i,
\qquad
|r_i|\le C e_i^2
$$
with probability approaching one. Hence
$$
A_{n,112}^a(z)
=
A_{n,112,1}^a(z)+A_{n,112,2}^a(z),
$$
where
$$
A_{n,112,1}^a(z)
:=
-\frac1{\sqrt n}\sum_{i=1}^n
m_1(X_i)e_i(D_i-p_i)q_i(z),
$$
and
$$
A_{n,112,2}^a(z)
:=
\frac1{\sqrt n}\sum_{i=1}^n
\mu_1(X_i)r_i(D_i-p_i)q_i(z).
$$
The second term is negligible because $|D_i-p_i|\le1$, $|q_i(z)|\le1$, and $\mu_1$ is bounded:
$$
\sup_{z\in\mathcal Z}|A_{n,112,2}^a(z)|
\le
C\sqrt n\,\|\hat p-p\|_\infty^2
=
O_p(\sqrt n\,r_{\infty,n}^2)
=
o_p(1).
$$

It remains to control $A_{n,112,1}^a(z)$. Let
$p_K^*(x):=\Lambda\{R^K(x)'\pi_K^*\}$ be the population pseudo-true series-logit approximation. Decompose
$e_i=\{\hat p_i-p_K^*(X_i)\}+\{p_K^*(X_i)-p_i\}$ and write
$$
A_{n,112,1}^a(z)
=
A_{n,112,11}^a(z)+A_{n,112,12}^a(z),
$$
where
$$
A_{n,112,11}^a(z)
:=
-\frac1{\sqrt n}\sum_{i=1}^n
m_1(X_i)\{\hat p_i-p_K^*(X_i)\}(D_i-p_i)q_i(z),
$$
and
$$
A_{n,112,12}^a(z)
:=
-\frac1{\sqrt n}\sum_{i=1}^n
m_1(X_i)\{p_K^*(X_i)-p_i\}(D_i-p_i)q_i(z).
$$

We first handle the approximation component. Since $\E[D_i-p_i\mid X_i]=0$, the process $A_{n,112,12}^a(z)$ is centered conditional on $X_i$. Applying the same Pollard's maximal inequality as in Lemma \ref{lem:An111a-control}, with envelope
$|m_1(X)\{p_K^*(X)-p(X)\}|$, gives
$$
\sup_{z\in\mathcal Z}|A_{n,112,12}^a(z)|
=
O_p\!\left(\|m_1(p_K^*-p)\|_2\right).
$$
Since $m_1$ is bounded and Lemma \ref{lem:p-rate} gives
$\|p_K^*-p\|_2=O(K^{-s/(2r)})$, it follows that
$$
\sup_{z\in\mathcal Z}|A_{n,112,12}^a(z)|
=
O_p(K^{-s/(2r)})
=
o_p(1).
$$

Next consider $A_{n,112,11}^a(z)$. A Taylor expansion of
$\Lambda\{R^K(X_i)'\pi\}$ around $\pi_K^*$ gives
$$
\hat p_i-p_K^*(X_i)
=
\Lambda'\{R^K(X_i)'\pi_K^*\}R^K(X_i)'(\hat\pi_K-\pi_K^*)
+\rho_i,
\qquad
|\rho_i|\le C|R^K(X_i)'(\hat\pi_K-\pi_K^*)|^2.
$$
Therefore,
$$
A_{n,112,11}^a(z)
=
A_{n,112,111}^a(z)+A_{n,112,112}^a(z),
$$
where
$$
A_{n,112,111}^a(z)
:=
-(\hat\pi_K-\pi_K^*)'
\frac1{\sqrt n}\sum_{i=1}^n
R^K(X_i)\Lambda'\{R^K(X_i)'\pi_K^*\}
m_1(X_i)(D_i-p_i)q_i(z),
$$
and
$$
A_{n,112,112}^a(z)
:=
-\frac1{\sqrt n}\sum_{i=1}^n
m_1(X_i)\rho_i(D_i-p_i)q_i(z).
$$

For $A_{n,112,111}^a(z)$, define
$$
V_n(z):=
\frac1{\sqrt n}\sum_{i=1}^n
R^K(X_i)\Lambda'\{R^K(X_i)'\pi_K^*\}
m_1(X_i)(D_i-p_i)q_i(z).
$$
By Lemma \ref{lem:centered-lower-orthant-vc} and the same variance--trace argument as in Lemma \ref{lem:An111a-control},
$$
\sup_{z\in\mathcal Z}\|V_n(z)\|=O_p(\sqrt K),
$$
where we use the boundedness of $\Lambda'$, overlap, boundedness of $m_1$, and the HIR normalization.
Hence, by Lemma \ref{lem:p-rate} and the condition $K^4/n\to0$ of Lemma \ref{lem:p-rate},
$$
\sup_{z\in\mathcal Z}|A_{n,112,111}^a(z)|
\le
\|\hat\pi_K-\pi_K^*\|\sup_{z\in\mathcal Z}\|V_n(z)\|
=
O_p\!\left(\sqrt{\frac K n}\right)O_p(\sqrt K)
=
O_p\!\left(\frac K{\sqrt n}\right)
=
o_p(1).
$$

For the Taylor remainder, using $|D_i-p_i|\le1$, $|q_i(z)|\le1$, and boundedness of $m_1$,
$$
\sup_{z\in\mathcal Z}|A_{n,112,112}^a(z)|
\le
\frac C{\sqrt n}\sum_{i=1}^n
|R^K(X_i)'(\hat\pi_K-\pi_K^*)|^2.
$$
Thus
$$
\sup_{z\in\mathcal Z}|A_{n,112,112}^a(z)|
\le
C\sqrt n\,
(\hat\pi_K-\pi_K^*)'\hat S_K(\hat\pi_K-\pi_K^*),
$$
where $\hat S_K:=n^{-1}\sum_{i=1}^nR^K(X_i)R^K(X_i)'$. Since
$\lambda_{\max}(\hat S_K)=O_p(1)$ and
$\|\hat\pi_K-\pi_K^*\|=O_p(\sqrt{K/n})$ by Lemma \ref{lem:p-rate},
$$
\sup_{z\in\mathcal Z}|A_{n,112,112}^a(z)|
=
O_p\!\left(\frac K{\sqrt n}\right)
=
o_p(1).
$$
Combining the bounds for
$A_{n,112,2}^a$, $A_{n,112,12}^a$,
$A_{n,112,111}^a$, and $A_{n,112,112}^a$ proves
$$
\sup_{z\in\mathcal Z}|A_{n,112}^a(z)|=o_p(1).
$$
The proof for the control outcome term is identical after replacing $\mu_1$, $p$, and $\hat p$ by $\mu_0$, $1-p$, and $1-\hat p$, respectively.
\end{proof}

\begin{lemma}
\label{lem:orthogonal-propensity-term}
Suppose that the conditions of Lemmas \ref{lem:p-rate} and \ref{lem:mu-rate} hold. Let
$$
b_i^\phi
:=
-Y_i\left\{
\frac{D_i}{p(X_i)^2}
+
\frac{1-D_i}{\{1-p(X_i)\}^2}
\right\},
\qquad
a_i
:=
\frac{\mu_1(X_i)}{p(X_i)}
+
\frac{\mu_0(X_i)}{1-p(X_i)}.
$$
Then
$$
\sup_{z\in\mathcal Z}
\left|
\frac1{\sqrt n}\sum_{i=1}^n
(a_i+b_i^\phi)\{\hat p(X_i)-p(X_i)\}
\{\1(Z_i\le z)-F_Z(z)\}
\right|
=o_p(1).
$$
\end{lemma}

\begin{proof}
The argument is the uniform analogue of the estimated-propensity-score expansion in the HIR addendum; see also \url{https://keihirano.github.io/papers/AddendumProof_02nov10.pdf}. The present case is simpler because the HIR projection terms are driven by the conditional mean of the multiplier given $X$, which is identically zero here.

Write $p_i:=p(X_i)$, $\hat p_i:=\hat p(X_i)$, $I_i(z):=\1(Z_i\le z)$ and $q_i(z):=I_i(z)-F_Z(z)$. Define $c_i:=a_i+b_i^\phi$.
By the definitions of $a_i$ and $b_i^\phi$,
$$
\E[b_i^\phi\mid X_i]
=
-\frac{\mu_1(X_i)}{p_i}
-
\frac{\mu_0(X_i)}{1-p_i},
$$
and hence $\E[c_i\mid X_i]=0$. Thus the leading multiplier is conditionally centered. Let
$$
M_n(z):=
\frac1{\sqrt n}\sum_{i=1}^n
c_i(\hat p_i-p_i)q_i(z).
$$
Let $p_K^*(x):=\Lambda\{R^K(x)'\pi_K^*\}$ be the population pseudo-true series-logit approximation, with $\pi_K^*$ defined the same as in \eqref{eq:piKstar}. Decompose
$$
\hat p_i-p_i
=
\{\hat p_i-p_K^*(X_i)\}
+
\{p_K^*(X_i)-p_i\},
$$
and write $M_n(z)=M_{n,1}(z)+M_{n,2}(z)$,
where
$$
M_{n,1}(z)
:=
\frac1{\sqrt n}\sum_{i=1}^n
c_i\{\hat p_i-p_K^*(X_i)\}q_i(z),
$$
and
$$
M_{n,2}(z)
:=
\frac1{\sqrt n}\sum_{i=1}^n
c_i\{p_K^*(X_i)-p_i\}q_i(z).
$$

We first control the approximation component. Since $\E[c_i\mid X_i]=0$, the process $M_{n,2}(z)$ is centered. Applying the same Pollard's maximal inequality as in Lemma \ref{lem:An111a-control}, with envelope
$|c(W)|\,|p_K^*(X)-p(X)|$, gives
$$
\sup_{z\in\mathcal Z}|M_{n,2}(z)|
=
O_p\!\left(\|(p_K^*-p)c\|_2\right).
$$
By overlap, the maintained second-moment condition for $Y$, and boundedness of $\mu_1,\mu_0$, we have $\E[c(W)^2\mid X]\le C$. Therefore,
$$
\|(p_K^*-p)c\|_2
\le
C\|p_K^*-p\|_2
=
O(K^{-s/(2r)})
=
o(1),
$$
where the last rate follows from Lemma \ref{lem:p-rate}. Hence $\sup_{z\in\mathcal Z}|M_{n,2}(z)|=o_p(1)$.

It remains to control $M_{n,1}(z)$. By a Taylor expansion of
$\Lambda\{R^K(X_i)'\pi\}$ around $\pi_K^*$,
$$
\hat p_i-p_K^*(X_i)
=
\Lambda'\{R^K(X_i)'\pi_K^*\}R^K(X_i)'(\hat\pi_K-\pi_K^*)
+
\rho_i,
$$
where $|\rho_i|
\le
C|R^K(X_i)'(\hat\pi_K-\pi_K^*)|^2$. Accordingly,
$$
M_{n,1}(z)=M_{n,11}(z)+M_{n,12}(z),
$$
where
$$
M_{n,11}(z)
:=
(\hat\pi_K-\pi_K^*)'
\frac1{\sqrt n}\sum_{i=1}^n
R^K(X_i)\Lambda'\{R^K(X_i)'\pi_K^*\}c_iq_i(z),
$$
and
$$
M_{n,12}(z)
:=
\frac1{\sqrt n}\sum_{i=1}^n
c_i\rho_iq_i(z).
$$

For the linear stochastic term, define
$$
V_n(z)
:=
\frac1{\sqrt n}\sum_{i=1}^n
R^K(X_i)\Lambda'\{R^K(X_i)'\pi_K^*\}c_iq_i(z).
$$
By Lemma \ref{lem:centered-lower-orthant-vc} and the same variance--trace argument as above, $\sup_{z\in\mathcal Z}\|V_n(z)\|=O_p(\sqrt K)$. Therefore, by Lemma \ref{lem:p-rate},
$$
\sup_{z\in\mathcal Z}|M_{n,11}(z)|
\le
\|\hat\pi_K-\pi_K^*\|
\sup_{z\in\mathcal Z}\|V_n(z)\|
=
O_p\!\left(\sqrt{\frac K n}\right)O_p(\sqrt K)
=
O_p\!\left(\frac K{\sqrt n}\right)
=
o_p(1).
$$

For the Taylor remainder, using $|q_i(z)|\le1$,
$$
\sup_{z\in\mathcal Z}|M_{n,12}(z)|
\le
\frac C{\sqrt n}\sum_{i=1}^n
|c_i|\,|R^K(X_i)'(\hat\pi_K-\pi_K^*)|^2.
$$
Hence
$$
\sup_{z\in\mathcal Z}|M_{n,12}(z)|
\le
C\sqrt n\,
(\hat\pi_K-\pi_K^*)'
\left\{
\frac1n\sum_{i=1}^n |c_i|R^K(X_i)R^K(X_i)'
\right\}
(\hat\pi_K-\pi_K^*).
$$
Moreover,
$$
\frac1n\sum_{i=1}^n |c_i|\|R^K(X_i)\|^2=O_p(K),
$$
because $\E[|c_i|\mid X_i]\le C$ and $\E\|R^K(X)\|^2=K$. Therefore,
$$
\sup_{z\in\mathcal Z}|M_{n,12}(z)|
\le
O_p(\sqrt n\,K\|\hat\pi_K-\pi_K^*\|^2)
=
O_p\!\left(\frac{K^2}{\sqrt n}\right)
=
o_p(1),
$$
where the last equality follows from the maintained series-dimension condition.

Combining the bounds for $M_{n,2}$, $M_{n,11}$, and $M_{n,12}$, we obtain
$$
\sup_{z\in\mathcal Z}|M_n(z)|=o_p(1),
$$
which proves the lemma.
\end{proof}

 \paragraph{Proof of Proposition \ref{prop:ICM-process-uniform}.}
\begin{proof}
Let
$$
I_i(z):=\1\{Z_i\le z\},
\qquad
\widehat F_{Z,n}(z):=\frac1n\sum_{i=1}^n I_i(z),
$$
and write $\psi_i:=\psi(W_i;\eta_0)$. Since
$\sum_{i=1}^n(\hat\psi_i-\hat\tau)=0$, we have, for every $z\in\mathcal Z$,
$$
\widehat R_n(z)
=
\frac1{\sqrt n}\sum_{i=1}^n
(\hat\psi_i-\hat\tau)I_i(z)
=
\frac1{\sqrt n}\sum_{i=1}^n
(\hat\psi_i-\hat\tau)\{I_i(z)-\widehat F_{Z,n}(z)\}.
$$
Therefore,
\begin{equation}
    \label{eq:feas_vs_infeas}
    \widehat R_n(z)
-
\frac1{\sqrt n}\sum_{i=1}^n
(\psi_i-\tau_0)\{I_i(z)-F_Z(z)\}=A_n(z)+B_n(z)+C_n(z),
\end{equation}
where
$$
A_n(z)
:=
\frac1{\sqrt n}\sum_{i=1}^n
(\hat\psi_i-\psi_i)
\{I_i(z)-\widehat F_{Z,n}(z)\},
$$
\begin{equation}
    \label{eq:B_n}
    B_n(z)
:=
-\frac1{\sqrt n}(\hat\tau-\tau_0)
\sum_{i=1}^n
\{I_i(z)-\widehat F_{Z,n}(z)\},
\end{equation}
$$
C_n(z)
:=
\frac1{\sqrt n}\sum_{i=1}^n
(\psi_i-\tau_0)
\{F_Z(z)-\widehat F_{Z,n}(z)\}.
$$

The term $B_n(z)$ is identically zero for every $z\in\mathcal Z$, because $\sum_{i=1}^n\{I_i(z)-\widehat F_{Z,n}(z)\}=0$.

Next,
$$
C_n(z)
=
\{F_Z(z)-\widehat F_{Z,n}(z)\}
\frac1{\sqrt n}\sum_{i=1}^n(\psi_i-\tau_0).
$$
Since the centered lower-orthant class is VC-type, $\sup_{z\in\mathcal Z}|\widehat F_{Z,n}(z)-F_Z(z)|=O_p(n^{-1/2})$.
Also, by $\E[\psi(W;\eta_0)]=\tau_0$ and the maintained moment conditions,
$n^{-1/2}\sum_{i=1}^n(\psi_i-\tau_0)=O_p(1)$. Hence
\begin{equation}
    \label{eq:C_n}
    \sup_{z\in\mathcal Z}|C_n(z)|=o_p(1).
\end{equation}

It remains to prove $\sup_{z\in\mathcal Z}|A_n(z)|=o_p(1)$. Define
$$
\phi(W,p):=\frac{Y\{D-p(X)\}}{p(X)\{1-p(X)\}},
\qquad
a(X):=\frac{\mu_1(X)}{p(X)}+\frac{\mu_0(X)}{1-p(X)}.
$$
Let $\phi_i:=\phi(W_i,p)$, $\hat\phi_i:=\phi(W_i,\hat p)$, $a_i:=a(X_i)$, and
$\hat a_i:=\hat\mu_1(X_i)/\hat p(X_i)+\hat\mu_0(X_i)/\{1-\hat p(X_i)\}$. Then
$\psi_i=\phi_i-a_i\{D_i-p(X_i)\}$, $\hat\psi_i=\hat\phi_i-\hat a_i\{D_i-\hat p(X_i)\}$. Hence
\begin{equation}
\label{eq:A_n}
    A_n(z)
=
A_n^\phi(z)-A_n^a(z),
\end{equation}
where
$$
A_n^\phi(z)
:=
\frac1{\sqrt n}\sum_{i=1}^n(\hat\phi_i-\phi_i)
\{I_i(z)-\widehat F_{Z,n}(z)\},
$$
and
$$
A_n^a(z)
:=
\frac1{\sqrt n}\sum_{i=1}^n
\left[
\hat a_i\{D_i-\hat p(X_i)\}
-
a_i\{D_i-p(X_i)\}
\right]
\{I_i(z)-\widehat F_{Z,n}(z)\}.
$$
We first analyze $A_n^\phi(z)$.

\medskip
\noindent\textbf{Expansion of $A_n^\phi(z)$.}

Let $e_i:=\hat p(X_i)-p(X_i)$. Since
$$
\hat\phi_i-\phi_i
=
D_iY_i\left\{\frac1{\hat p(X_i)}-\frac1{p(X_i)}\right\}
-
(1-D_i)Y_i
\left\{\frac1{1-\hat p(X_i)}-\frac1{1-p(X_i)}\right\},
$$
we can write $\hat\phi_i-\phi_i=e_i b_i^\phi$, where
$$
b_i^\phi
:=
-Y_i\left[
\frac{D_i}{p(X_i)\hat p(X_i)}
+
\frac{1-D_i}{\{1-p(X_i)\}\{1-\hat p(X_i)\}}
\right].
$$
Therefore,
$$
A_n^\phi(z)
=
\underbrace{\frac1{\sqrt n}\sum_{i=1}^n e_i b_i^\phi\{I_i(z)-F_Z(z)\}}_{A_{n,1}^\phi(z)}
+
\underbrace{\frac1{\sqrt n}\sum_{i=1}^n e_i b_i^\phi\{F_Z(z)-\widehat F_{Z,n}(z)\}}_{A_{n,2}^\phi(z)}.
$$

The second term is negligible. Indeed, by overlap and Lemma \ref{lem:p-rate}, $p(X_i)$ and $\hat p(X_i)$ are bounded away from zero and one uniformly in $i$ with probability approaching one, so $|b_i^\phi|\le C|Y_i|$ on this event. Hence
$$
\sup_{z\in\mathcal Z}|A_{n,2}^\phi(z)|
\le
\sup_{z\in\mathcal Z}|\widehat F_{Z,n}(z)-F_Z(z)|
\cdot
\frac1{\sqrt n}\sum_{i=1}^n |e_i|\,|b_i^\phi|.
$$
We have
$\sup_{z\in\mathcal Z}|\widehat F_{Z,n}(z)-F_Z(z)|=O_p(n^{-1/2})$. Moreover,
$$
\frac1n\sum_{i=1}^n |e_i|\,|b_i^\phi|
\le
C\left(\frac1n\sum_{i=1}^n e_i^2\right)^{1/2}
\left(\frac1n\sum_{i=1}^n Y_i^2\right)^{1/2}
=o_p(1),
$$
where the last step follows from Lemma \ref{lem:p-rate} and the maintained moment condition for $Y$. Therefore,
$$
\sup_{z\in\mathcal Z}|A_{n,2}^\phi(z)|=O_p(n^{-1/2})\cdot o_p(n^{1/2})=o_p(1).
$$

It remains to analyze $A_{n,1}^\phi(z)$. Decompose
$b_i^\phi=b_i^{\phi,0}+r_i^\phi$, where
$$
b_i^{\phi,0}
:=
-Y_i\left[
\frac{D_i}{p(X_i)^2}
+
\frac{1-D_i}{\{1-p(X_i)\}^2}
\right].
$$
Then
$$
A_{n,1}^\phi(z)
=
A_{n,11}^\phi(z)+A_{n,12}^\phi(z),
$$
where
$$
A_{n,11}^\phi(z)
:=
\frac1{\sqrt n}\sum_{i=1}^n
e_i b_i^{\phi,0}\{I_i(z)-F_Z(z)\},\quad A_{n,12}^\phi(z)
:=
\frac1{\sqrt n}\sum_{i=1}^n
e_i r_i^\phi\{I_i(z)-F_Z(z)\}.
$$
The term $A_{n,11}^\phi(z)$ is the leading first-step contribution from estimating the propensity score, which will be combined below with the corresponding leading term from $A_n^a(z)$.

The remainder $A_{n,12}^\phi(z)$ is of second order. On the same high-probability overlap event, the mean value theorem gives $|r_i^\phi|
\le C |Y_i|\,|e_i|$. Therefore,
$$
\sup_{z\in\mathcal Z}|A_{n,12}^\phi(z)|
\le
\frac1{\sqrt n}\sum_{i=1}^n |e_i|\,|r_i^\phi|
\le
C\sqrt n\,\|\hat p-p\|_\infty^2
\left(\frac1n\sum_{i=1}^n |Y_i|\right).
$$
By the maintained moment condition, $n^{-1}\sum_{i=1}^n |Y_i|=O_p(1)$. Hence
$$
\sup_{z\in\mathcal Z}|A_{n,12}^\phi(z)|
=
O_p\!\left(\sqrt n\,\|\hat p-p\|_\infty^2\right).
$$
By Lemma \ref{lem:p-rate}, $\|\hat p-p\|_\infty
=
O_p\!\left(\sqrt{{K^3}/{n}}+K^{1-s/(2r)}\right)$. Indeed, under $K\asymp n^\nu$,
$$
n^{1/4}\left(\sqrt{\frac{K^3}{n}}+K^{1-s/(2r)}\right)
\asymp
n^{3\nu/2-1/4}
+
n^{1/4+\nu\{1-s/(2r)\}}
\to0
$$
by Assumptions \ref{ass:pscore} and \ref{ass:SLE}. Thus
$\|\hat p-p\|_\infty=o_p(n^{-1/4})$, and consequently $\sup_{z\in\mathcal Z}|A_{n,12}^\phi(z)|=o_p(1)$. Therefore,
\begin{equation}
    \label{eq:A_n_phi}
    A_n^\phi(z)=A_{n,11}^\phi(z)+o_p(1)
\end{equation}
uniformly in $z\in\mathcal Z$.

\medskip
\noindent\textbf{Expansion of $A_n^a(z)$.}

We next turn to the term involving the augmentation component. Recall that
$$
A_n^a(z)
=
\frac1{\sqrt n}\sum_{i=1}^n
\left[
\hat a_i\{D_i-\hat p(X_i)\}
-
a_i\{D_i-p(X_i)\}
\right]
\{I_i(z)-\widehat F_{Z,n}(z)\},
$$
where
$$
a_i:=\frac{\mu_1(X_i)}{p(X_i)}+\frac{\mu_0(X_i)}{1-p(X_i)},
\qquad
\hat a_i:=\frac{\hat\mu_1(X_i)}{\hat p(X_i)}
+\frac{\hat\mu_0(X_i)}{1-\hat p(X_i)}.
$$
Since
$$
\hat a_i\{D_i-\hat p(X_i)\}-a_i\{D_i-p(X_i)\}
=
(\hat a_i-a_i)\{D_i-p(X_i)\}-\hat a_i e_i,
$$
where $e_i:=\hat p(X_i)-p(X_i)$, we write
\begin{equation}
    \label{eq:A_na_decomp}
    A_n^a(z)=A_{n,1}^a(z)-A_{n,2}^a(z),
\end{equation}
with
$$
A_{n,1}^a(z)
:=
\frac1{\sqrt n}\sum_{i=1}^n
(\hat a_i-a_i)\{D_i-p(X_i)\}
\{I_i(z)-\widehat F_{Z,n}(z)\},
$$
and
$$
A_{n,2}^a(z)
:=
\frac1{\sqrt n}\sum_{i=1}^n
\hat a_i e_i
\{I_i(z)-\widehat F_{Z,n}(z)\}.
$$

We first decompose $A_{n,2}^a(z)$, since this term contains the leading first-step contribution. Adding and subtracting $F_Z(z)$, we obtain $A_{n,2}^a(z)
=
A_{n,21}^a(z)+A_{n,22}^a(z)$, where
$$
A_{n,21}^a(z)
:=
\frac1{\sqrt n}\sum_{i=1}^n
\hat a_i e_i\{I_i(z)-F_Z(z)\},\quad A_{n,22}^a(z)
:=
\frac1{\sqrt n}\sum_{i=1}^n
\hat a_i e_i\{F_Z(z)-\widehat F_{Z,n}(z)\}.
$$
The second term is negligible. Indeed,
$$
\sup_{z\in\mathcal Z}|A_{n,22}^a(z)|
\le
\sup_{z\in\mathcal Z}|\widehat F_{Z,n}(z)-F_Z(z)|
\left|
\frac1{\sqrt n}\sum_{i=1}^n\hat a_i e_i
\right|.
$$
We have
$\sup_{z\in\mathcal Z}|\widehat F_{Z,n}(z)-F_Z(z)|=O_p(n^{-1/2})$. Moreover, by Cauchy--Schwarz,
$$
\left|
\frac1{\sqrt n}\sum_{i=1}^n\hat a_i e_i
\right|
\le
\sqrt n
\left(\frac1n\sum_{i=1}^n\hat a_i^2\right)^{1/2}
\|\hat p-p\|_{2,n}.
$$
Under overlap, Lemma \ref{lem:mu-rate}, and the maintained moment conditions,
$n^{-1}\sum_{i=1}^n\hat a_i^2=O_p(1)$. Hence
$$
\sup_{z\in\mathcal Z}|A_{n,22}^a(z)|
=
O_p(\|\hat p-p\|_{2,n})
=
o_p(1),
$$
where the last equality follows from Lemma \ref{lem:p-rate}, Assumption \ref{ass:pscore} and \ref{ass:SLE}.

It remains to study $A_{n,21}^a(z)$. Decompose $A_{n,21}^a(z)=
A_{n,211}^a(z)+A_{n,212}^a(z)$, where
$$
A_{n,211}^a(z)
:=
\frac1{\sqrt n}\sum_{i=1}^n
a_i e_i\{I_i(z)-F_Z(z)\},\quad A_{n,212}^a(z)
:=
\frac1{\sqrt n}\sum_{i=1}^n
(\hat a_i-a_i)e_i\{I_i(z)-F_Z(z)\}.
$$
The first term $A_{n,211}^a(z)$ is the leading contribution from estimating the propensity score in the augmentation part. It will be combined with $A_{n,11}^\phi(z)$ below. The second term is of higher order: since $|I_i(z)-F_Z(z)|\le 1$,
$$
\sup_{z\in\mathcal Z}|A_{n,212}^a(z)|
\le
\sqrt n\,\|\hat a-a\|_{2,n}\|\hat p-p\|_{2,n}.
$$
By overlap and the definitions of $a$ and $\hat a$,
$$
\|\hat a-a\|_{2,n}
=
O_p\!\left(
\|\hat\mu_1-\mu_1\|_{2,n}
+
\|\hat\mu_0-\mu_0\|_{2,n}
+
\|\hat p-p\|_{2,n}
\right).
$$
By Lemmas \ref{lem:p-rate} and \ref{lem:mu-rate},
$$
\|\hat p-p\|_{2,n}
=
O_p\!\left(\sqrt{\frac K n}+K^{-s/(2r)}\right),
\qquad
\|\hat\mu_d-\mu_d\|_{2,n}
=
O_p\!\left(\sqrt{\frac K n}+K^{-s/(2r)}+K^{-s_\mu/r}\right),
$$
for $d=0,1$. Hence
$$
\sqrt n\,\|\hat a-a\|_{2,n}\|\hat p-p\|_{2,n}
=
O_p\!\left[
\sqrt n
\left(\sqrt{\frac K n}+K^{-s/(2r)}+K^{-s_\mu/r}\right)
\left(\sqrt{\frac K n}+K^{-s/(2r)}\right)
\right].
$$
Under $K\asymp n^\nu$, the deterministic rate inside the bracket is bounded by a constant multiple of
$$
\frac K{\sqrt n}
+
K^{1/2-s/(2r)}
+
K^{1/2-s_\mu/r}
+
\sqrt n K^{-s/r}
+
\sqrt n K^{-s/(2r)-s_\mu/r}.
$$
Equivalently, it is of order
$$
n^{\nu-1/2}
+
n^{\nu\{1/2-s/(2r)\}}
+
n^{\nu(1/2-s_\mu/r)}
+
n^{1/2-\nu s/r}
+
n^{1/2-\nu\{s/(2r)+s_\mu/r\}}.
$$
By the maintained rate conditions in Assumption \ref{ass:dist}(c), \ref{ass:pscore} and \ref{ass:SLE}, each exponent is negative. Therefore, $\sqrt n\,\|\hat a-a\|_{2,n}\|\hat p-p\|_{2,n}=o_p(1)$, and consequently $\sup_{z\in\mathcal Z}|A_{n,212}^a(z)|=o_p(1)$. Consequently,
\begin{equation}
    \label{eq:A_n2a_expansion}
    A_{n,2}^a(z)
    =
    A_{n,211}^a(z)+o_p(1)
\end{equation}
uniformly in $z\in\mathcal Z$, where
$$
A_{n,211}^a(z)
=
\frac1{\sqrt n}\sum_{i=1}^n
a_i\{\hat p(X_i)-p(X_i)\}\{I_i(z)-F_Z(z)\}.
$$
We now control $A_{n,1}^a(z)$. Adding and subtracting $F_Z(z)$, write
$$
A_{n,1}^a(z)
=
A_{n,11}^a(z)+A_{n,12}^a(z),
$$
where
$$
A_{n,11}^a(z)
:=
\frac1{\sqrt n}\sum_{i=1}^n
(\hat a_i-a_i)\{D_i-p(X_i)\}\{I_i(z)-F_Z(z)\},
$$
and
$$
A_{n,12}^a(z)
:=
\frac1{\sqrt n}\sum_{i=1}^n
(\hat a_i-a_i)\{D_i-p(X_i)\}\{F_Z(z)-\widehat F_{Z,n}(z)\}.
$$
The second term is negligible. Indeed,
$$
\sup_{z\in\mathcal Z}|A_{n,12}^a(z)|
\le
\sup_{z\in\mathcal Z}|\widehat F_{Z,n}(z)-F_Z(z)|
\left|
\frac1{\sqrt n}\sum_{i=1}^n
(\hat a_i-a_i)\{D_i-p(X_i)\}
\right|.
$$
First, $\sup_{z\in\mathcal Z}|\widehat F_{Z,n}(z)-F_Z(z)|=O_p(n^{-1/2})$. Moreover, by Cauchy--Schwarz and $|D_i-p(X_i)|\le 1$,
$$
\left|
\frac1{\sqrt n}\sum_{i=1}^n
(\hat a_i-a_i)\{D_i-p(X_i)\}
\right|
\le
\sqrt n\,\|\hat a-a\|_{2,n}=o_p(\sqrt{n}).
$$
Therefore, $\sup_{z\in\mathcal Z}|A_{n,12}^a(z)|=o_p(1)$.

It remains to show that $\sup_{z\in\mathcal Z}|A_{n,11}^a(z)|=o_p(1)$. By definition,
$$
\hat a_i-a_i
=
\left\{\frac{\hat\mu_1(X_i)}{\hat p(X_i)}-\frac{\mu_1(X_i)}{p(X_i)}\right\}
+
\left\{\frac{\hat\mu_0(X_i)}{1-\hat p(X_i)}-\frac{\mu_0(X_i)}{1-p(X_i)}\right\}.
$$
Equivalently,
$$
\hat a_i-a_i
=
\frac{\hat\mu_1(X_i)-\mu_1(X_i)}{\hat p(X_i)}
+
\mu_1(X_i)\left\{\frac1{\hat p(X_i)}-\frac1{p(X_i)}\right\}
+
\frac{\hat\mu_0(X_i)-\mu_0(X_i)}{1-\hat p(X_i)}
+
\mu_0(X_i)\left\{\frac1{1-\hat p(X_i)}-\frac1{1-p(X_i)}\right\}.
$$
Accordingly, decompose
$$
A_{n,11}^a(z)
=
\sum_{\ell=1}^4 A_{n,11\ell}^a(z),
$$
where
$$
A_{n,111}^a(z)
:=
\frac1{\sqrt n}\sum_{i=1}^n
\frac{\hat\mu_1(X_i)-\mu_1(X_i)}{\hat p(X_i)}
\{D_i-p(X_i)\}\{I_i(z)-F_Z(z)\},
$$
$$
A_{n,112}^a(z)
:=
\frac1{\sqrt n}\sum_{i=1}^n
\mu_1(X_i)
\left\{\frac1{\hat p(X_i)}-\frac1{p(X_i)}\right\}
\{D_i-p(X_i)\}\{I_i(z)-F_Z(z)\},
$$
$$
A_{n,113}^a(z)
:=
\frac1{\sqrt n}\sum_{i=1}^n
\frac{\hat\mu_0(X_i)-\mu_0(X_i)}{1-\hat p(X_i)}
\{D_i-p(X_i)\}\{I_i(z)-F_Z(z)\},
$$
$$
A_{n,114}^a(z)
:=
\frac1{\sqrt n}\sum_{i=1}^n
\mu_0(X_i)
\left\{\frac1{1-\hat p(X_i)}-\frac1{1-p(X_i)}\right\}
\{D_i-p(X_i)\}\{I_i(z)-F_Z(z)\}.
$$

We now verify that the four terms are uniformly negligible. By Lemma \ref{lem:An111a-control}, applied respectively to the treated and control outcome regressions, 
$$
\sup_{z\in\mathcal Z}|A_{n,111}^a(z)|=o_p(1),
\qquad
\sup_{z\in\mathcal Z}|A_{n,113}^a(z)|=o_p(1),
$$
provided that $\sqrt n\,r_{\mu,n}r_{p,n}\to0$ and $\sqrt K\,r_{\mu,n}\to0$. These rate conditions are implied by Assumptions \ref{ass:dist}--\ref{ass:SLE}. Similarly, by Lemma \ref{lem:An112a-control},
$$
\sup_{z\in\mathcal Z}|A_{n,112}^a(z)|=o_p(1),
\qquad
\sup_{z\in\mathcal Z}|A_{n,114}^a(z)|=o_p(1),
$$
since $\sqrt n\,r_{\infty,n}^2\to0$ and $K/\sqrt n\to0$ under the same assumptions. Hence,
$$
\sup_{z\in\mathcal Z}|A_{n,11}^a(z)|=o_p(1),
$$
and consequently,
\begin{equation}
\label{eq:A_n1a}
    \sup_{z\in\mathcal Z}|A_{n,1}^a(z)|=o_p(1).
\end{equation}
Combining \eqref{eq:A_na_decomp}, \eqref{eq:A_n1a} and \eqref{eq:A_n2a_expansion} gives 
\begin{equation}
    \label{eq:A_na}
    A_{n}^a(z)=-A_{n,211}^a(z)+o_p(1)
\end{equation}
uniformly in $z\in\mathcal{Z}$. Combining \eqref{eq:feas_vs_infeas}, \eqref{eq:B_n}, \eqref{eq:C_n}, \eqref{eq:A_n}, \eqref{eq:A_n_phi} and \eqref{eq:A_na} gives
\begin{align}
     &\quad\widehat R_n(z)-\frac{1}{\sqrt{n}}\sum_{i=1}^n(\psi_i-\tau_0)\left\{I_i(z)-F_Z(z)\right\} \notag\\
     &=A_{n,11}^\phi(z)+A_{n,211}^a(z)+o_p(1) \notag\\
     &=\frac{1}{\sqrt{n}}\sum_{i=1}^n\left(a_i+b_i^\phi\right)\left\{\hat p(X_i)-p(X_i)\right\}\left\{I_i(z)-F_Z(z)\right\}+o_p(1),
\end{align}
uniformly in $z\in\mathcal{Z}$. Finally, by Lemma \ref{lem:orthogonal-propensity-term},
$$
\sup_{z\in\mathcal Z}
\left|
A_{n,11}^\phi(z)+A_{n,211}^a(z)
\right|
=o_p(1).
$$
Together with the preceding expansion,
$$
\sup_{z\in\mathcal Z}
\left|
\widehat R_n(z)-
\frac1{\sqrt n}\sum_{i=1}^n
(\psi_i-\tau_0)\{I_i(z)-F_Z(z)\}
\right|
=o_p(1),
$$
as claimed.

\end{proof}

\paragraph{Proof of Theorem \ref{thm:ICM-Gaussian}}
\begin{proof}
Let $\chi(W):=\psi(W;\eta_0)-\tau_0$. By Proposition \ref{prop:ICM-process-uniform}, it suffices to study
$$
R_n^0(z):=
\frac1{\sqrt n}\sum_{i=1}^n
\chi(W_i)\{\1(Z_i\le z)-F_Z(z)\},
\qquad z\in\mathcal Z.
$$
Under $\mathbb H_0$, \eqref{eq:cate_cmr} gives $\E[\chi(W)\mid Z]=0$. Hence
$$
R_n^0(z)=\mathbb G_n f_z,
\qquad
f_z(w):=\chi(w)\{\1(Z(w)\le z)-F_Z(z)\}.
$$
By Lemma \ref{lem:centered-lower-orthant-vc} and $\E[\chi(W)^2]<\infty$, the class
$\mathcal F:=\{f_z:z\in\mathcal Z\}$ is $P$-Donsker. Therefore
$$
R_n^0 \rightsquigarrow \mG
\qquad \text{in } \ell^\infty(\mathcal Z),
$$
where $\mG$ is mean zero with covariance
$$
\E[f_{z_1}(W)f_{z_2}(W)]
=
\E\!\left[
\chi(W)^2
\{\1(Z\le z_1)-F_Z(z_1)\}
\{\1(Z\le z_2)-F_Z(z_2)\}
\right].
$$
Since $\E[\chi(W)\mid Z]=0$, $\E[\chi(W)^2\mid Z]=\Var(\psi(W;\eta_0)\mid Z)=\sigma_\psi^2(Z)$, giving the stated covariance kernel. The result follows from Proposition \ref{prop:ICM-process-uniform} and Slutsky's theorem.
\end{proof}

\paragraph{Proof of Corollary \ref{cor:null-dist}.}
\begin{proof}
The weak convergence of $KS_n$ follows directly from Theorem \ref{thm:ICM-Gaussian} and the continuous mapping theorem, since the map $z\mapsto \sup_{u\in\mathcal Z}|z(u)|$ is continuous on $\ell^\infty(\mathcal Z)$.

We now prove the convergence of $CvM_n$. By Theorem \ref{thm:ICM-Gaussian},
$$
\widehat R_n \rightsquigarrow \mathbb G
\qquad \text{in } \ell^\infty(\mathcal Z),
$$
and, by the Glivenko--Cantelli theorem,
$$
\sup_{z\in\mathcal Z}|\widehat F_{Z,n}(z)-F_Z(z)|=o_p(1),
$$
where $\mathbb P_n(z):=n^{-1}\sum_{i=1}^n \1\{Z_i\le z\}$.

By the Skorohod representation theorem, there exists a probability space on which one can construct versions of $\widehat R_n$, $\mathbb G$, and $\widehat F_{Z,n}$ such that
$$
\sup_{z\in\mathcal Z}|\widehat R_n(z)-\mathbb G(z)|\to 0
\quad \text{a.s.},
\qquad
\sup_{z\in\mathcal Z}|\widehat F_{Z,n}(z)-F_Z(z)|\to 0
\quad \text{a.s.}
$$
Hence,
\begin{align*}
&\left|
\int_{\mathcal Z}\widehat R_n(z)^2\,d\widehat F_{Z,n}(z)
-
\int_{\mathcal Z}\mathbb G(z)^2\,dF_Z(z)
\right| \\
&\le
\left|
\int_{\mathcal Z}\bigl(\widehat R_n(z)^2-\mathbb G(z)^2\bigr)\,d\widehat F_{Z,n}(z)
\right|
+
\left|
\int_{\mathcal Z}\mathbb G(z)^2\,d\bigl(\widehat F_{Z,n}-F_Z\bigr)(z)
\right|.
\end{align*}
For the first term,
$$
\left|
\int_{\mathcal Z}\bigl(\widehat R_n(z)^2-\mathbb G(z)^2\bigr)\,d\widehat F_{Z,n}(z)
\right|
\le
\sup_{z\in\mathcal Z}\bigl|\widehat R_n(z)^2-\mathbb G(z)^2\bigr|
\to 0
\quad \text{a.s.}
$$
For the second term, note that the sample paths of $\mathbb G$ are almost surely bounded and continuous on $\mathcal Z$. Since $\widehat F_{Z,n}$ converges weakly to $F_Z$ almost surely, the Helly--Bray theorem implies
$$
\int_{\mathcal Z}\mathbb G(z)^2\,d\widehat F_{Z,n}(z)
-
\int_{\mathcal Z}\mathbb G(z)^2\,dF_Z(z)
\to 0
\quad \text{a.s.}
$$
Therefore,
$$
CvM_n
=
\int_{\mathcal Z}\widehat R_n(z)^2\,d\widehat F_{Z,n}(z)
\to
\int_{\mathcal Z}\mathbb G(z)^2\,dF_Z(z)
\quad \text{a.s.},
$$
on the Skorohod space, which implies
$$
CvM_n \xrightarrow{d} \int_{\mathcal Z}\mathbb G(z)^2\,dF_Z(z).
$$
This completes the proof.
\end{proof}

\paragraph{Proof of Theorem \ref{thm:fixed-alt}}
\begin{proof}
Let $\chi(W):=\psi(W;\eta_0)-\tau_0$. By Proposition \ref{prop:ICM-process-uniform},
$$
\sup_{z\in\mathcal Z}
\left|
\widehat R_n(z)-
\frac1{\sqrt n}\sum_{i=1}^n
\chi(W_i)\{ \1(Z_i\le z)-F_Z(z)\}
\right|
=o_p(1).
$$
Dividing by $\sqrt n$, it suffices to study
$$
\frac1n\sum_{i=1}^n
\chi(W_i)\{ \1(Z_i\le z)-F_Z(z)\}.
$$
By Lemma \ref{lem:centered-lower-orthant-vc} and $\E[\chi(W)^2]<\infty$, the class
$$
\left\{
w\mapsto \chi(w)\{\1(Z(w)\le z)-F_Z(z)\}:z\in\mathcal Z
\right\}
$$
is Glivenko--Cantelli. Hence
$$
\sup_{z\in\mathcal Z}
\left|
\frac1n\sum_{i=1}^n
\chi(W_i)\{\1(Z_i\le z)-F_Z(z)\}
-
\E\!\left[\chi(W)\{\1(Z\le z)-F_Z(z)\}\right]
\right|
=o_p(1).
$$
Since $\E[\chi(W)]=0$, the expectation above equals
$$
\E[\chi(W)\1\{Z\le z\}]
=
\Gamma(z).
$$
Therefore,
$$
\sup_{z\in\mathcal Z}
\left|
\frac{1}{\sqrt n}\widehat R_n(z)-\Gamma(z)
\right|
=o_p(1).
$$

It remains to show that $\Gamma$ is not identically zero under $\mathbb H_1$. If $\Gamma(z)=0$ for all $z\in\mathcal Z$, then
$$
\E\!\left[\{\psi(W;\eta_0)-\tau_0\}\1\{Z\le z\}\right]=0
\qquad \forall z\in\mathcal Z.
$$
By the equivalence between the conditional moment restriction and the indicator-weighted unconditional moments, this implies
$$
\E[\psi(W;\eta_0)-\tau_0\mid Z]=0
\quad\text{a.s.},
$$
or equivalently $\tau(Z)=\tau_0$ a.s., contradicting $\mathbb H_1$. Thus $\Gamma\not\equiv0$.

Finally, by dominated convergence, $\Gamma$ is continuous on $\mathcal Z$ under the maintained continuity conditions for $Z$. Hence $\sup_{z\in\mathcal Z}|\Gamma(z)|>0$, and
$$
\frac{KS_n}{\sqrt n}
=
\sup_{z\in\mathcal Z}
\left|
\frac1{\sqrt n}\widehat R_n(z)
\right|
\xrightarrow{p}
\sup_{z\in\mathcal Z}|\Gamma(z)|>0.
$$
Similarly,
$$
\frac{CvM_n}{n}
=
\int_{\mathcal Z}
\left(
\frac1{\sqrt n}\widehat R_n(z)
\right)^2
\,d\widehat F_{Z,n}(z)
\xrightarrow{p}
\int_{\mathcal Z}\Gamma(z)^2\,dF_Z(z)>0.
$$
Therefore, $KS_n$ and $CvM_n$ diverge in probability, proving consistency against fixed alternatives.
\end{proof}

\paragraph{Proof of Theorem \ref{thm:local-alt}}
\begin{proof}
Let $\E_n$ denote expectation under the local alternative. Under $\mathbb H_{1n}$,
$$
\E_n[\psi(W;\eta_0)\mid Z]
=
\tau_n(Z)
=
\tau_0+n^{-1/2}\lambda(Z).
$$
Hence the population average treatment effect under the local alternative is
$$
\tau_{P,n}:=\E_n[\psi(W;\eta_0)]
=
\tau_0+n^{-1/2}\bar\lambda,
\qquad
\bar\lambda:=\E_n[\lambda(Z)].
$$
By the same feasible-to-oracle argument as in Proposition \ref{prop:ICM-process-uniform}, with
$\tau_0$ replaced by $\tau_{P,n}$,
$$
\sup_{z\in\mathcal Z}
\left|
\widehat R_n(z)
-
\frac1{\sqrt n}\sum_{i=1}^n
\{\psi(W_i;\eta_0)-\tau_{P,n}\}
\{ \1\{Z_i\le z\}-F_Z(z)\}
\right|
=o_p(1).
$$
It therefore suffices to study
$$
R_n^0(z):=
\frac1{\sqrt n}\sum_{i=1}^n
\{\psi(W_i;\eta_0)-\tau_{P,n}\}
\{ \1\{Z_i\le z\}-F_Z(z)\}.
$$

Decompose
$$
\begin{aligned}
R_n^0(z)
&=
\frac1{\sqrt n}\sum_{i=1}^n
\Bigl[
\{\psi(W_i;\eta_0)-\tau_{P,n}\}
\{ \1\{Z_i\le z\}-F_Z(z)\}  \\
&\qquad\qquad
-
\E_n\!\left[
\{\psi(W;\eta_0)-\tau_{P,n}\}
\{ \1\{Z\le z\}-F_Z(z)\}
\right]
\Bigr]  \\
&\quad+
\sqrt n\,
\E_n\!\left[
\{\psi(W;\eta_0)-\tau_{P,n}\}
\{ \1\{Z\le z\}-F_Z(z)\}
\right].
\end{aligned}
$$
For the drift term, since
$$
\E_n[\psi(W;\eta_0)-\tau_{P,n}\mid Z]
=
n^{-1/2}\{\lambda(Z)-\bar\lambda\},
$$
we obtain
$$
\begin{aligned}
&\sqrt n\,
\E_n\!\left[
\{\psi(W;\eta_0)-\tau_{P,n}\}
\{ \1\{Z\le z\}-F_Z(z)\}
\right] \\
&\qquad
=
\E_n\!\left[
\{\lambda(Z)-\bar\lambda\}
\{ \1\{Z\le z\}-F_Z(z)\}
\right] \\
&\qquad
=
\E_n\!\left[
\lambda(Z)\{ \1\{Z\le z\}-F_Z(z)\}
\right]
=
\Lambda(z),
\end{aligned}
$$
where the term involving $\bar\lambda$ vanishes because
$\E_n[\1\{Z\le z\}-F_Z(z)]=0$.

For the centered empirical process term, the same Donsker argument as in the proof of
Theorem \ref{thm:ICM-Gaussian} applies to the class
$$
\left\{
w\mapsto
\{\psi(w;\eta_0)-\tau_{P,n}\}
\{ \1\{Z(w)\le z\}-F_Z(z)\}:z\in\mathcal Z
\right\}.
$$
Under the local alternatives, the conditional mean is shifted only by order $n^{-1/2}$, so the stochastic part has the same weak limit as under the null. Therefore,
$$
\frac1{\sqrt n}\sum_{i=1}^n
\Bigl[
\{\psi(W_i;\eta_0)-\tau_{P,n}\}
\{ \1\{Z_i\le z\}-F_Z(z)\}
-
\E_n\{\cdot\}
\Bigr]
\rightsquigarrow \mathbb G
$$
in $\ell^\infty(\mathcal Z)$, where $\mathbb G$ is the centered Gaussian process in
Theorem \ref{thm:ICM-Gaussian}. Consequently,
$$
R_n^0 \rightsquigarrow \mathbb G+\Lambda
\qquad \text{in } \ell^\infty(\mathcal Z).
$$
Combining this with Slutsky's theorem gives
$$
\widehat R_n \rightsquigarrow \mathbb G+\Lambda
\qquad \text{in } \ell^\infty(\mathcal Z).
$$

It remains to justify that $\Lambda\not\equiv0$ when $\lambda(Z)$ is not a.s. constant. If
$\Lambda\equiv0$, then
$$
\E\!\left[\lambda(Z)\{\1\{Z\le z\}-F_Z(z)\}\right]=0
\qquad \forall z\in\mathcal Z,
$$
which implies
$$
\E[\lambda(Z)-\E\lambda(Z)\mid Z]=0
\quad\text{a.s.},
$$
and hence $\lambda(Z)=\E\lambda(Z)$ a.s., a contradiction. Therefore
$\Lambda\not\equiv0$ whenever $\lambda(Z)$ is not a.s. constant.

The weak limits of $KS_n$ and $CvM_n$ follow from the continuous mapping theorem. The CvM conclusion uses the same random-measure replacement argument as in Corollary \ref{cor:null-dist}. This proves the theorem.
\end{proof}

\subsection{Proofs for Section \ref{sec:mb}}

\paragraph{Proof of Theorem \ref{thm:bootstrap}}
\begin{proof}
Write
\[
q_i(z):=\1\{Z_i\le z\}-F_Z(z),
\qquad
\hat q_i(z):=\1\{Z_i\le z\}-\widehat F_{Z,n}(z).
\]
By definition,
\[
\widehat R_n^*(z)
=
\frac1{\sqrt n}\sum_{i=1}^n
V_i(\hat\psi_i-\hat\tau)\hat q_i(z).
\]
We first show that
\begin{equation}
\label{eq:bootstrap-feasible-oracle}
\sup_{z\in\mathcal Z}
\left|
\widehat R_n^*(z)
-
\frac1{\sqrt n}\sum_{i=1}^n
V_i\{\psi(W_i;\eta_0)-\tau_0\}q_i(z)
\right|
=o_p(1).
\end{equation}

Decompose
\[
\widehat R_n^*(z)
-
\frac1{\sqrt n}\sum_{i=1}^n
V_i\{\psi(W_i;\eta_0)-\tau_0\}q_i(z)
=
T_{1n}^*(z)+T_{2n}^*(z),
\]
where
\[
T_{1n}^*(z):=
\frac1{\sqrt n}\sum_{i=1}^n
V_i\{(\hat\psi_i-\hat\tau)-(\psi(W_i;\eta_0)-\tau_0)\}q_i(z),
\]
and
\[
T_{2n}^*(z):=
\frac1{\sqrt n}\sum_{i=1}^n
V_i(\hat\psi_i-\hat\tau)\{\hat q_i(z)-q_i(z)\}.
\]

We first control \(T_{1n}^*(z)\). Since
\[
(\hat\psi_i-\hat\tau)-(\psi(W_i;\eta_0)-\tau_0)
=
\{\hat\psi_i-\psi(W_i;\eta_0)\}-(\hat\tau-\tau_0),
\]
we have
\[
T_{1n}^*(z)
=
\frac1{\sqrt n}\sum_{i=1}^n
V_i\{\hat\psi_i-\psi(W_i;\eta_0)\}q_i(z)
-
(\hat\tau-\tau_0)
\frac1{\sqrt n}\sum_{i=1}^n V_iq_i(z).
\]
The first term is uniformly \(o_p(1)\). Indeed, the nuisance decomposition used in the proof of Proposition \ref{prop:ICM-process-uniform} applies with the instrument \(q_i(z)\) replaced by \(V_iq_i(z)\). Since the multipliers are independent of the data and have bounded support, the same bounds apply to the corresponding starred terms:
\[
A_{n,1}^{\phi,*}(z)=A_{n,11}^{\phi,*}(z)+o_p(1),
\qquad
A_{n,11}^{a,*}(z)=o_p(1),
\qquad
A_{n,21}^{a,*}(z)=A_{n,211}^{a,*}(z)+o_p(1),
\]
uniformly in \(z\in\mathcal Z\), where the starred terms are defined as their unstarred counterparts with \(q_i(z)\) replaced by \(V_iq_i(z)\). Hence
\[
\frac1{\sqrt n}\sum_{i=1}^n
V_i\{\hat\psi_i-\psi(W_i;\eta_0)\}q_i(z)
=
A_{n,11}^{\phi,*}(z)+A_{n,211}^{a,*}(z)+o_p(1)
\]
uniformly in \(z\in\mathcal Z\). Moreover,
\[
A_{n,11}^{\phi,*}(z)+A_{n,211}^{a,*}(z)
=
\frac1{\sqrt n}\sum_{i=1}^n
V_i(a_i+b_i^\phi)\{\hat p(X_i)-p(X_i)\}q_i(z).
\]
The last term is uniformly \(o_p(1)\) by the same orthogonality argument as in Lemma \ref{lem:orthogonal-propensity-term}; boundedness and independence of \(V_i\) leave the envelope, maximal-inequality, and conditional-centering arguments unchanged. Therefore,
\[
\sup_{z\in\mathcal Z}
\left|
\frac1{\sqrt n}\sum_{i=1}^n
V_i\{\hat\psi_i-\psi(W_i;\eta_0)\}q_i(z)
\right|
=o_p(1).
\]
For the second term in \(T_{1n}^*(z)\), \(\hat\tau-\tau_0=o_p(1)\), and the same Pollard maximal-inequality argument applied to the VC-type class \(\{V_iq_i(z):z\in\mathcal Z\}\) gives
\[
\sup_{z\in\mathcal Z}
\left|
\frac1{\sqrt n}\sum_{i=1}^nV_iq_i(z)
\right|
=O_p(1).
\]
Consequently,
\[
\sup_{z\in\mathcal Z}|T_{1n}^*(z)|=o_p(1).
\]

Next consider \(T_{2n}^*(z)\). Since
\(\hat q_i(z)-q_i(z)=F_Z(z)-\widehat F_{Z,n}(z)\),
\[
T_{2n}^*(z)
=
\{F_Z(z)-\widehat F_{Z,n}(z)\}
\frac1{\sqrt n}\sum_{i=1}^n V_i(\hat\psi_i-\hat\tau).
\]
Hence
\[
\sup_{z\in\mathcal Z}|T_{2n}^*(z)|
\le
\sup_{z\in\mathcal Z}|\widehat F_{Z,n}(z)-F_Z(z)|
\left|
\frac1{\sqrt n}\sum_{i=1}^n V_i(\hat\psi_i-\hat\tau)
\right|.
\]
The first factor is \(O_p(n^{-1/2})\). It remains to show that the second factor is \(O_p(1)\).

On the high-probability overlap event, a direct expansion of the DR score gives
\[
|\hat\psi_i-\psi(W_i;\eta_0)|
\le
C\{|\hat\mu_1(X_i)-\mu_1(X_i)|
+
|\hat\mu_0(X_i)-\mu_0(X_i)|\}
+
C(1+|Y_i|)|\hat p(X_i)-p(X_i)|.
\]
Thus,
\[
\begin{aligned}
\frac1n\sum_{i=1}^n\{\hat\psi_i-\psi(W_i;\eta_0)\}^2
&\le
C\|\hat\mu_1-\mu_1\|_{2,n}^2
+
C\|\hat\mu_0-\mu_0\|_{2,n}^2  \\
&\quad+
C\frac1n\sum_{i=1}^n(1+Y_i^2)\{\hat p(X_i)-p(X_i)\}^2.
\end{aligned}
\]
By Lemmas \ref{lem:p-rate} and \ref{lem:mu-rate}, and by the same conditional Markov argument used in the proof of Lemma \ref{lem:mu-rate},
\begin{equation}
    \label{eq:dr-score-l2}
    \frac1n\sum_{i=1}^n\{\hat\psi_i-\psi(W_i;\eta_0)\}^2=o_p(1).
\end{equation}
Moreover,
\[
|\hat\tau-\tau_0|
\le
\left|\frac1n\sum_{i=1}^n\{\hat\psi_i-\psi(W_i;\eta_0)\}\right|
+
\left|\frac1n\sum_{i=1}^n\{\psi(W_i;\eta_0)-\tau_0\}\right|
=o_p(1).
\]
Therefore,
\[
\frac1n\sum_{i=1}^n
\left[
(\hat\psi_i-\hat\tau)-\{\psi(W_i;\eta_0)-\tau_0\}
\right]^2
=o_p(1),
\]
and hence
\[
\frac1n\sum_{i=1}^n(\hat\psi_i-\hat\tau)^2
\le
\frac2n\sum_{i=1}^n\{\psi(W_i;\eta_0)-\tau_0\}^2
+
o_p(1)
=
O_p(1),
\]
where the last equality follows from \(\E[\{\psi(W;\eta_0)-\tau_0\}^2]<\infty\).

Conditional on the sample, using \(\E[V_i]=0\), \(\E[V_i^2]=1\), and independence of the multipliers,
\[
\E\!\left[
\left.
\left(
\frac1{\sqrt n}\sum_{i=1}^n V_i(\hat\psi_i-\hat\tau)
\right)^2
\right|\,W_1,\ldots,W_n
\right]
=
\frac1n\sum_{i=1}^n(\hat\psi_i-\hat\tau)^2
=
O_p(1).
\]
By Markov's inequality,
\[
\frac1{\sqrt n}\sum_{i=1}^n V_i(\hat\psi_i-\hat\tau)=O_p(1).
\]
Therefore,
\[
\sup_{z\in\mathcal Z}|T_{2n}^*(z)|=o_p(1).
\]
Combining the bounds for \(T_{1n}^*(z)\) and \(T_{2n}^*(z)\) proves \eqref{eq:bootstrap-feasible-oracle}.

It remains to obtain the weak limit of the infeasible multiplier process
$$
R_n^{*,0}(z)
:=
\frac1{\sqrt n}\sum_{i=1}^n
V_i\{\psi(W_i;\eta_0)-\tau_0\}
\{\1\{Z_i\le z\}-F_Z(z)\}.
$$
Under $\mathbb H_0$ or $\mathbb H_{1n}$, the class
$$
\left\{
w\mapsto
\{\psi(w;\eta_0)-\tau_0\}
\{\1(Z(w)\le z)-F_Z(z)\}:z\in\mathcal Z
\right\}
$$
is $P$-Donsker by Lemma \ref{lem:centered-lower-orthant-vc} and the square-integrability of $\psi(W;\eta_0)-\tau_0$. Under $\mathbb H_{1n}$, the conditional mean is shifted only by order $n^{-1/2}$, which does not affect the covariance limit of the multiplier process. Hence, by the multiplier central limit theorem applied to process $R_n^{*,0}(z)$; see \cite{van1996weak} (Theorem 2.9.2, p. 179),
$$
R_n^{*,0}\rightsquigarrow^* \mathbb G
\qquad
\text{in probability, in }\ell^\infty(\mathcal Z),
$$
where $\mathbb G$ is the Gaussian process in Theorem \ref{thm:ICM-Gaussian}. Together with \eqref{eq:bootstrap-feasible-oracle}, this gives
$$
\widehat R_n^*\rightsquigarrow^* \mathbb G
\qquad
\text{in probability, in }\ell^\infty(\mathcal Z).
$$

Under the fixed alternative $\mathbb H_1$, the multiplier process remains centered because $\E[V_i]=0$, but the covariance kernel is generally different from the null kernel. The same argument gives
$$
\widehat R_n^*\rightsquigarrow^* \mathbb G_1
\qquad
\text{in probability, in }\ell^\infty(\mathcal Z),
$$
where $\mathbb G_1$ is a mean-zero Gaussian process with covariance function
$$
K_1(z_1,z_2)
=
\E\!\left[
\{\psi(W;\eta_0)-\tau_0\}^2
\{\1\{Z\le z_1\}-F_Z(z_1)\}
\{\1\{Z\le z_2\}-F_Z(z_2)\}
\right].
$$
This proves the theorem.
\end{proof}
\subsection{Proofs for Section \ref{sec:extension}}

\paragraph{Proof of Proposition \ref{prop:parametric-process-uniform}}
\begin{proof}
Write
$$
I_i(z):=\1\{Z_i\le z\},\qquad
\psi_i:=\psi(W_i;\eta_0),\qquad
h_i(\theta):=h(Z_i,\theta),
$$
and let $\dot h_i(\theta):=\dot h_\theta(Z_i,\theta)$ and
$\ddot h_i(\theta):=\ddot h_{\theta\theta}(Z_i,\theta)$. Also write $\chi_i^\dagger:=\psi_i-h_i(\theta_0)$.

We first record the following oracle-replacement bounds for the feasible DR score:
\begin{equation}
\label{eq:param-oracle-replacement-indicator}
\sup_{z\in\mathcal Z}
\left|
\frac1{\sqrt n}\sum_{i=1}^n
(\hat\psi_i-\psi_i)I_i(z)
\right|
=o_p(1),
\end{equation}
and
\begin{equation}
\label{eq:param-oracle-replacement-hdot}
\left\|
\frac1{\sqrt n}\sum_{i=1}^n
(\hat\psi_i-\psi_i)\dot h_i(\theta_0)
\right\|
=o_p(1).
\end{equation}
They follow from the nuisance decomposition used in the proof of Proposition \ref{prop:ICM-process-uniform}. To see the first claim, define
$$
\widetilde A_n(z)
:=
\frac1{\sqrt n}\sum_{i=1}^n
(\hat\psi_i-\psi_i)\{I_i(z)-F_Z(z)\}.
$$
Using the notation in the proof of Proposition \ref{prop:ICM-process-uniform},
$$
\widetilde A_n(z)
=
A_{n,1}^{\phi}(z)-A_{n,11}^{a}(z)+A_{n,21}^{a}(z).
$$
The proof of Proposition \ref{prop:ICM-process-uniform} shows that
$$
A_{n,1}^{\phi}(z)=A_{n,11}^{\phi}(z)+o_p(1),
\qquad
A_{n,11}^{a}(z)=o_p(1),
\qquad
A_{n,21}^{a}(z)=A_{n,211}^{a}(z)+o_p(1),
$$
uniformly in $z\in\mathcal Z$. Hence $\widetilde A_n(z)
=
A_{n,11}^{\phi}(z)+A_{n,211}^{a}(z)+o_p(1)$ uniformly in $z\in\mathcal Z$, where
$$
A_{n,11}^{\phi}(z)+A_{n,211}^{a}(z)
=
\frac1{\sqrt n}\sum_{i=1}^n
(a_i+b_i^\phi)\{\hat p(X_i)-p(X_i)\}\{I_i(z)-F_Z(z)\},
$$
with
$$
a_i:=
\frac{\mu_1(X_i)}{p(X_i)}
+
\frac{\mu_0(X_i)}{1-p(X_i)},
\qquad
b_i^\phi:=
-Y_i\left\{
\frac{D_i}{p(X_i)^2}
+
\frac{1-D_i}{\{1-p(X_i)\}^2}
\right\}.
$$
By Lemma \ref{lem:orthogonal-propensity-term},
$$
\sup_{z\in\mathcal Z}
\left|A_{n,11}^{\phi}(z)+A_{n,211}^{a}(z)\right|
=o_p(1),
$$
and therefore
$$
\sup_{z\in\mathcal Z}
\left|
\frac1{\sqrt n}\sum_{i=1}^n
(\hat\psi_i-\psi_i)\{I_i(z)-F_Z(z)\}
\right|
=o_p(1).
$$

The same argument remains valid when $I_i(z)-F_Z(z)$ is replaced by $I_i(z)$. The bounds for $A_{n,1}^{\phi}(z)$, $A_{n,11}^{a}(z)$, and $A_{n,21}^{a}(z)$ use only the same nuisance-rate bounds, boundedness of the instrument, and Pollard's maximal inequality for the lower-orthant VC class. The final leading term also remains conditionally centered because $Z_i\subset X_i$ and
$$
\E[a_i+b_i^\phi\mid X_i]=0,
\qquad
\E\!\left[(a_i+b_i^\phi)I_i(z)\mid X_i\right]
=
I_i(z)\E[a_i+b_i^\phi\mid X_i]
=0.
$$
Thus the same argument as in Lemma A.6 gives
$$
\sup_{z\in\mathcal Z}
\left|
\frac1{\sqrt n}\sum_{i=1}^n
(a_i+b_i^\phi)\{\hat p(X_i)-p(X_i)\}I_i(z)
\right|
=o_p(1),
$$
which proves \eqref{eq:param-oracle-replacement-indicator}. The finite-dimensional bound
\eqref{eq:param-oracle-replacement-hdot} follows by the same reasoning with $I_i(z)$
replaced componentwise by $\dot h_i(\theta_0)$; the moment bound in Assumption
\ref{ass:param-moment} supplies the required square-integrable envelope.

We next establish consistency of $\hat\theta$. Define
$$
Q_n(\theta):=\frac1n\sum_{i=1}^n\{\hat\psi_i-h_i(\theta)\}^2,\quad
Q_n^0(\theta):=\frac1n\sum_{i=1}^n\{\psi_i-h_i(\theta)\}^2,\quad Q(\theta):=\E\!\left[\{\psi(W;\eta_0)-h(Z,\theta)\}^2\right].
$$
By Assumption \ref{ass:param-space}, $\theta\mapsto \{\psi(W;\eta_0)-h(Z,\theta)\}^2$ is continuous almost surely. Moreover, 
$
\sup_{\theta\in\Theta} \{\psi(W;\eta_0)-h(Z,\theta)\}^2
\le
2\,\psi(W;\eta_0)^2+2H_0(Z)^2,
$
which is integrable by Assumption \ref{ass:param-moment}. Therefore, by Lemma 2.4 of \cite{newey1994large},
$$
\sup_{\theta\in\Theta}|Q_n^0(\theta)-Q(\theta)|=o_p(1).
$$
Moreover, by the elementary inequality
$$
\sup_{\theta\in\Theta}|Q_n(\theta)-Q_n^0(\theta)|
\le
\frac1n\sum_{i=1}^n(\hat\psi_i-\psi_i)^2
+
2\left\{\frac1n\sum_{i=1}^n(\hat\psi_i-\psi_i)^2\right\}^{1/2}
\left\{\frac1n\sum_{i=1}^n\sup_{\theta\in\Theta}|\psi_i-h_i(\theta)|^2\right\}^{1/2},
$$
together with \eqref{eq:dr-score-l2} and the moment bounds, we have
$$
\sup_{\theta\in\Theta}|Q_n(\theta)-Q_n^0(\theta)|=o_p(1).
$$
Therefore,
$$
\sup_{\theta\in\Theta}|Q_n(\theta)-Q(\theta)|=o_p(1).
$$
Since $Q(\theta)$ is uniquely minimized at $\theta_0$ under Assumption \ref{ass:param-ident}, the argmin theorem (Theorem 2.1 of \cite{newey1994large}) implies
$\hat\theta\xrightarrow{p}\theta_0$.

Since $\theta_0$ is an interior point of $\Theta$, the first-order condition holds:
$$
0=
\frac1n\sum_{i=1}^n
\{\hat\psi_i-h_i(\hat\theta)\}\dot h_i(\hat\theta).
$$

We now derive the linear representation of $\hat\theta$. Let
$$
S_n(\theta):=\frac1n\sum_{i=1}^n
\{\hat\psi_i-h_i(\theta)\}\dot h_i(\theta).
$$
By the mean-value expansion of $S_n(\hat\theta)$ around $\theta_0$,
$$
0
=
S_n(\theta_0)
+
\dot S_{n,\theta}(\bar\theta)(\hat\theta-\theta_0),
$$
where $\bar\theta$ lies on the line segment joining $\hat\theta$ and $\theta_0$. First,
$$
S_n(\theta_0)
=
\frac1n\sum_{i=1}^n\chi_i^\dagger\dot h_i(\theta_0)
+
\frac1n\sum_{i=1}^n(\hat\psi_i-\psi_i)\dot h_i(\theta_0),
$$
and hence, by \eqref{eq:param-oracle-replacement-hdot},
$$
S_n(\theta_0)
=
\frac1n\sum_{i=1}^n\chi_i^\dagger\dot h_i(\theta_0)
+
o_p(n^{-1/2}).
$$
Next,
$$
\dot S_{n,\theta}(\theta)
=
\frac1n\sum_{i=1}^n
\left[
-\dot h_i(\theta)\dot h_i(\theta)'
+
\{\hat\psi_i-h_i(\theta)\}\ddot h_i(\theta)
\right].
$$
We claim that
\begin{equation}
\label{eq:param-jacobian-conv}
\dot S_{n,\theta}(\bar\theta)\xrightarrow{p}-H.
\end{equation}
To prove this claim, define the infeasible sample Jacobian
$$
\dot S_{n,\theta}^0(\theta)
:=
\frac1n\sum_{i=1}^n
\left[
-\dot h_i(\theta)\dot h_i(\theta)'
+
\{\psi_i-h_i(\theta)\}\ddot h_i(\theta)
\right],
$$
and its population counterpart
$$
\dot S_\theta(\theta)
:=
\E\left[
-\dot h_\theta(Z,\theta)\dot h_\theta(Z,\theta)'
+
\{\psi(W;\eta_0)-h(Z,\theta)\}\ddot h_{\theta\theta}(Z,\theta)
\right].
$$
Fix $\varepsilon>0$ and let $\mathcal N_\varepsilon:=\{\theta\in\Theta:\|\theta-\theta_0\|\le\varepsilon\}$. On the event $\{\bar\theta\in\mathcal N_\varepsilon\}$,
$$
\begin{aligned}
\|\dot S_{n,\theta}(\bar\theta)-\dot S_\theta(\theta_0)\|
&\le
\sup_{\theta\in\mathcal N_\varepsilon}
\|\dot S_{n,\theta}(\theta)-\dot S_{n,\theta}^0(\theta)\|  \\
&\quad+
\sup_{\theta\in\mathcal N_\varepsilon}
\|\dot S_{n,\theta}^0(\theta)-\dot S_\theta(\theta)\|
+
\sup_{\theta\in\mathcal N_\varepsilon}
\|\dot S_\theta(\theta)-\dot S_\theta(\theta_0)\|.
\end{aligned}
$$
The first term is $o_p(1)$. Indeed, by Cauchy--Schwarz,
$$
\begin{aligned}
\sup_{\theta\in\mathcal N_\varepsilon}
\|\dot S_{n,\theta}(\theta)-\dot S_{n,\theta}^0(\theta)\|
&\le
\frac1n\sum_{i=1}^n
|\hat\psi_i-\psi_i|
\sup_{\theta\in\Theta}\|\ddot h_i(\theta)\|  \\
&\le
\left\{\frac1n\sum_{i=1}^n(\hat\psi_i-\psi_i)^2\right\}^{1/2}
\left\{\frac1n\sum_{i=1}^n H_2(Z_i)^2\right\}^{1/2}
=o_p(1),
\end{aligned}
$$
where we use \eqref{eq:dr-score-l2} and Assumption \ref{ass:param-moment}. The second term is $o_p(1)$ by the uniform law of large numbers applied to the class indexed by $\theta\in\mathcal N_\varepsilon$, whose envelope is controlled by $H_1(Z)^2+\{|\psi(W;\eta_0)|+H_0(Z)\}H_2(Z)$. This envelope is integrable by the maintained moment conditions and Cauchy--Schwarz. The third term can be made arbitrarily small by taking $\varepsilon$ sufficiently small, by dominated convergence and Assumptions \ref{ass:param-space} and \ref{ass:param-moment}. Since $\bar\theta\xrightarrow{p}\theta_0$, we have $\mP(\bar\theta\in\mathcal N_\varepsilon)\to1$, and hence $\dot S_{n,\theta}(\bar\theta)-\dot S_\theta(\theta_0)=o_p(1)$.

Finally, under $\mathbb H_0^\dagger$, $\E[\psi(W;\eta_0)-h(Z,\theta_0)\mid Z]=0$. Therefore,
$$
\begin{aligned}
\dot S_\theta(\theta_0)
&=
-\E[\dot h_\theta(Z,\theta_0)\dot h_\theta(Z,\theta_0)']
+
\E\!\left[
\{\psi(W;\eta_0)-h(Z,\theta_0)\}
\ddot h_{\theta\theta}(Z,\theta_0)
\right] \\
&=
-\E[\dot h_\theta(Z,\theta_0)\dot h_\theta(Z,\theta_0)']
= -H.
\end{aligned}
$$
This proves \eqref{eq:param-jacobian-conv}.
The nonsingularity of $H$ then implies
\begin{equation}
\label{eq:param-theta-linear}
\sqrt n(\hat\theta-\theta_0)
=
H^{-1}\frac1{\sqrt n}\sum_{i=1}^n
\chi_i^\dagger\dot h_i(\theta_0)
+
o_p(1).
\end{equation}
In particular, $\sqrt n(\hat\theta-\theta_0)=O_p(1)$.

We now expand the test process. We have
$$
\widehat R_n^\dagger(z)
=
\frac1{\sqrt n}\sum_{i=1}^n
\{\hat\psi_i-h_i(\theta_0)\}I_i(z)
-
\frac1{\sqrt n}\sum_{i=1}^n
\{h_i(\hat\theta)-h_i(\theta_0)\}I_i(z).
$$
By \eqref{eq:param-oracle-replacement-indicator},
$$
\sup_{z\in\mathcal Z}
\left|
\frac1{\sqrt n}\sum_{i=1}^n
\{\hat\psi_i-h_i(\theta_0)\}I_i(z)
-
\frac1{\sqrt n}\sum_{i=1}^n
\chi_i^\dagger I_i(z)
\right|
=o_p(1).
$$
For the second term, Taylor expansion gives
$$
h_i(\hat\theta)-h_i(\theta_0)
=
\dot h_i(\theta_0)'(\hat\theta-\theta_0)+r_i,
\qquad
|r_i|\le H_2(Z_i)\|\hat\theta-\theta_0\|^2.
$$
Therefore,
$$
\sup_{z\in\mathcal Z}
\left|
\frac1{\sqrt n}\sum_{i=1}^n r_i I_i(z)
\right|
\le
\left(\frac1n\sum_{i=1}^n H_2(Z_i)\right)
\sqrt n\,\|\hat\theta-\theta_0\|^2
=o_p(1).
$$
Define
$$
G_n(z):=\frac1n\sum_{i=1}^n\dot h_i(\theta_0)I_i(z),
\qquad
G(z):=\E[\dot h_\theta(Z,\theta_0)\1\{Z\le z\}].
$$
By Lemma \ref{lem:centered-lower-orthant-vc} and the moment bound for $H_1$, $\sup_{z\in\mathcal Z}\|G_n(z)-G(z)\|=o_p(1).$ Since $\sqrt n(\hat\theta-\theta_0)=O_p(1)$, it follows that
$$
\sup_{z\in\mathcal Z}
\left|
\frac1{\sqrt n}\sum_{i=1}^n
\{h_i(\hat\theta)-h_i(\theta_0)\}I_i(z)
-
G(z)'\sqrt n(\hat\theta-\theta_0)
\right|
=o_p(1).
$$
Combining the preceding displays,
$$
\widehat R_n^\dagger(z)
=
\frac1{\sqrt n}\sum_{i=1}^n
\chi_i^\dagger I_i(z)
-
G(z)'\sqrt n(\hat\theta-\theta_0)
+
o_p(1)
$$
uniformly in $z\in\mathcal Z$. Substituting \eqref{eq:param-theta-linear} yields
$$
\widehat R_n^\dagger(z)
=
\frac1{\sqrt n}\sum_{i=1}^n
\chi_i^\dagger I_i(z)
-
G(z)'H^{-1}
\frac1{\sqrt n}\sum_{i=1}^n
\chi_i^\dagger\dot h_i(\theta_0)
+
o_p(1),
$$
uniformly in $z\in\mathcal Z$. Hence
$$
\sup_{z\in\mathcal Z}
\left|
\widehat R_n^\dagger(z)
-
\frac1{\sqrt n}\sum_{i=1}^n
\chi_i^\dagger
\{I_i(z)-G(z)'H^{-1}\dot h_i(\theta_0)\}
\right|
=o_p(1).
$$
Recalling that $\chi_i^\dagger=\psi(W_i;\eta_0)-h(Z_i,\theta_0)$ proves the lemma.
\end{proof}

\paragraph{Proof of Proposition \ref{prop:clate-process-uniform}}
\begin{proof}
Write
\[
\psi_i^L
:=
\psi_Y^L(W_i;\eta_0)-\tau_0\psi_D^L(W_i;\eta_0),
\qquad
\widehat\psi_i^L(\tau)
:=
\widehat\psi_Y^L(W_i)-\tau\widehat\psi_D^L(W_i),
\]
and let \(I_i(z):=\1\{Z_i\le z\}\). Define
\[
\kappa(z):=\E[\psi_D^L(W;\eta_0)\1\{Z\le z\}],
\qquad
\kappa_P:=\E[\psi_D^L(W;\eta_0)].
\]
By local relevance, \(\kappa_P=\mP(\mathcal{C})>0\).

First, under the maintained IV analogues of the first-step smoothness and rate conditions, the oracle-replacement argument used in Proposition \ref{prop:ICM-process-uniform} continues to apply. In particular, equation \eqref{eq:param-oracle-replacement-indicator}, applied with \(A\) as the binary assignment variable and \(Y-\tau_0D\) as the outcome, gives
\[
\sup_{z\in\mathcal Z}
\left|
\frac1{\sqrt n}\sum_{i=1}^n
\{\widehat\psi_i^L(\tau_0)-\psi_i^L\}I_i(z)
\right|
=o_p(1).
\]
Similarly, applying the same nuisance-replacement argument to the first-stage score yields
\[
\sup_{z\in\mathcal Z}
\left|
\frac1n\sum_{i=1}^n
\{\widehat\psi_D^L(W_i)-\psi_D^L(W_i;\eta_0)\}I_i(z)
\right|
=o_p(1),
\qquad
\frac1n\sum_{i=1}^n
\{\widehat\psi_D^L(W_i)-\psi_D^L(W_i;\eta_0)\}
=o_p(1).
\]
Together with the uniform law of large numbers for the lower-orthant class, this implies
\[
\sup_{z\in\mathcal Z}
\left|
\frac1n\sum_{i=1}^n \widehat\psi_D^L(W_i)I_i(z)-\kappa(z)
\right|
=o_p(1),
\qquad
\frac1n\sum_{i=1}^n \widehat\psi_D^L(W_i)=\kappa_P+o_p(1).
\]

Next, since
\[
\widehat\tau-\tau_0
=
\frac{n^{-1}\sum_{i=1}^n\widehat\psi_i^L(\tau_0)}
     {n^{-1}\sum_{i=1}^n\widehat\psi_D^L(W_i)},
\]
the oracle replacement with \(I_i(z)\equiv1\), together with the preceding convergence of the denominator, gives
\[
\sqrt n(\widehat\tau-\tau_0)
=
\frac1{\kappa_P}\frac1{\sqrt n}\sum_{i=1}^n\psi_i^L
+o_p(1).
\]

Now decompose
\[
\widehat R_n^L(z)
=
\frac1{\sqrt n}\sum_{i=1}^n
\widehat\psi_i^L(\tau_0)I_i(z)
-
\sqrt n(\widehat\tau-\tau_0)
\left\{
\frac1n\sum_{i=1}^n\widehat\psi_D^L(W_i)I_i(z)
\right\}.
\]
Using the preceding displays, uniformly over \(z\in\mathcal Z\),
\[
\widehat R_n^L(z)
=
\frac1{\sqrt n}\sum_{i=1}^n\psi_i^L I_i(z)
-
\frac{\kappa(z)}{\kappa_P}
\frac1{\sqrt n}\sum_{i=1}^n\psi_i^L
+o_p(1).
\]
Equivalently,
\[
\sup_{z\in\mathcal Z}
\left|
\widehat R_n^L(z)
-
\frac1{\sqrt n}\sum_{i=1}^n
\psi_i^L
\left\{
I_i(z)-\frac{\kappa(z)}{\kappa_P}
\right\}
\right|
=o_p(1).
\]

Finally, under the conditional IV assumptions and monotonicity, \(\psi_D^L(W;\eta_0)\) identifies the first-stage complier weight. Hence
\[
\frac{\kappa(z)}{\kappa_P}
=
\frac{\E[\psi_D^L(W;\eta_0)\1\{Z\le z\}]}
{\E[\psi_D^L(W;\eta_0)]}
=
F_{Z\mid \mathcal{C}}(z).
\]
Substituting this identity into the preceding expansion gives the claimed representation.
\end{proof}
\newpage

\section{Uniform IPW Linearization}
\label{sec:app_add_HIR}

This appendix provides the formal uniform IPW linearization result referenced in the main text. The main text develops the test from the doubly robust score. The result below shows that the same infeasible process is obtained if one instead starts from the plug-in IPW process and accounts for the first-order effect of estimating the propensity score. The proof has two steps. We first derive the desired uniform representation by combining the ATE expansion of \citet{hirano2003efficient} with a uniform version of their estimated-propensity-score expansion. We then verify the uniform expansion needed in the first step.

\newtheorem*{assumptionstar}{Assumption}
The following assumptions are imposed only for the IPW linearization in this appendix. 
They coincide with the low-level regularity conditions used in \citet{hirano2003efficient} for the series logit estimator.

\begin{assumptionstar}[\ref{ass:dist}$'$: Distribution of $X,Y(0),Y(1)$]
\label{ass:dist_ipw}
(a) The support of the $r$-dimensional covariate vector $X$ is a Cartesian product of compact intervals,
$\mathcal X=\prod_{j=1}^{r}[x_{lj},x_{uj}]$. Moreover, $X$ admits a density $f$ on $\mathcal X$ such that
$0<\underline c\le f(x)\le \bar c<\infty$ for all $x\in\mathcal X$. (b) For $d=0,1$, $\E[|Y(d)|^{2+\delta}]<\infty$ for some $\delta>0$, and
$\sup_{x\in\mathcal X}\Var(Y(d)\mid X=x)<\infty$. (c) The conditional mean functions $\mu_d(x):=\E[Y(d)\mid X=x]$, $d=0,1$, are continuously differentiable on $\mathcal X$.
\end{assumptionstar}

\begin{assumptionstar}[\ref{ass:pscore}$'$: Selection probability]
\label{ass:pscore_ipw}
The propensity score $p(x):=\mP(D=1\mid X=x)$ is $s$-times continuously differentiable on $\mathcal X$, with $s\ge 7r$.
\end{assumptionstar}

\begin{assumptionstar}[\ref{ass:SLE}$'$: Series logit estimator]
\label{ass:SLE_ipw}
The series logit estimator of $p(x)$ uses a power series with $K=a n^\nu$ for some $a>0$ and ${r}/{4(s-r)}<\nu<1/9$.
\end{assumptionstar}

\begin{lemma}
\label{lem:ipw-linearization-uniform}
Suppose Assumptions \ref{ass:unconfd}, \ref{ass:overlap}, 3', 4' and 5' hold. Let
$$
\phi(W,p)
:=
\frac{Y\{D-p(X)\}}{p(X)\{1-p(X)\}},
$$
and define the plug-in IPW estimator of the average treatment effect by
$$
\hat\tau_{\mathrm{IPW}}
:=
\frac1n\sum_{i=1}^n \phi(W_i,\hat p).
$$
Then
$$
\sup_{z\in\mathcal Z}
\left|
\frac1{\sqrt n}\sum_{i=1}^n
\{\phi(W_i,\hat p)-\hat\tau_{\mathrm{IPW}}\}\1\{Z_i\le z\}
-
\frac1{\sqrt n}\sum_{i=1}^n
\left(\psi(W_i;\eta_0)-\tau_0\right)\{\1\{Z_i\le z\}-F_Z(z)\}
\right|
=o_p(1).
$$
\end{lemma}

We now prove Lemma \ref{lem:ipw-linearization-uniform}. The only nonstandard ingredient relative to the ATE result in \citet{hirano2003efficient} is that the IPW functional is indexed by $z$ through the lower-orthant indicator $\1\{Z_i\le z\}$. The uniform version of this first-step expansion is established after the proof.

\paragraph{Proof of Lemma \ref{lem:ipw-linearization-uniform}.}
\begin{proof}
Write
$$
\phi(W,p):=\frac{Y\left\{D-p(X)\right\}}{p(X)\left\{1-p(X)\right\}}
,\quad M_n(p;z):=\frac{1}{n}\sum_{i=1}^n \phi(W_i,p)\1\{Z_i\le z\},
\quad
\widehat F_{Z,n}(z):=\frac{1}{n}\sum_{i=1}^n \1\{Z_i\le z\}.
$$
Then
$$
\widehat R_n(z)=\sqrt n\bigl(M_n(\hat p;z)-\hat\tau\,\widehat F_{Z,n}(z)\bigr).
$$
Adding and subtracting $M_n(p;z)$, $\tau_0$, and $F_Z(z)$, we obtain
\begin{align*}
\widehat R_n(z)
&=\sqrt n\bigl(M_n(p;z)-\tau_0F_Z(z)\bigr)
+\sqrt n\bigl(M_n(\hat p;z)-M_n(p;z)\bigr) \\
&\qquad
-F_Z(z)\sqrt n(\hat\tau-\tau_0)
-\tau_0\sqrt n\bigl(\widehat F_{Z,n}(z)-F_Z(z)\bigr) \\
&\qquad
-\sqrt n(\hat\tau-\tau_0)\bigl(\widehat F_{Z,n}(z)-F_Z(z)\bigr).
\end{align*}
Since
$$
M_n(p;z)-\tau_0F_Z(z)
=
\frac{1}{n}\sum_{i=1}^n\bigl(\phi(W_i,p)-\tau_0\bigr)\1\{Z_i\le z\}
+\tau_0\bigl(\widehat F_{Z,n}(z)-F_Z(z)\bigr),
$$
the terms involving $\tau_0(\widehat F_{Z,n}(z)-F_Z(z))$ cancel, so that
\begin{align*}
\widehat R_n(z)
&=
\frac{1}{\sqrt n}\sum_{i=1}^n\bigl(\phi(W_i,p)-\tau_0\bigr)\1\{Z_i\le z\}
+\sqrt n\bigl(M_n(\hat p;z)-M_n(p;z)\bigr) \\
&\qquad
-F_Z(z)\sqrt n(\hat\tau-\tau_0)
-\sqrt n(\hat\tau-\tau_0)\bigl(\widehat F_{Z,n}(z)-F_Z(z)\bigr).
\end{align*}

Now invoke three standard ingredients. First, by Theorem 1 of \cite{hirano2003efficient}, the asymptotic linear representation for $\hat\tau$ is given by,
$$
\sqrt n(\hat\tau-\tau_0)=\frac{1}{\sqrt n}\sum_{i=1}^n \xi(W_i)+o_p(1).
$$
Second, the required uniform linearization of the first-step estimation effect follows from the uniform extension of the HIR expansion established below. In our setting, the relevant weighted IPW functional is indexed by $z\in\mathcal Z$ through the lower-orthant indicator
$$
g_z(Z):=\1\{Z\le z\},
$$
and the corresponding first-step adjustment takes the usual HIR form with
$$
a(X):=\frac{\mu_1(X)}{p(X)}+\frac{\mu_0(X)}{1-p(X)}.
$$
The argument below shows that the HIR remainder remains negligible uniformly over $z$; closely related uniform expansions appear in Lemma 3.1 of \citet{hsu2017consistent} and Lemma 2 of \citet{firpo2016identification}. Therefore,
$$
\sup_{z\in\mathcal Z}
\left|
\sqrt n\bigl(M_n(\hat p;z)-M_n(p;z)\bigr)
+\frac{1}{\sqrt n}\sum_{i=1}^n \1\{Z_i\le z\}\,a(X_i)\{D_i-p(X_i)\}
\right|
=o_p(1).
$$
Third, by the multivariate Glivenko--Cantelli lemma,
$$
\sup_{z\in\mathcal Z}\bigl|\widehat F_{Z,n}(z)-F_Z(z)\bigr|=o_p(1).
$$

Substituting the first two expansions above into the previous display gives
\begin{align*}
\widehat R_n(z)
&=
\frac{1}{\sqrt n}\sum_{i=1}^n
\Bigl[\phi(W_i,p)-\tau_0-a(X_i)\{D_i-p(X_i)\}\Bigr]\1\{Z_i\le z\} \\
&\qquad
-\frac{F_Z(z)}{\sqrt n}\sum_{i=1}^n
\Bigl[\phi(W_i,p)-\tau_0-a(X_i)\{D_i-p(X_i)\}\Bigr] \\
&\qquad
-\sqrt n(\hat\tau-\tau_0)\bigl(\widehat F_{Z,n}(z)-F_Z(z)\bigr)
+o_p(1) \\
&=
\frac{1}{\sqrt n}\sum_{i=1}^n \xi(W_i)\bigl(\1\{Z_i\le z\}-F_Z(z)\bigr)
-\sqrt n(\hat\tau-\tau_0)\bigl(\widehat F_{Z,n}(z)-F_Z(z)\bigr)
+o_p(1).
\end{align*}
Taking suprema over $z\in\mathcal Z$, the last term is $o_p(1)$ because $\sqrt n(\hat\tau-\tau_0)=O_p(1)$ and $\sup_{z\in\mathcal Z}\bigl|\widehat F_{Z,n}(z)-F_Z(z)\bigr|=o_p(1)$. Hence
$$
\sup_{z\in\mathcal Z}
\left|
\widehat R_n(z)-
\frac{1}{\sqrt n}\sum_{i=1}^n
\xi(W_i)\bigl(\1\{Z_i\le z\}-F_Z(z)\bigr)
\right|
=o_p(1),
$$
as claimed.
\end{proof}

It remains to justify the uniform first-step expansion invoked in the proof. This is the part that extends the estimated-propensity-score expansion in Theorem 1 of \cite{hirano2003efficient} to the present setting, where the IPW functional is indexed by $z\in\mathcal Z$. To simplify the exposition, we only treat the contribution from the treated sample; the corresponding argument for the untreated sample is entirely analogous. Our objective is to show that the HIR remainder remains asymptotically negligible uniformly over $z$. Specifically, we prove that
$$
\sup_{z\in\mathcal Z}
\left|
\frac{1}{\sqrt N}\sum_{i=1}^N \frac{T_iY_i\mathbf 1\{Z_i\le z\}}{\hat p_K(X_i)}
-
\frac{1}{\sqrt N}\sum_{i=1}^N \left[\frac{T_iY_i\mathbf 1\{Z_i\le z\}}{p(X_i)}
-
\1\left\{Z_i\le z\right\}
\frac{\mu_1(X_i)}{p(X_i)}\left(T_i-p(X_i)\right)\right]
\right|
=o_p(1).
$$

For ease of comparison, we follow the notation in the HIR addendum throughout this verification, even when it differs slightly from the notation used in the main text. We do not rewrite the entire HIR proof. Instead, we use exactly the same decomposition as in HIR and explain, term by term, what additional argument is needed to pass from the pointwise case to the uniform case. The original pointwise proof can be found in the HIR addendum; see \url{https://keihirano.github.io/papers/AddendumProof_02nov10.pdf}.
\begin{proof}
With $\hat p_K(X_i)=L(R^K(X_i)^\prime\hat\pi_K)$, following the decomposition of HIR, we write
\begin{align}
&
\qquad\frac{1}{\sqrt N}\sum_{i=1}^N \frac{T_iY_i\mathbf 1\{Z_i\le z\}}{\hat p_K(X_i)}
\nonumber\\
&\quad=
\frac{1}{\sqrt N}\sum_{i=1}^N
\left(
\frac{T_iY_i}{\hat p_K(X_i)}
-
\frac{T_iY_i}{p(X_i)}
+
\frac{T_iY_i}{p(X_i)^2}\bigl(\hat p_K(X_i)-p(X_i)\bigr)
\right)\mathbf 1\{Z_i\le z\}
\tag{4}\\
&\quad\quad+
\frac{1}{\sqrt N}\sum_{i=1}^N
\left(
-\frac{T_iY_i\mathbf 1\{Z_i\le z\}}{p(X_i)^2}\bigl(\hat p_K(X_i)-p(X_i)\bigr)
+
\int_\mathcal{X} \frac{\mu_1(u)}{p(u)}\bigl(\hat p_K(u)-p(u)\bigr)\1(u_Z\le z)\,dF_0(u)
\right)
\tag{5}\\
&\quad\quad-
\sqrt N\int_\mathcal{X} \frac{\mu_1(u)}{p(u)}\bigl(\hat p_K(u)-p(u)\bigr)\1(u_Z\le z)\,dF_0(u)
-
\frac{1}{\sqrt N}\sum_{i=1}^N
\tilde\delta_K(z;X_i)\frac{T_i-p_K(X_i)}{\sqrt{p_K(X_i)(1-p_K(X_i))}}
\tag{6}
\\
&\quad\quad+
\frac{1}{\sqrt N}\sum_{i=1}^N
\bigl(\tilde\delta_K(z;X_i)-\delta_K(z;X_i)\bigr)
\frac{T_i-p_K(X_i)}{\sqrt{p_K(X_i)(1-p_K(X_i))}}
\tag{7}
\\
&\quad\quad+
\frac{1}{\sqrt N}\sum_{i=1}^N
\left(
\delta_K(z;X_i)\frac{T_i-p_K(X_i)}{\sqrt{p_K(X_i)(1-p_K(X_i))}}
-
\delta_0(z;X_i)\frac{T_i-p(X_i)}{\sqrt{p(X_i)(1-p(X_i))}}
\right)
\tag{8}
\\
&\quad\quad+
\frac{1}{\sqrt N}\sum_{i=1}^N
\left\{
\frac{T_iY_i\mathbf 1\{Z_i\le z\}}{p(X_i)}
+
\delta_0(z;X_i)\frac{T_i-p(X_i)}{\sqrt{p(X_i)(1-p(X_i))}}
\right\}.
\tag{9}
\end{align}
Here $F_0$ denotes the population distribution of $X$, and the objects $\tilde\delta_K(z;\cdot)$, $\delta_K(z;\cdot)$, and $\delta_0(z;\cdot)$ are defined by
\begin{align*}
\widetilde{\Psi}_K(z)
&:=
-\int
\frac{\mu_1(u)}{p(u)}
L'\!\bigl(R^K(u)^\prime \widehat{\pi}_K\bigr)
R^K(u)\,
\mathbf 1\{u_Z\le z\}\,dF_0(u),\\
\Psi_K(z)
&:=
-\int
\frac{\mu_1(u)}{p(u)}
L'\!\bigl(R^K(u)^\prime \pi_K\bigr)
R^K(u)\,
\mathbf 1\{u_Z\le z\}\,dF_0(u),\\
\widetilde{\delta}_K(z;u)
&:=
\widetilde{\Psi}_K(z)^\prime
\widetilde{\Sigma}_K^{-1}
R^K(u)\sqrt{p_K(u)\bigl(1-p_K(u)\bigr)},\\
\delta_K(z;u)
&:=
\Psi_K(z)^\prime
\Sigma_K^{-1}
R^K(u)\sqrt{p_K(u)\bigl(1-p_K(u)\bigr)},\\
\delta_0(z;u)
&:=
-\mathbf 1\{u_Z\le z\}\,
\frac{\mu_1(u)}{p(u)}
\sqrt{p(u)\bigl(1-p(u)\bigr)}.
\end{align*}
Thus, to establish the desired uniform asymptotic linearization, it suffices to show that terms \textup{(4)}--\textup{(8)} are each $o_p(1)$ uniformly over $z\in\mathcal Z$. In the sequel, we verify these bounds one by one, emphasizing only the modifications needed for the uniform case.

\paragraph{Bound on (4).}
Let
$$
A_{4,N}(z):=
\frac{1}{\sqrt N}\sum_{i=1}^N
\left(
\frac{T_iY_i}{\hat p_K(X_i)}
-
\frac{T_iY_i}{p(X_i)}
+
\frac{T_iY_i}{p(X_i)^2}\bigl(\hat p_K(X_i)-p(X_i)\bigr)
\right)\mathbf 1\{Z_i\le z\}.
$$
As in HIR, we use the identity
$$
\frac{TY}{\hat p_K(X)}
-
\frac{TY}{p(X)}
+
\frac{TY}{p(X)^2}\bigl(\hat p_K(X)-p(X)\bigr)
=
\frac{TY}{p(X)^2\hat p_K(X)}\bigl(\hat p_K(X)-p(X)\bigr)^2,
$$
and then decompose
$$
\bigl(\hat p_K(X_i)-p(X_i)\bigr)^2
=
\bigl(\hat p_K(X_i)-p_K(X_i)\bigr)^2
+
\bigl(p_K(X_i)-p(X_i)\bigr)^2
+
2\bigl(\hat p_K(X_i)-p_K(X_i)\bigr)\bigl(p_K(X_i)-p(X_i)\bigr).
$$
The only additional issue in the present uniform setting is the indicator $\mathbf 1\{Z_i\le z\}$. This causes no difficulty, since
$
0\le \mathbf 1\{Z_i\le z\}\le 1
\quad\text{for all }z\in\mathcal Z.
$
Hence
\begin{align*}
\sup_{z\in\mathcal Z}|A_{4,N}(z)|
&\le
\frac{1}{\sqrt N}\sum_{i=1}^N
\left|
\frac{T_iY_i}{\hat p_K(X_i)}
-
\frac{T_iY_i}{p(X_i)}
+
\frac{T_iY_i}{p(X_i)^2}\bigl(\hat p_K(X_i)-p(X_i)\bigr)
\right|,
\end{align*}
so the original HIR argument applies verbatim after replacing the pointwise term by the RHS above. In particular, the bounds for the three components corresponding to \textup{(14)}--\textup{(16)} in HIR remain unchanged, and we obtain
$$
\sup_{z\in\mathcal Z}|A_{4,N}(z)|
=
O_p\!\left(\frac{\zeta(K)^3}{\sqrt N}\right)
+
O_p\!\left(\sqrt N\,\zeta(K)^2K^{-s/r}\right)
+
O_p\!\left(\zeta(K)^{5/2}K^{-s/(2r)}\right).
$$
Under the same regularity conditions as in HIR, the RHS is $o_p(1)$. Therefore,
$
\sup_{z\in\mathcal Z}|A_{4,N}(z)|=o_p(1).
$

\paragraph{Bound on (5).}
Following the decomposition in HIR, write the term in \textup{(5)} as
$$
A_{5,N}(z)=A_{5,N}^{(23)}(z)+A_{5,N}^{(24)}(z),
$$
where
\begin{align*}
A_{5,N}^{(23)}
&:=
\frac{1}{\sqrt N}\sum_{i=1}^N
\left(
-\frac{T_iY_i}{p(X_i)^2}\bigl(\hat p_K(X_i)-p_K(X_i)\bigr)\mathbf 1\{Z_i\le z\}
+
\int_{\mathcal X}\frac{\mu_1(u)}{p(u)}\bigl(\hat p_K(u)-p_K(u)\bigr)\mathbf 1\{u_Z\le z\}\,dF_0(u)
\right),\\
A_{5,N}^{(24)}
&:=
\frac{1}{\sqrt N}\sum_{i=1}^N
\left(
-\frac{T_iY_i}{p(X_i)^2}\bigl(p_K(X_i)-p(X_i)\bigr)\mathbf 1\{Z_i\le z\}
+
\int_{\mathcal X}\frac{\mu_1(u)}{p(u)}\bigl(p_K(u)-p(u)\bigr)\mathbf 1\{u_Z\le z\}\,dF_0(u)
\right).
\end{align*}

We first consider $A_{5,N}^{(24)}(z)$. This term contains no estimation effect, and is therefore an empirical process indexed by the lower-orthant class. Define
$$
\xi_K(W)
:=
-\frac{TY}{p(X)^2}\bigl(p_K(X)-p(X)\bigr),
\qquad
g_z(t):=\mathbf 1\{t\le z\},
\qquad
\mathcal G_Z:=\{g_z:z\in\mathcal Z\}.
$$
Then
$$
A_{5,N}^{(24)}(z)
=
\mathbb G_N\!\left(\xi_K(W)\,g_z(Z)\right).
$$
Since $\mathcal G_Z$ is a VC class, the class
$
\mathcal F_{24,K}
:=
\{\xi_K(W)\,g_z(Z):z\in\mathcal Z\}
$
is a VC-subgraph class with envelope
$
F_{24,K}(W):=|\xi_K(W)|.
$
Hence Pollard's maximal inequality (see Theorem 2.14.1 in \cite{van1996weak}) yields
$$
\E\!\left[\sup_{z\in\mathcal Z}|A_{5,N}^{(24)}(z)|\right]
\lesssim
\|F_{24,K}\|_{P,2}.
$$
Now
$$
\|F_{24,K}\|_{P,2}^2
=
\E\!\left[
\frac{T Y^2}{p(X)^4}\bigl(p_K(X)-p(X)\bigr)^2
\right]
\le
C\,\E\!\left[
\frac{E(Y^2\mid X,T=1)}{p(X)^3}\bigl(p_K(X)-p(X)\bigr)^2
\right].
$$
Exactly as in HIR, using $E(Y^2)<\infty$, overlap, and the approximation bound for $p_K-p$, we obtain
$$
\|F_{24,K}\|_{P,2}
=
O\!\left(\zeta(K)K^{-s/(2r)}\right).
$$
Therefore
$$
\sup_{z\in\mathcal Z}|A_{5,N}^{(24)}(z)|
=
O_p\!\left(\zeta(K)K^{-s/(2r)}\right).
$$

We next consider $A_{5,N}^{(23)}(z)$. By the mean value theorem,
$$
\hat p_K(u)-p_K(u)
=
L'\!\bigl(R^K(u)^\prime \widetilde\pi_K(u)\bigr)\,R^K(u)^\prime(\hat\pi_K-\pi_K),
$$
where $\widetilde\pi_K(u)$ lies on the segment joining $\hat\pi_K$ and $\pi_K$. Hence
$$
A_{5,N}^{(23)}(z)
=
W_K(z)^\prime(\hat\pi_K-\pi_K),
$$
where
\begin{align*}
W_K(z)
:=
\frac{1}{\sqrt N}\sum_{i=1}^N
\Bigg[
&-\frac{T_iY_i}{p(X_i)^2}
L'\!\bigl(R^K(X_i)^\prime \widetilde\pi_K(X_i)\bigr)
R^K(X_i)\mathbf 1\{Z_i\le z\} \\
&\qquad\qquad
+\int_{\mathcal X}\frac{\mu_1(u)}{p(u)}
L'\!\bigl(R^K(u)^\prime \widetilde\pi_K(u)\bigr)
R^K(u)\mathbf 1\{u_Z\le z\}\,dF_0(u)
\Bigg].
\end{align*}
A second application of the mean value theorem gives, exactly as in HIR,
$$
W_K(z)=W_{1K}(z)-W_{2K}(z)+W_{3K}(z),
$$
where $W_{1K}(z)$ is the analogue of HIR's \textup{(33)} with the indicator $\mathbf 1\{Z_i\le z\}$ inserted, and $W_{2K}(z)$ and $W_{3K}(z)$ are the analogues of HIR's \textup{(34)}--\textup{(35)}.

The term $W_{1K}(z)$ is the only genuinely indexed empirical-process term. For each coordinate $j=1,\dots,K$, let $W_{1K,j}(z)$ denote its $j$th component. Then
$
\{W_{1K,j}(z):z\in\mathcal Z\}
$
is again indexed by the lower-orthant class. By Pollard's maximal inequality, for each $j$,
$$
\E\!\left[\sup_{z\in\mathcal Z}|W_{1K,j}(z)|^2\right]
\lesssim
P\!\left[\xi_{K,j}(W)^2\right],
$$
where $\xi_{K,j}(W)$ is the corresponding $j$th coordinate envelope. Summing over $j$ and using
$$
\sup_{z\in\mathcal Z}\|W_{1K}(z)\|^2
\le
\sum_{j=1}^K \sup_{z\in\mathcal Z}|W_{1K,j}(z)|^2,
$$
we obtain exactly the same order as in HIR:
$
\E\!\left[\sup_{z\in\mathcal Z}\|W_{1K}(z)\|\right]
\le
C\zeta(K).
$

For $W_{2K}(z)$ and $W_{3K}(z)$, the indicator $\mathbf 1\{Z_i\le z\}$ causes no additional difficulty, since it is bounded by one. Thus the original HIR bounds continue to hold uniformly over $z$, and we obtain
$
\E\!\left[\sup_{z\in\mathcal Z}\|W_{2K}(z)\|\right]
\le
C\zeta(K)^{3/2},
\
\E\!\left[\sup_{z\in\mathcal Z}\|W_{3K}(z)\|\right]
\le
C\zeta(K)^{3/2}.
$
Combining these bounds with the HIR rate for $\|\hat\pi_K-\pi_K\|$, we obtain
$$
\sup_{z\in\mathcal Z}|A_{5,N}^{(23)}(z)|
=
O_p\!\left(\frac{\zeta(K)^2}{\sqrt N}\right).
$$

Putting the two parts together, we conclude that
$$
\sup_{z\in\mathcal Z}|A_{5,N}(z)|
=
O_p\!\left(\zeta(K)K^{-s/(2r)}\right)
+
O_p\!\left(\frac{\zeta(K)^2}{\sqrt N}\right).
$$
This is of the same order as in HIR. Under the maintained regularity conditions,
$
\sup_{z\in\mathcal Z}|A_{5,N}(z)|=o_p(1).
$

\paragraph{Bound on (6).}
Let
$$
A_{6,N}(z)
:=
-\sqrt N\int_{\mathcal X}\frac{\mu_1(u)}{p(u)}\bigl(\hat p_K(u)-p(u)\bigr)\mathbf 1\{u_Z\le z\}\,dF_0(u)
-
\frac{1}{\sqrt N}\sum_{i=1}^N
\tilde\delta_K(z;X_i)\frac{T_i-p_K(X_i)}{\sqrt{p_K(X_i)(1-p_K(X_i))}}.
$$
As in HIR, we decompose the first integral as
\begin{align*}
-\sqrt N\int_{\mathcal X}\frac{\mu_1(u)}{p(u)}\bigl(\hat p_K(u)-p(u)\bigr)\mathbf 1\{u_Z\le z\}\,dF_0(u)
&=
-\sqrt N\int_{\mathcal X}\frac{\mu_1(u)}{p(u)}\bigl(\hat p_K(u)-p_K(u)\bigr)\mathbf 1\{u_Z\le z\}\,dF_0(u)\\
&\quad
-\sqrt N\int_{\mathcal X}\frac{\mu_1(u)}{p(u)}\bigl(p_K(u)-p(u)\bigr)\mathbf 1\{u_Z\le z\}\,dF_0(u).
\end{align*}
The first term on the right-hand side is handled exactly as in HIR. Using the mean value expansion for $\hat p_K-p_K$ together with the definition of $\widetilde{\Psi}_K(z)$ and $\widetilde{\delta}_K(z;\cdot)$, one obtains the same identity as in the pointwise proof,
$$
-\sqrt N\int_{\mathcal X}\frac{\mu_1(u)}{p(u)}\bigl(\hat p_K(u)-p_K(u)\bigr)\mathbf 1\{u_Z\le z\}\,dF_0(u)
=
\frac{1}{\sqrt N}\sum_{i=1}^N
\tilde\delta_K(z;X_i)\frac{T_i-p_K(X_i)}{\sqrt{p_K(X_i)(1-p_K(X_i))}}.
$$
Hence these two terms cancel, and we are left with
$$
A_{6,N}(z)
=
\sqrt N\int_{\mathcal X}\frac{\mu_1(u)}{p(u)}\bigl(p_K(u)-p(u)\bigr)\mathbf 1\{u_Z\le z\}\,dF_0(u).
$$

Thus the only issue in the uniform case is the indicator $\mathbf 1\{u_Z\le z\}$. Since $0\le \mathbf 1\{u_Z\le z\}\le 1$, we obtain
\begin{align*}
\sup_{z\in\mathcal Z}|A_{6,N}(z)|
&\le
\sqrt N\int_{\mathcal X}
\left|\frac{\mu_1(u)}{p(u)}\right|
|p_K(u)-p(u)|\,dF_0(u) \\
&\le
C\sqrt N\,\|p_K-p\|_\infty.
\end{align*}
Using the HIR bound
$$
\|p_K-p\|_\infty
=
O\!\left(\zeta(K)K^{-s/(2r)}\right),
$$
we conclude that
$$
\sup_{z\in\mathcal Z}|A_{6,N}(z)|
=
O\!\left(\sqrt N\,\zeta(K)K^{-s/(2r)}\right).
$$
This is exactly the same order as in HIR. Under the maintained regularity conditions, $\sup_{z\in\mathcal Z}|A_{6,N}(z)|=o_p(1)$.

\paragraph{Bound on (7).}
Let
$$
A_{7,N}(z)
:=
\frac{1}{\sqrt N}\sum_{i=1}^N
\bigl(\widetilde\delta_K(z;X_i)-\delta_K(z;X_i)\bigr)
\frac{T_i-p_K(X_i)}{\sqrt{p_K(X_i)(1-p_K(X_i))}}.
$$
As in HIR, write
$$
A_{7,N}(z)
=
\bigl(\widetilde\Psi_K(z)^\prime \widetilde\Sigma_K^{-1}
-
\Psi_K(z)^\prime \Sigma_K^{-1}\bigr)V_K,
\qquad
V_K:=\frac{1}{\sqrt N}\sum_{i=1}^N R^K(X_i)\bigl(T_i-p_K(X_i)\bigr).
$$
Then, decompose
\begin{align*}
A_{7,N}(z)
&=
\bigl(\widetilde\Psi_K(z)-\Psi_K(z)\bigr)^\prime \widetilde\Sigma_K^{-1}V_K
+
\Psi_K(z)^\prime\bigl(\widetilde\Sigma_K^{-1}-\Sigma_K^{-1}\bigr)V_K.
\end{align*}
Hence
\begin{align*}
\sup_{z\in\mathcal Z}|A_{7,N}(z)|
&\le
\sup_{z\in\mathcal Z}
\left|
\bigl(\widetilde\Psi_K(z)-\Psi_K(z)\bigr)^\prime \widetilde\Sigma_K^{-1}V_K
\right|
+
\sup_{z\in\mathcal Z}
\left|
\Psi_K(z)^\prime\bigl(\widetilde\Sigma_K^{-1}-\Sigma_K^{-1}\bigr)V_K
\right|.
\end{align*}

We first bound the difference $\widetilde\Psi_K(z)-\Psi_K(z)$. By definition,
\begin{align*}
\widetilde\Psi_K(z)-\Psi_K(z)
&=
-\int_{\mathcal X}
\frac{\mu_1(u)}{p(u)}
\Bigl[
L'\!\bigl(R^K(u)^\prime\widehat\pi_K\bigr)
-
L'\!\bigl(R^K(u)^\prime\pi_K\bigr)
\Bigr]
R^K(u)\mathbf 1\{u_Z\le z\}\,dF_0(u).
\end{align*}
Since $L''$ is bounded, the mean value theorem gives
$$
\left|
L'\!\bigl(R^K(u)^\prime\widehat\pi_K\bigr)
-
L'\!\bigl(R^K(u)^\prime\pi_K\bigr)
\right|
\le
C\,\|R^K(u)\|\,\|\widehat\pi_K-\pi_K\|.
$$
Therefore, using $\mathbf 1\{u_Z\le z\}\le 1$,
\begin{align*}
\sup_{z\in\mathcal Z}\|\widetilde\Psi_K(z)-\Psi_K(z)\|
&\le
C\int_{\mathcal X}\|R^K(u)\|^2\,dF_0(u)\,\|\widehat\pi_K-\pi_K\| \\
&\le
C\,\zeta(K)^2\|\widehat\pi_K-\pi_K\|.
\end{align*}
Likewise, again using $\mathbf 1\{u_Z\le z\}\le 1$, we have
\begin{align*}
\sup_{z\in\mathcal Z}\|\Psi_K(z)\|
&\le
\int_{\mathcal X}
\left|\frac{\mu_1(u)}{p(u)}\right|
L'\!\bigl(R^K(u)^\prime\pi_K\bigr)\|R^K(u)\|\,dF_0(u)
\le
C\,\zeta(K).
\end{align*}

We now bound the two terms above. By the matrix norm inequality,
\begin{align*}
\sup_{z\in\mathcal Z}
\left|
\bigl(\widetilde\Psi_K(z)-\Psi_K(z)\bigr)^\prime \widetilde\Sigma_K^{-1}V_K
\right|
&\le
\frac{1}{\lambda_{\min}(\widetilde\Sigma_K)}
\|V_K\|\sup_{z\in\mathcal Z}\|\widetilde\Psi_K(z)-\Psi_K(z)\| \\
&\le
C\,\zeta(K)^2\|\widehat\pi_K-\pi_K\|\,\|V_K\|.
\end{align*}
Similarly,
\begin{align*}
\sup_{z\in\mathcal Z}
\left|
\Psi_K(z)^\prime\bigl(\widetilde\Sigma_K^{-1}-\Sigma_K^{-1}\bigr)V_K
\right|
&\le
\sup_{z\in\mathcal Z}\|\Psi_K(z)\|\,
\bigl\|\bigl(\widetilde\Sigma_K^{-1}-\Sigma_K^{-1}\bigr)V_K\bigr\| \\
&\le
C\,\zeta(K)\,
\bigl\|\bigl(\widetilde\Sigma_K^{-1}-\Sigma_K^{-1}\bigr)V_K\bigr\|.
\end{align*}

At this point the argument is exactly the same as in HIR. The bounds on $\|V_K\|$ and on
$
\bigl\|\bigl(\widetilde\Sigma_K^{-1}-\Sigma_K^{-1}\bigr)V_K\bigr\|
$
do not involve the index $z$, and therefore remain unchanged. Combining these with the bounds above yields the same stochastic order as in HIR:
$$
\sup_{z\in\mathcal Z}|A_{7,N}(z)|
=
O_p\!\left(\frac{\zeta(K)^{9/2}}{\sqrt N}\right).
$$
Therefore, under the maintained regularity conditions, 
$
\sup_{z\in\mathcal Z}|A_{7,N}(z)|=o_p(1).
$
\paragraph{Bound on (8).}
The treatment of \textup{(8)} requires a different argument from the pointwise proof in HIR. In HIR, the key step is a smooth approximation bound for a single function of the evaluation variable. Using $u$ to denote that variable, the argument there relies on a bound of the form
$$
\sup_{u\in\mathcal X}|\delta_K(u)-\delta_0(u)|\le CK^{-t/r}.
$$
In the present uniform setting, however, the relevant object is the indexed family
$$
\{\delta_0(z;\cdot):z\in\mathcal Z\},
\qquad
\delta_0(z;u)=-\mathbf 1\{u_Z\le z\}h_0(u),
$$
and what would be needed to mimic the HIR proof is a bound of the form
$$
\sup_{z\in\mathcal Z}\sup_{u\in\mathcal X}|\delta_K(z;u)-\delta_0(z;u)|\le CK^{-t/r}.
$$
Such a bound is not available in general. The reason is that, for each fixed $z$, the map $u\mapsto\delta_0(z;u)$ is no longer a smooth function, but a lower-orthant indicator multiplied by a smooth weight, so the family is indexed by a moving discontinuity. Consequently, the smooth approximation argument used in HIR does not directly extend to the present uniform case. We therefore proceed differently: we rewrite \textup{(8)} as the sum of three empirical-process terms, show that the associated classes are of VC type, and then apply a maximal inequality to obtain a uniform $o_p(1)$ bound.

Define
$$
A_{8,N}(z):=
\frac{1}{\sqrt N}\sum_{i=1}^N
\left\{
\delta_K(z;X_i)\frac{T_i-p_K(X_i)}{\sqrt{p_K(X_i)(1-p_K(X_i))}}
-
\delta_0(z;X_i)\frac{T_i-p(X_i)}{\sqrt{p(X_i)(1-p(X_i))}}
\right\}.
$$
Let
$$
Z_0(W):=\frac{T-p(X)}{\sqrt{p(X)(1-p(X))}},
\qquad
Z_K(W):=\frac{T-p_K(X)}{\sqrt{p_K(X)(1-p_K(X))}},
\qquad
W:=(Y,T,X).
$$
To align the notation with the construction in \cite{hirano2003efficient}, define
$$
h_0(u):=\frac{\mu_1(u)}{p(u)}\sqrt{p(u)(1-p(u))},
\qquad
\bar h_K(u):=\frac{\mu_1(u)}{p(u)}\sqrt{p_K(u)(1-p_K(u))},
$$
and
$$
\delta_0(z;u):=-\mathbf 1\{u_Z\le z\}h_0(u),
\qquad
\bar\delta_{0,K}(z;u):=-\mathbf 1\{u_Z\le z\}\bar h_K(u).
$$
Further, let
$$
b_K(u):=\sqrt{w_K(u)}\,R^K(u),
\qquad
w_K(u):=p_K(u)(1-p_K(u)),
\qquad
\Sigma_K:=\E\!\left[b_K(X)b_K(X)^\prime\right],
$$
and define, for any square-integrable $f$,
$$
\Pi_K f(u):=b_K(u)^\prime \Sigma_K^{-1}\E\!\left[b_K(X)f(X)\right].
$$
Then
$$
\delta_K(z;\cdot)=\Pi_K\bar\delta_{0,K}(z;\cdot).
$$
Indeed, since
$$
b_K(v)\bar\delta_{0,K}(z;v)
=
-\mathbf 1\{v_Z\le z\}\frac{\mu_1(v)}{p(v)}\,p_K(v)(1-p_K(v))\,R^K(v),
$$
we have
$$
\Pi_K\bar\delta_{0,K}(z;u)
=
-\;b_K(u)^\prime\Sigma_K^{-1}
\int
\mathbf 1\{v_Z\le z\}\frac{\mu_1(v)}{p(v)}\,p_K(v)(1-p_K(v))\,R^K(v)\,dF_0(v)=\delta_K(z;u).
$$
Accordingly, write
$$
A_{8,N}(z)=A_{8,N}^{(1)}(z)+A_{8,N}^{(2)}(z)+A_{8,N}^{(3)}(z),
$$
where
$$
A_{8,N}^{(1)}(z):=
\frac{1}{\sqrt N}\sum_{i=1}^N
\bigl(\delta_K(z;X_i)-\bar\delta_{0,K}(z;X_i)\bigr)Z_0(W_i),
$$
$$
A_{8,N}^{(2)}(z):=
\frac{1}{\sqrt N}\sum_{i=1}^N
\bigl(\bar\delta_{0,K}(z;X_i)-\delta_0(z;X_i)\bigr)Z_0(W_i),
$$
$$
A_{8,N}^{(3)}(z):=
\frac{1}{\sqrt N}\sum_{i=1}^N
\delta_K(z;X_i)\bigl(Z_K(W_i)-Z_0(W_i)\bigr).
$$
Define the associated classes
$$
\eta_{1,z}(W):=\bigl(\delta_K(z;X)-\bar\delta_{0,K}(z;X)\bigr)Z_0(W),
$$
$$
\eta_{2,z}(W):=\bigl(\bar\delta_{0,K}(z;X)-\delta_0(z;X)\bigr)Z_0(W),
$$
$$
\eta_{3,z}(W):=\delta_K(z;X)\bigl(Z_K(W)-Z_0(W)\bigr),
$$
and
$$
\mathcal H_{j,K}:=\{\eta_{j,z}:z\in\mathcal Z\},
\qquad j=1,2,3.
$$

We first establish the covering number for $\mathcal H_{1,K}$. Let
$$
\mathbb K_K(u,v):=b_K(u)^\prime\Sigma_K^{-1}b_K(v),
\qquad
g_z(t):=\mathbf 1\{t\le z\},
\qquad
\mathcal G_Z:=\{g_z:z\in\mathcal Z\}.
$$
Since $\mathcal G_Z$ is the lower-orthant class on $\mathcal Z$, it is a VC class. Hence there exist constants $A_0,v_0>0$, depending only on $d_z$, such that for every probability measure $\mu$ on $\mathcal Z$ and every $0<\varepsilon\le1$,
$$
N\bigl(\mathcal G_Z,L_2(\mu),\varepsilon\bigr)
\le
\left(\frac{A_0}{\varepsilon}\right)^{v_0}.
$$
By linearity of $\Pi_K$,
$$
\delta_K(z;u)
=
-\int \mathbb K_K(u,v)\,\mathbf 1\{v_Z\le z\}\,\bar h_K(v)\,dF_0(v),
$$
and therefore
$$
\eta_{1,z}(W)
=
Z_0(W)\left[
\mathbf 1\{Z\le z\}\bar h_K(X)
-
\int \mathbb K_K(X,v)\,\mathbf 1\{v_Z\le z\}\,\bar h_K(v)\,dF_0(v)
\right].
$$
Next define
\begin{equation}
\label{eq:envelope1}
  S_{K}(u):=\int |\mathbb K_K(u,v)|\,|\bar h_K(v)|\,dF_0(v),
\qquad
F_{1,K}(W):=|Z_0(W)|\bigl(|\bar h_K(X)|+S_{K}(X)\bigr).  
\end{equation}
Then $F_{1,K}$ is an envelope for $\mathcal H_{1,K}$.

Fix any probability measure $Q$ on the sample space of $W$. For $z,z'\in\mathcal Z$, let $\Delta_{z,z'}:=g_z-g_{z'}$. Then
$$
\eta_{1,z}(W)-\eta_{1,z'}(W)
=
Z_0(W)\left[
\bar h_K(X)\Delta_{z,z'}(Z)
-
\int \mathbb K_K(X,v)\,\bar h_K(v)\Delta_{z,z'}(v_Z)\,dF_0(v)
\right].
$$
Hence, by $(a-b)^2\le 2a^2+2b^2$ and Cauchy--Schwarz,
\begin{align*}
\|\eta_{1,z}-\eta_{1,z'}\|_{Q,2}^2
&\le
2\int Z_0(W)^2 \bar h_K(X)^2 |\Delta_{z,z'}(Z)|^2\,dQ(W) \\
&\quad
+
2\int Z_0(W)^2 S_{K}(X)
\left(
\int |\mathbb K_K(X,v)|\,|\bar h_K(v)|\,|\Delta_{z,z'}(v_Z)|^2\,dF_0(v)
\right)dQ(W).
\end{align*}
Define finite measures $\nu_{1,Q}$ and $\nu_{2,Q}$ on $\mathcal Z$ by
$$
\nu_{1,Q}(A):=
\int Z_0(W)^2 \bar h_K(X)^2\mathbf 1\{Z\in A\}\,dQ(W),
$$
$$
\nu_{2,Q}(A):=
\int Z_0(W)^2 S_{K}(X)
\left(
\int |\mathbb K_K(X,v)|\,|\bar h_K(v)|\,\mathbf 1\{v_Z\in A\}\,dF_0(v)
\right)dQ(W).
$$
Let $\nu_Q:=\nu_{1,Q}+\nu_{2,Q}$. Then
$$
\|\eta_{1,z}-\eta_{1,z'}\|_{Q,2}^2
\le
2\int |\Delta_{z,z'}(t)|^2\,d\nu_Q(t).
$$
Moreover,
$$
\nu_Q(\mathcal Z)
=
\int Z_0(W)^2\bigl(\bar h_K(X)^2+S_{K}(X)^2\bigr)\,dQ(W)
\le
\|F_{1,K}\|_{Q,2}^2.
$$
If $\nu_Q(\mathcal Z)=0$, the claim is immediate. Otherwise, define the probability measure $\widetilde Q:=\nu_Q/\nu_Q(\mathcal Z)$. Then
$$
\|\eta_{1,z}-\eta_{1,z'}\|_{Q,2}
\le
\sqrt{2}\,\|F_{1,K}\|_{Q,2}\,\|g_z-g_{z'}\|_{L_2(\widetilde Q)}.
$$
Therefore,
$$
N\!\left(
\mathcal H_{1,K},
L_2(Q),
\sqrt{2}\,\varepsilon\,\|F_{1,K}\|_{Q,2}
\right)
\le
N\!\left(
\mathcal G_Z,
L_2(\widetilde Q),
\varepsilon
\right)
\le
\left(\frac{A_0}{\varepsilon}\right)^{v_0},
\qquad 0<\varepsilon\le1.
$$
Absorbing the factor $\sqrt{2}$ into the generic constant yields
$$
N\!\left(
\mathcal H_{1,K},
L_2(Q),
\varepsilon\,\|F_{1,K}\|_{Q,2}
\right)
\le
\left(\frac{A}{\varepsilon}\right)^v,
\qquad 0<\varepsilon\le1,
$$
for some constants $A,v>0$. Thus, $\mathcal H_{1,K}$ is of VC type.

For $\mathcal H_{2,K}$, note that
$
\eta_{2,z}(W)
=
-\mathbf 1\{Z\le z\}\bigl(\bar h_K(X)-h_0(X)\bigr)Z_0(W).
$
Hence $\mathcal H_{2,K}$ is a weighted lower-orthant class. With envelope
\begin{equation}
\label{eq:envelope2}
    F_{2,K}(W):=|\bar h_K(X)-h_0(X)|\,|Z_0(W)|,
\end{equation}
the same VC bound follows immediately from that of $\mathcal G_Z$:
$$
N\!\left(
\mathcal H_{2,K},
L_2(Q),
\varepsilon\,\|F_{2,K}\|_{Q,2}
\right)
\le
\left(\frac{A}{\varepsilon}\right)^v,
\qquad 0<\varepsilon\le1.
$$

Finally, for $\mathcal H_{3,K}$, write
$$
\eta_{3,z}(W)
=
-\bigl(Z_K(W)-Z_0(W)\bigr)
\int \mathbb K_K(X,v)\,\mathbf 1\{v_Z\le z\}\,\bar h_K(v)\,dF_0(v).
$$
Thus $\mathcal H_{3,K}$ has the same structure as $\mathcal H_{1,K}$, with $Z_0$ replaced by $Z_K-Z_0$. Defining
$
F_{3,K}(W):=|Z_K(W)-Z_0(W)|\,S_{K}(X),
$
the same argument as for $\mathcal H_{1,K}$ gives
$$
N\!\left(
\mathcal H_{3,K},
L_2(Q),
\varepsilon\,\|F_{3,K}\|_{Q,2}
\right)
\le
\left(\frac{A}{\varepsilon}\right)^v,
\qquad 0<\varepsilon\le1.
$$
Hence each of the three classes $\mathcal H_{1,K}$, $\mathcal H_{2,K}$, and $\mathcal H_{3,K}$ is of VC type, with complexity inherited from the lower-orthant class. Since the corresponding envelopes are $K$-dependent and need not remain uniformly bounded, a direct application of Pollard's maximal inequality (see Theorem 2.14.1 in \cite{van1996weak}) is not well suited. We therefore invoke Corollary 5.1 of \cite{chernozhukov2014gaussian} (Hereafter CCK14), which applies to VC-type classes with square-integrable envelopes.

We first treat $\mathcal H_{1,K}$. Recall that
$$
\eta_{1,z}(W)
=
\bigl(\delta_K(z;X)-\bar\delta_{0,K}(z;X)\bigr)Z_0(W),
\qquad
\mathcal H_{1,K}:=\{\eta_{1,z}:z\in\mathcal Z\}.
$$
Since $\E[Z_0(W)\mid X]=0$, we have $\E[\eta_{1,z}(W)]=0$ for every $z\in\mathcal Z$, and therefore
$
\sup_{z\in\mathcal Z}|A_{8,N}^{(1)}(z)|
=
\|\mathbb G_N\|_{\mathcal H_{1,K}}.
$ We begin with the variance term required by Corollary 5.1 of CCK 14. By definition,
\begin{align*}
\sup_{z\in\mathcal Z}P\eta_{1,z}^2
&=
\sup_{z\in\mathcal Z}
\E\Bigl[
\bigl(\delta_K(z;X)-\bar\delta_{0,K}(z;X)\bigr)^2 Z_0(W)^2
\Bigr].
\end{align*}
Under overlap, $Z_0$ is uniformly bounded, so
$
\sup_{z\in\mathcal Z}P\eta_{1,z}^2
\le
C\sup_{z\in\mathcal Z}
\|\delta_K(z;\cdot)-\bar\delta_{0,K}(z;\cdot)\|_{L_2(F_0)}^2.
$
We therefore set
$
r_{1,K}
:=
\sup_{z\in\mathcal Z}
\|\delta_K(z;\cdot)-\bar\delta_{0,K}(z;\cdot)\|_{L_2(F_0)}.
$
To bound $r_{1,K}$, write
$$
\delta_K(z;\cdot)-\bar\delta_{0,K}(z;\cdot)
=
\bigl(\delta_K(z;\cdot)-\delta_0(z;\cdot)\bigr)
+
\bigl(\delta_0(z;\cdot)-\bar\delta_{0,K}(z;\cdot)\bigr).
$$
Since $\delta_K(z;\cdot)=\Pi_K\bar\delta_{0,K}(z;\cdot)$ and $\Pi_K$ is a contraction on $L_2(F_0)$,
\begin{align*}
\|\delta_K(z;\cdot)-\bar\delta_{0,K}(z;\cdot)\|_{L_2(F_0)}
&=
\|\Pi_K\bar\delta_{0,K}(z;\cdot)-\bar\delta_{0,K}(z;\cdot)\|_{L_2(F_0)} \\
&\le\|\Pi_K\bar\delta_{0,K}(z;\cdot)-\Pi_K\delta_{0}(z;\cdot)\|_{L_2(F_0)}+\|\Pi_K\delta_0(z;\cdot)-\delta_0(z;\cdot)\|_{L_2(F_0)}\\
&\quad +\|\delta_0(z;\cdot)-\bar\delta_{0,K}(z;\cdot)\|_{L_2(F_0)}\\
&\le
\|\Pi_K\delta_0(z;\cdot)-\delta_0(z;\cdot)\|_{L_2(F_0)}
+
2\|\bar\delta_{0,K}(z;\cdot)-\delta_0(z;\cdot)\|_{L_2(F_0)}.
\end{align*}
We now bound these two terms. First, since
$$
\delta_0(z;u)=-\mathbf 1\{u_Z\le z\}h_0(u),
\qquad
\bar\delta_{0,K}(z;u)=-\mathbf 1\{u_Z\le z\}\bar h_K(u),
$$
we have
$$
\bar\delta_{0,K}(z;u)-\delta_0(z;u)
=
-\mathbf 1\{u_Z\le z\}\bigl(\bar h_K(u)-h_0(u)\bigr),
$$
and therefore
$$
\|\bar\delta_{0,K}(z;\cdot)-\delta_0(z;\cdot)\|_{L_2(F_0)}\le \|\bar h_K-h_0\|_{L_2(F_0)}\le \|\bar h_K-h_0\|_\infty.
$$
Now
$$
h_0(u)=\frac{\mu_1(u)}{p(u)}\sqrt{p(u)(1-p(u))},
\qquad
\bar h_K(u)=\frac{\mu_1(u)}{p(u)}\sqrt{p_K(u)(1-p_K(u))}.
$$
Since the map $p\mapsto \sqrt{p(1-p)}$ is Lipschitz on any compact subinterval of $(0,1)$, overlap and boundedness of $\mu_1/p$ imply $\|\bar h_K-h_0\|_\infty \le C\|p_K-p\|_\infty$. Thus,
$$
\sup_{z\in\mathcal Z}\|\bar\delta_{0,K}(z;\cdot)-\delta_0(z;\cdot)\|_{L_2(F_0)}=O(\|p_K-p\|_\infty)=O(\zeta(K)K^{-s/2r}).
$$

The first term is the genuine projection approximation error for the lower-orthant weighted class, we use the crude bound
$$
\sup_{z\in\mathcal Z}
\|\Pi_K\delta_0(z;\cdot)-\delta_0(z;\cdot)\|_{L_2(F_0)}=O\!\left(K^{-1/2r}\right).
$$
Combining the two displays above gives
$$
r_{1,K}
=
O\!\left(K^{-1/2r}+\zeta(K)K^{-s/2r}\right)=O(K^{-1/2r}),
$$
and hence, in Corollary 5.1, $\sigma_{1,K}^2=O(K^{-1/r})$.

We next turn to the envelope. Recall \eqref{eq:envelope1}, an envelope for $\mathcal H_{1,K}$ is
$
F_{1,K}(W)
:=
|Z_0(W)|\bigl(S_{K}(X)+|\bar h_K(X)|\bigr).
$ We now compute its order. Under overlap, $|Z_0(W)|\le C$. Further, compact support and overlap imply
$
\sup_{u\in\mathcal X}|\bar h_K(u)|\le C.
$
Next, using
$
|\mathbb K_K(u,v)|
\le
\|b_K(u)\|\,\|\Sigma_K^{-1}\|\,\|b_K(v)\|,
$
we obtain
$$
 S_{K}(u)
\le
\|b_K(u)\|\,\|\Sigma_K^{-1}\|
\int \|b_K(v)\|\,|\bar h_K(v)|\,dF_0(v).
$$
Since $\sup_{u\in\mathcal X}\|b_K(u)\|\lesssim \zeta(K)$ and $\sup_{u\in\mathcal X}|\bar h_K(u)|\le C$, it follows that
$
\sup_{u\in\mathcal X}S_{K}(u)\lesssim \zeta(K)^2.
$
Consequently,
$
F_{1,K}(W)\lesssim \zeta(K)^2
\quad\text{a.s.},
$
and hence
$
\|F_{1,K}\|_{P,2}=O(\zeta(K)^2)=O(K^2).
$

Let
$
M_{1,K}:=\max_{1\le i\le N}F_{1,K}(W_i).
$
Since $F_{1,K}(W)\lesssim \zeta(K)^2$ almost surely, we immediately have $\|M_{1,K}\|_2=O(\zeta(K)^2)=O(K^2)$.

We are now in a position to apply Corollary 5.1 of CCK14. Since $\mathcal H_{1,K}$ is of VC type, there exist constants $A,v>0$ such that
$$
\sup_Q
N\!\left(
\mathcal H_{1,K},
e_Q,
\varepsilon\|F_{1,K}\|_{Q,2}
\right)
\le
\left(\frac{A}{\varepsilon}\right)^v,
\qquad
0<\varepsilon\le 1.
$$
Hence
\begin{align*}
\E\bigl[\|\mathbb G_N\|_{\mathcal H_{1,K}}\bigr]
&\lesssim
\sqrt{
v\,\sigma_{1,K}^2
\log\!\left(\frac{A\|F_{1,K}\|_{P,2}}{\sigma_{1,K}}\right)
}
+
\frac{v\|M_{1,K}\|_2}{\sqrt N}
\log\!\left(\frac{A\|F_{1,K}\|_{P,2}}{\sigma_{1,K}}\right).
\end{align*}
Substituting the bounds just derived yields
\begin{align*}
\E\bigl[\|\mathbb G_N\|_{\mathcal H_{1,K}}\bigr]
&\lesssim
\sqrt{
v\,r_{1,K}^2
\log\!\left(\frac{A\zeta(K)^2}{r_{1,K}}\right)
}
+
\frac{v\zeta(K)^2}{\sqrt N}
\log\!\left(\frac{A\zeta(K)^2}{r_{1,K}}\right) \\
&\lesssim
r_{1,K}\sqrt{\log\!\left(\frac{A\zeta(K)^2}{r_{1,K}}\right)}
+
\frac{\zeta(K)^2}{\sqrt N}
\log\!\left(\frac{A\zeta(K)^2}{r_{1,K}}\right)\\
&\lesssim
K^{-1/(2r)}
\sqrt{
\log\!\left(
A K^{2+1/(2r)}
\right)
}
+
\frac{K^2}{\sqrt N}
\log\!\left(
A K^{2+1/(2r)}
\right) \\
&\lesssim
K^{-1/(2r)}\sqrt{\log K}
+
\frac{K^2}{\sqrt N}\log K.
\end{align*}
The first term converges to zero whenever $K\to\infty$. The second converges to zero provided $K^2\log K/\sqrt{N}\to0$, which is valid under the assumption  $K=K_N\asymp N^\vartheta$ with $\vartheta<1/9$. Therefore, $\E\bigl[\|\mathbb G_N\|_{\mathcal H_{1,K}}\bigr]\to 0$,
an application of Markov's inequality gives $\sup_{z\in\mathcal Z}|A_{8,N}^{(1)}(z)|=o_p(1)$.

We next turn to $\mathcal H_{2,K}$. Recall that
$$
\eta_{2,z}(W)
=
\bigl(\bar\delta_{0,K}(z;X)-\delta_0(z;X)\bigr)Z_0(W),
\qquad
\mathcal H_{2,K}:=\{\eta_{2,z}:z\in\mathcal Z\}.
$$
Since $\E[Z_0(W)\mid X]=0$, we have $\E[\eta_{2,z}(W)]=0$ for every $z\in\mathcal Z$, and therefore
$
\sup_{z\in\mathcal Z}|A_{8,N}^{(2)}(z)|
=
\|\mathbb G_N\|_{\mathcal H_{2,K}}.
$
We first bound the variance term. Since $|Z_0(W)|\le C$ under overlap,
$$
\sup_{z\in\mathcal Z}P\eta_{2,z}^2=\sup_{z\in\mathcal Z}
\E\Bigl[
\mathbf 1\{Z\le z\}\bigl(\bar h_K(X)-h_0(X)\bigr)^2 Z_0(W)^2
\Bigr]\le C\,\E\Bigl[\bigl(\bar h_K(X)-h_0(X)\bigr)^2\Bigr].
$$
Since $\|\bar h_K-h_0\|_\infty\le C\|p_K-p\|_\infty$ and $\|p_K-p\|_\infty=O\!\left(\zeta(K)K^{-s/2r}\right)$, we may therefore take $\sigma_{2,K}^2=O\!\left(\zeta(K)^2K^{-s/r}\right)$.

As shown in \eqref{eq:envelope2}, an envelope for $\mathcal H_{2,K}$ is
$F_{2,K}(W):=|\bar h_K(X)-h_0(X)|\,|Z_0(W)|$.
Again using the bound above,
$$
F_{2,K}(W)\le C\|\bar h_K-h_0\|_\infty
\le C\|p_K-p\|_\infty
\qquad\text{a.s.},
$$
so that $\|F_{2,K}\|_{P,2}
=
O\!\left(\|p_K-p\|_\infty\right)
=
O\!\left(\zeta(K)K^{-s/2r}\right)$.
Let $M_{2,K}:=\max_{1\le i\le N}F_{2,K}(W_i)$, then the same bound gives
$\|M_{2,K}\|_2
=
O\!\left(\|p_K-p\|_\infty\right)
=
O\!\left(\zeta(K)K^{-s/2r}\right).
$

Since $\mathcal H_{2,K}$ is of VC type, Corollary 5.1 of CCK14 implies
\begin{align*}
E\bigl[\|\mathbb G_N\|_{\mathcal H_{2,K}}\bigr]
&\lesssim
\sqrt{
v\,\sigma_{2,K}^2
\log\!\left(\frac{A\|F_{2,K}\|_{P,2}}{\sigma_{2,K}}\right)
}
+
\frac{v\|M_{2,K}\|_2}{\sqrt N}
\log\!\left(\frac{A\|F_{2,K}\|_{P,2}}{\sigma_{2,K}}\right).
\end{align*}
Since $\sigma_{2,K}$, $\|F_{2,K}\|_{P,2}$ and $\|M_{2,K}\|_2$ are all of the same order,
$$
\log\!\left(\frac{A\|F_{2,K}\|_{P,2}}{\sigma_{2,K}}\right)=O(1),
$$
and therefore
$$
\E\bigl[\|\mathbb G_N\|_{\mathcal H_{2,K}}\bigr]
\lesssim
\sigma_{2,K}
+
\frac{\|M_{2,K}\|_2}{\sqrt N}
=
O\!\left(\zeta(K)K^{-s/2r}\right).
$$
For power series, $\zeta(K)\asymp K$, so
$$
\E\bigl[\|\mathbb G_N\|_{\mathcal H_{2,K}}\bigr]
=
O\!\left(K^{1-s/2r}\right)\to 0,
$$
since $s/r\ge 7$. Hence, by Markov's inequality, $\sup_{z\in\mathcal Z}|A_{8,N}^{(2)}(z)|=o_p(1)$.

We finally consider $\mathcal H_{3,K}$. Recall that
$$
\eta_{3,z}(W)
=
\delta_K(z;X)\bigl(Z_K(W)-Z_0(W)\bigr),
\qquad
\mathcal H_{3,K}:=\{\eta_{3,z}:z\in\mathcal Z\}.
$$
We first note that this class is centered. Indeed, since $\delta_K(z;\cdot)\in \mathrm{span}\{b_K(\cdot)\}$, there exists a coefficient vector $\alpha_K(z)$ such that
$$
\delta_K(z;u)=b_K(u)^\prime \alpha_K(z),
\qquad
b_K(u)=\sqrt{p_K(u)(1-p_K(u))}\,R^K(u).
$$
Moreover,
$$
\E[Z_K(W)-Z_0(W)\mid X]
=
\frac{p(X)-p_K(X)}{\sqrt{p_K(X)(1-p_K(X))}}.
$$
Hence
$$
\E[\eta_{3,z}(W)]=\E\!\left[
\delta_K(z;X)\frac{p(X)-p_K(X)}{\sqrt{p_K(X)(1-p_K(X))}}
\right] = \alpha_K(z)^\prime \E\!\left[R^K(X)\bigl(p(X)-p_K(X)\bigr)\right].
$$

By the first-order condition defining the pseudo-true SLE, $\E\!\left[R^K(X)\bigl(T-p_K(X)\bigr)\right]=0$, and therefore $\E\!\left[R^K(X)\bigl(p(X)-p_K(X)\bigr)\right]=0$. It follows that $\E[\eta_{3,z}(W)]=0$ for every $z\in\mathcal Z$, and thus $\sup_{z\in\mathcal Z}|A_{8,N}^{(3)}(z)|=
\|\mathbb G_N\|_{\mathcal H_{3,K}}$.

We next bound the variance term. 
As in the discussion above, under overlap,
$$
|Z_K(W)-Z_0(W)|\le C\,|p_K(X)-p(X)|\le C\|p_K-p\|_\infty.
$$
Further, from the representation
$$
\delta_K(z;u)
=
-\int \mathbb K_K(u,v)\,\mathbf 1\{v_Z\le z\}\,\bar h_K(v)\,dF_0(v),
$$
we have
$$
|\delta_K(z;u)|\le S_K(u),
\qquad
S_K(u):=\int |\mathbb K_K(u,v)|\,|\bar h_K(v)|\,dF_0(v).
$$
Hence
$$
\sup_{z\in\mathcal Z}P\eta_{3,z}^2\le\E\!\left[(Z_K(W)-Z_0(W))^2S_K(X)^2\right]\le C\|p_K-p\|_\infty^2\E[S_K(X)^2].
$$
Using the previously established bound $S_K(u)\lesssim \zeta(K)^2$, we obtain
$$
\sup_{z\in\mathcal Z}P\eta_{3,z}^2
=
O\!\left(\zeta(K)^4\|p_K-p\|_\infty^2\right).
$$
Thus, in Corollary 5.1 of CCK14, we may take
$
\sigma_{3,K}^2
=
O\!\left(\zeta(K)^4\|p_K-p\|_\infty^2\right).
$

An envelope for $\mathcal H_{3,K}$ is
$
F_{3,K}(W):=|Z_K(W)-Z_0(W)|\,S_K(X).
$
By the same bounds, $F_{3,K}(W)\lesssim\zeta(K)^2\|p_K-p\|_\infty$, so that $\|F_{3,K}\|_{P,2}=O\!\left(\zeta(K)^2\|p_K-p\|_\infty\right).
$
Likewise, 
$
\|M_{3,K}\|_2
=
O\!\left(\zeta(K)^2\|p_K-p\|_\infty\right).
$

Since $\mathcal H_{3,K}$ is of VC type, Corollary 5.1 of CCK14 yields
\begin{align*}
E\bigl[\|\mathbb G_N\|_{\mathcal H_{3,K}}\bigr]
&\lesssim
\sqrt{
v\,\sigma_{3,K}^2
\log\!\left(\frac{A\|F_{3,K}\|_{P,2}}{\sigma_{3,K}}\right)
}
+
\frac{v\|M_{3,K}\|_2}{\sqrt N}
\log\!\left(\frac{A\|F_{3,K}\|_{P,2}}{\sigma_{3,K}}\right).
\end{align*}
Since $\sigma_{3,K}$, $\|F_{3,K}\|_{P,2}$, and $\|M_{3,K}\|_2$ are of the same order, the logarithmic factor is $O(1)$, and hence
$$
\E\bigl[\|\mathbb G_N\|_{\mathcal H_{3,K}}\bigr]
=
O\!\left(\zeta(K)^2\|p_K-p\|_\infty\right)=O\!\left(K^{3-s/2r}\right)\to 0,
$$
since $s/r\ge 7$. Therefore, by Markov's inequality, $\sup_{z\in\mathcal Z}|A_{8,N}^{(3)}(z)|=o_p(1)$.  Combining this with the bounds for $A_{8,N}^{(1)}(z)$ and $A_{8,N}^{(2)}(z)$ established above, we finally conclude that
$
\sup_{z\in\mathcal Z}|A_{8,N}(z)|=o_p(1).
$

Combining the bounds for \textup{(4)}--\textup{(8)} established above, we conclude that each of these terms is $o_p(1)$ uniformly over $z\in\mathcal Z$. Therefore, by Slutsky's theorem, the desired uniform asymptotic linearization follows.
\end{proof}
\newpage
\section{Additional Simulations for the Extensions in Section \ref{sec:extension}}
\label{sec:app-extension-sim}

This section reports additional simulation results for the two extensions discussed in Section \ref{sec:extension}. Subsection \ref{sec:app-parametric-sim} studies the finite-sample performance of the parametric specification tests for the CATE, while Subsection \ref{sec:app-late} considers the CLATE heterogeneity test in settings where treatment assignment is endogenous but a binary instrument is available.

\subsection{Additional Simulations for Parametric-Form Tests}
\label{sec:app-parametric-sim}

This subsection reports additional Monte Carlo results for the parametric-form extensions discussed in Section \ref{sec:ex-para}. We consider two specification tests: a linear-form test of
$$
\mathbb H_0^{\mathrm{lin}}:\ \tau(x)=\theta_0+\theta_1x,
$$
and a quadratic-form test of
$$
\mathbb H_0^{\mathrm{quad}}:\ \tau(x)=\theta_0+\theta_1x+\theta_2x^2,
$$
where $x=X_1$. The treatment assignment mechanisms, covariate dimensions, nuisance estimation, and multiplier bootstrap implementation are the same as in the main simulation design. Thus, the purpose of this appendix is not to introduce new assignment designs, but to examine whether the proposed projected ICM processes correctly detect misspecification of low-dimensional parametric CATE forms.

Across all designs, the outcome is generated as
$$
Y=\mu_0(X)+\tau(X_1)D+\varepsilon,
\qquad \tau_0=0.5,
$$
with $\varepsilon\sim \mathcal{N}(0,\sigma_u^2)$ with $\sigma_u=0.25$. We vary only the functional form of $\tau(\cdot)$. Let $u=x-0.5$. The seven CATE specifications are
$$
\begin{array}{ll}
\text{constant:} & \tau(x)=\tau_0,\\
\text{linear:} & \tau(x)=\tau_0+\delta u,\\
\text{cubic:} & \tau(x)=\tau_0+\delta(0.5u+6u^3),\\
\text{linear threshold:} & \tau(x)=\tau_0+1.5\delta u\,\1\{x>0.5\},\\
\text{quadratic:} & \tau(x)=\tau_0+0.3\delta u-2\delta u^2,\\
\text{quadratic exponential:} &
\tau(x)=\tau_0+0.3\delta u-\delta u^2
+0.8\delta\{\exp(3u)-1-3u\},\\
\text{quadratic bump:} &
\tau(x)=\tau_0+0.3\delta u-\delta u^2+0.8\delta\exp(-40u^2).
\end{array}
$$
Figure \ref{fig:cate-forms-parametric} plots these seven functions.

\begin{figure}[!htbp]
    \centering
    \includegraphics[width=0.92\textwidth]{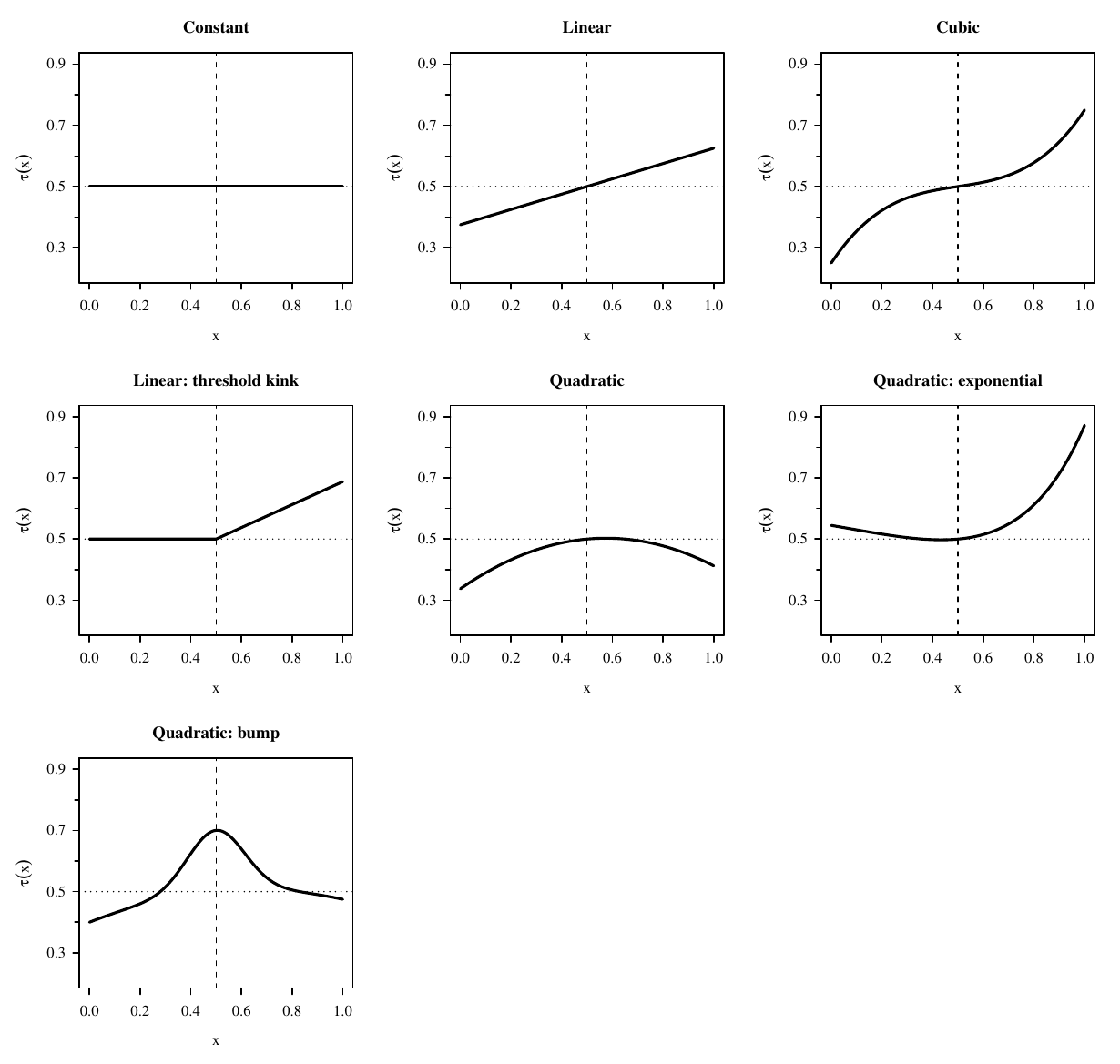}
    \caption{CATE functions used in the parametric-form simulation designs. The horizontal dotted line marks the baseline level $\tau_0=0.5$, and the vertical dashed line marks $x=0.5$.}
    \label{fig:cate-forms-parametric}
\end{figure}

For the linear specification test, the constant and linear designs satisfy $H_0^{\mathrm{lin}}$ and are therefore used to evaluate size, while the others are alternatives and are used to evaluate power. For the quadratic specification test, the constant, linear, and quadratic designs satisfy $H_0^{\mathrm{quad}}$, whereas the others violate the null. In each case, $\hat\theta$ is estimated by least squares using the feasible doubly robust pseudo-outcome, and inference is based on the multiplier bootstrap process described in equation \eqref{eq:bootstrap_process_para} in Section \ref{sec:ex-para}.

\begin{table}[!htbp]
\centering
\scriptsize
\setlength{\tabcolsep}{5.0pt}
\renewcommand{\arraystretch}{0.96}
\begin{threeparttable}
\caption{Rejection frequencies for the CvM linear-form test. Constant and Linear correspond to size, and the remaining CATE forms to power.}
\label{tab:linear-form-cvm}
\begin{tabular}{ll l *{9}{c}}
\toprule
& & & \multicolumn{3}{c}{$n=500$} & \multicolumn{3}{c}{$n=1000$} & \multicolumn{3}{c}{$n=2000$} \\
\cmidrule(lr){4-6}\cmidrule(lr){7-9}\cmidrule(lr){10-12}
Family & $d$ & CATE form & 1\% & 5\% & 10\% & 1\% & 5\% & 10\% & 1\% & 5\% & 10\% \\
\midrule
\multirow{21}{*}{Smooth logit} & \multirow{7}{*}{3} & Constant  & 0.7 & 4.7 & 9.2 & 0.9 & 5.4 & 11.1 & 0.9 & 4.2 & 10.5 \\
 &  & Linear  & 1.1 & 5.8 & 11.7 & 0.5 & 5.2 & 10.0 & 1.0 & 4.7 & 10.0 \\
 &  & Cubic  & 2.3 & 11.0 & 21.6 & 3.1 & 15.9 & 25.7 & 12.0 & 34.3 & 47.7 \\
 &  & Linear threshold  & 7.6 & 18.4 & 27.8 & 16.1 & 32.8 & 44.1 & 36.6 & 57.4 & 70.4 \\
 &  & Quadratic  & 13.1 & 31.5 & 40.9 & 33.2 & 56.2 & 67.6 & 68.5 & 85.4 & 91.0 \\
 &  & Quadratic exponential  & 41.3 & 63.2 & 75.4 & 81.8 & 93.0 & 96.2 & 99.1 & 99.9 & 100.0 \\
 &  & Quadratic bump  & 74.3 & 89.6 & 94.6 & 98.4 & 99.5 & 99.9 & 100.0 & 100.0 & 100.0 \\
\addlinespace[0.20em]
 & \multirow{7}{*}{5} & Constant  & 1.3 & 5.5 & 10.6 & 0.6 & 4.2 & 9.4 & 1.1 & 4.4 & 9.4 \\
 &  & Linear  & 1.3 & 5.7 & 11.4 & 1.0 & 4.9 & 10.9 & 0.9 & 4.4 & 8.7 \\
 &  & Cubic  & 1.9 & 10.4 & 19.1 & 4.4 & 18.8 & 30.6 & 10.2 & 32.9 & 47.9 \\
 &  & Linear threshold  & 7.6 & 17.5 & 26.7 & 14.1 & 30.4 & 43.8 & 38.8 & 60.9 & 72.6 \\
 &  & Quadratic  & 13.9 & 31.6 & 42.9 & 31.4 & 57.5 & 68.2 & 68.7 & 85.7 & 90.7 \\
 &  & Quadratic exponential  & 41.7 & 62.4 & 72.9 & 81.6 & 92.8 & 96.4 & 99.5 & 100.0 & 100.0 \\
 &  & Quadratic bump  & 79.1 & 91.7 & 95.1 & 99.0 & 99.9 & 99.9 & 100.0 & 100.0 & 100.0 \\
\addlinespace[0.20em]
 & \multirow{7}{*}{10} & Constant  & 1.1 & 5.6 & 10.0 & 1.4 & 5.7 & 12.1 & 0.9 & 5.6 & 11.7 \\
 &  & Linear  & 0.6 & 4.6 & 9.1 & 0.9 & 4.1 & 9.0 & 1.1 & 6.2 & 11.0 \\
 &  & Cubic  & 2.0 & 10.6 & 20.2 & 4.0 & 17.6 & 30.6 & 10.2 & 34.6 & 48.9 \\
 &  & Linear threshold  & 4.7 & 14.9 & 22.9 & 15.4 & 34.8 & 45.1 & 35.9 & 59.1 & 70.8 \\
 &  & Quadratic  & 12.7 & 30.8 & 43.2 & 31.9 & 56.0 & 66.4 & 68.3 & 84.7 & 90.0 \\
 &  & Quadratic exponential  & 40.2 & 61.3 & 73.6 & 81.7 & 93.9 & 96.8 & 99.4 & 100.0 & 100.0 \\
 &  & Quadratic bump  & 75.0 & 89.5 & 94.0 & 98.8 & 99.5 & 99.8 & 100.0 & 100.0 & 100.0 \\
\midrule
\multirow{21}{*}{Threshold} & \multirow{7}{*}{3} & Constant  & 0.4 & 4.2 & 8.5 & 1.9 & 6.5 & 11.7 & 1.4 & 5.2 & 11.3 \\
 &  & Linear  & 1.4 & 5.6 & 10.6 & 0.6 & 4.6 & 9.5 & 1.3 & 6.2 & 10.8 \\
 &  & Cubic  & 2.4 & 9.2 & 16.8 & 3.8 & 14.8 & 25.2 & 8.4 & 28.4 & 43.8 \\
 &  & Linear threshold  & 6.8 & 18.7 & 27.6 & 13.4 & 31.8 & 43.5 & 32.8 & 56.4 & 67.8 \\
 &  & Quadratic  & 10.3 & 27.4 & 37.6 & 27.4 & 48.4 & 61.2 & 58.7 & 79.2 & 86.5 \\
 &  & Quadratic exponential  & 35.4 & 57.6 & 70.6 & 73.3 & 88.8 & 93.1 & 98.1 & 99.6 & 99.7 \\
 &  & Quadratic bump  & 68.3 & 86.3 & 91.2 & 96.6 & 99.1 & 99.6 & 100.0 & 100.0 & 100.0 \\
\addlinespace[0.20em]
 & \multirow{7}{*}{5} & Constant  & 1.3 & 5.6 & 11.2 & 1.6 & 6.0 & 13.4 & 1.8 & 5.5 & 9.1 \\
 &  & Linear  & 1.0 & 4.9 & 10.5 & 1.5 & 4.9 & 9.4 & 0.8 & 4.1 & 8.9 \\
 &  & Cubic  & 1.8 & 8.2 & 15.4 & 4.0 & 13.9 & 24.4 & 9.3 & 26.8 & 41.9 \\
 &  & Linear threshold  & 3.6 & 12.5 & 21.4 & 10.8 & 26.8 & 36.6 & 30.0 & 53.2 & 65.2 \\
 &  & Quadratic  & 9.9 & 26.1 & 37.0 & 24.9 & 45.7 & 59.5 & 57.7 & 77.0 & 85.5 \\
 &  & Quadratic exponential  & 25.7 & 46.7 & 57.4 & 72.4 & 88.0 & 92.6 & 97.7 & 99.3 & 99.8 \\
 &  & Quadratic bump  & 62.0 & 81.8 & 87.5 & 95.4 & 98.4 & 99.7 & 100.0 & 100.0 & 100.0 \\
\addlinespace[0.20em]
 & \multirow{7}{*}{10} & Constant  & 1.3 & 5.1 & 10.5 & 1.9 & 5.4 & 10.5 & 0.9 & 5.4 & 10.3 \\
 &  & Linear  & 1.1 & 5.7 & 10.1 & 0.9 & 4.2 & 8.9 & 1.2 & 5.4 & 11.3 \\
 &  & Cubic  & 2.3 & 8.8 & 14.7 & 3.1 & 14.2 & 23.8 & 10.1 & 29.3 & 42.0 \\
 &  & Linear threshold  & 5.8 & 14.6 & 20.9 & 11.6 & 28.6 & 40.3 & 29.2 & 51.7 & 63.4 \\
 &  & Quadratic  & 9.4 & 22.7 & 33.2 & 27.1 & 47.6 & 59.9 & 59.3 & 79.4 & 85.6 \\
 &  & Quadratic exponential  & 26.5 & 50.2 & 60.9 & 71.3 & 87.2 & 91.8 & 97.4 & 99.5 & 99.8 \\
 &  & Quadratic bump  & 58.5 & 78.8 & 87.3 & 95.3 & 98.7 & 99.4 & 100.0 & 100.0 & 100.0 \\
\midrule
\multirow{21}{*}{Fractional} & \multirow{7}{*}{3} & Constant  & 1.4 & 3.9 & 8.8 & 1.1 & 4.6 & 9.2 & 1.0 & 5.2 & 10.0 \\
 &  & Linear  & 0.8 & 5.0 & 11.0 & 1.3 & 5.3 & 10.0 & 1.6 & 5.1 & 9.5 \\
 &  & Cubic  & 2.9 & 10.9 & 20.3 & 4.9 & 20.7 & 33.2 & 16.4 & 44.6 & 61.4 \\
 &  & Linear threshold  & 15.4 & 33.9 & 46.4 & 33.6 & 55.3 & 66.2 & 67.4 & 85.4 & 90.8 \\
 &  & Quadratic  & 15.9 & 35.0 & 47.4 & 39.5 & 63.1 & 75.7 & 76.8 & 89.8 & 93.6 \\
 &  & Quadratic exponential  & 44.8 & 67.9 & 78.3 & 85.7 & 96.2 & 98.3 & 99.5 & 99.8 & 99.9 \\
 &  & Quadratic bump  & 99.4 & 99.9 & 99.9 & 100.0 & 100.0 & 100.0 & 100.0 & 100.0 & 100.0 \\
\addlinespace[0.20em]
 & \multirow{7}{*}{5} & Constant  & 0.5 & 4.5 & 9.7 & 1.0 & 4.9 & 8.9 & 0.7 & 4.5 & 11.1 \\
 &  & Linear  & 1.2 & 4.9 & 9.7 & 1.9 & 5.9 & 11.6 & 0.6 & 4.9 & 10.5 \\
 &  & Cubic  & 2.2 & 11.3 & 20.0 & 4.8 & 18.8 & 31.5 & 17.2 & 42.4 & 56.2 \\
 &  & Linear threshold  & 16.6 & 34.6 & 44.9 & 35.0 & 60.8 & 71.2 & 66.5 & 84.9 & 90.7 \\
 &  & Quadratic  & 16.2 & 31.5 & 42.7 & 37.3 & 58.2 & 68.0 & 76.9 & 90.8 & 94.7 \\
 &  & Quadratic exponential  & 46.5 & 69.1 & 78.1 & 84.8 & 94.7 & 97.2 & 99.5 & 100.0 & 100.0 \\
 &  & Quadratic bump  & 99.5 & 99.5 & 99.8 & 100.0 & 100.0 & 100.0 & 100.0 & 100.0 & 100.0 \\
\addlinespace[0.20em]
 & \multirow{7}{*}{10} & Constant  & 1.3 & 5.9 & 9.8 & 0.9 & 6.2 & 11.6 & 1.3 & 5.6 & 10.3 \\
 &  & Linear  & 1.2 & 5.9 & 11.2 & 1.3 & 5.4 & 10.6 & 1.6 & 5.7 & 10.6 \\
 &  & Cubic  & 2.7 & 12.3 & 20.5 & 6.6 & 20.4 & 31.2 & 16.4 & 42.3 & 57.0 \\
 &  & Linear threshold  & 15.1 & 33.2 & 43.5 & 30.9 & 53.7 & 64.6 & 66.1 & 83.9 & 90.4 \\
 &  & Quadratic  & 11.8 & 27.5 & 38.1 & 34.1 & 56.3 & 66.2 & 71.2 & 87.2 & 93.0 \\
 &  & Quadratic exponential  & 49.1 & 70.4 & 79.6 & 85.4 & 94.3 & 97.4 & 99.5 & 99.9 & 100.0 \\
 &  & Quadratic bump  & 98.8 & 99.7 & 99.8 & 100.0 & 100.0 & 100.0 & 100.0 & 100.0 & 100.0 \\
\bottomrule
\end{tabular}
\begin{tablenotes}
\footnotesize
\item Notes: Rejection frequencies are reported in percent. The nominal levels are 1\%, 5\%, and 10\%. The CATE forms are plotted in Figure \ref{fig:cate-forms-parametric}.
\end{tablenotes}
\end{threeparttable}
\end{table}

\begin{table}[!htbp]
\centering
\scriptsize
\setlength{\tabcolsep}{5.0pt}
\renewcommand{\arraystretch}{0.96}
\begin{threeparttable}
\caption{Rejection frequencies for the KS linear-form test. Constant and Linear correspond to size, and the remaining CATE forms to power.}
\label{tab:linear-form-ks}
\begin{tabular}{ll l *{9}{c}}
\toprule
& & & \multicolumn{3}{c}{$n=500$} & \multicolumn{3}{c}{$n=1000$} & \multicolumn{3}{c}{$n=2000$} \\
\cmidrule(lr){4-6}\cmidrule(lr){7-9}\cmidrule(lr){10-12}
Family & $d$ & CATE form & 1\% & 5\% & 10\% & 1\% & 5\% & 10\% & 1\% & 5\% & 10\% \\
\midrule
\multirow{21}{*}{Smooth logit} & \multirow{7}{*}{3} & Constant  & 0.6 & 5.3 & 9.5 & 1.0 & 6.2 & 11.8 & 1.1 & 6.3 & 10.4 \\
 &  & Linear  & 1.4 & 6.0 & 12.4 & 0.7 & 4.5 & 8.5 & 1.1 & 5.7 & 11.1 \\
 &  & Cubic  & 2.9 & 11.4 & 19.5 & 2.3 & 12.5 & 24.7 & 10.7 & 29.5 & 43.0 \\
 &  & Linear threshold  & 6.4 & 15.0 & 23.0 & 11.5 & 26.1 & 38.7 & 24.6 & 46.7 & 60.8 \\
 &  & Quadratic  & 9.5 & 24.4 & 36.9 & 23.5 & 48.0 & 59.2 & 54.8 & 77.7 & 86.8 \\
 &  & Quadratic exponential  & 32.1 & 58.7 & 67.6 & 70.2 & 87.5 & 91.3 & 98.2 & 99.9 & 100.0 \\
 &  & Quadratic bump  & 63.1 & 82.3 & 89.6 & 95.2 & 99.0 & 99.6 & 100.0 & 100.0 & 100.0 \\
\addlinespace[0.20em]
 & \multirow{7}{*}{5} & Constant  & 1.4 & 5.9 & 9.9 & 0.9 & 4.9 & 9.5 & 1.5 & 6.2 & 11.0 \\
 &  & Linear  & 1.1 & 5.4 & 10.6 & 0.8 & 5.8 & 10.9 & 0.7 & 4.6 & 9.3 \\
 &  & Cubic  & 2.2 & 10.1 & 18.0 & 4.1 & 16.5 & 29.0 & 8.6 & 27.1 & 40.5 \\
 &  & Linear threshold  & 5.5 & 16.0 & 24.6 & 9.1 & 25.5 & 37.5 & 27.5 & 50.5 & 63.4 \\
 &  & Quadratic  & 9.8 & 25.5 & 37.7 & 21.7 & 45.9 & 60.6 & 55.4 & 77.6 & 86.6 \\
 &  & Quadratic exponential  & 30.8 & 54.4 & 65.8 & 69.2 & 88.3 & 93.3 & 96.6 & 99.8 & 100.0 \\
 &  & Quadratic bump  & 63.6 & 85.5 & 92.0 & 94.8 & 99.4 & 99.8 & 100.0 & 100.0 & 100.0 \\
\addlinespace[0.20em]
 & \multirow{7}{*}{10} & Constant  & 0.9 & 5.5 & 11.1 & 1.1 & 6.4 & 12.1 & 0.9 & 5.8 & 11.0 \\
 &  & Linear  & 0.6 & 5.3 & 11.2 & 1.1 & 4.7 & 9.0 & 1.7 & 6.3 & 11.2 \\
 &  & Cubic  & 2.7 & 10.2 & 17.0 & 4.5 & 17.8 & 27.5 & 9.8 & 28.8 & 42.9 \\
 &  & Linear threshold  & 3.1 & 11.6 & 22.3 & 11.8 & 30.1 & 41.1 & 25.7 & 49.3 & 62.6 \\
 &  & Quadratic  & 9.9 & 25.3 & 37.2 & 20.7 & 46.8 & 60.0 & 55.5 & 76.5 & 85.2 \\
 &  & Quadratic exponential  & 30.0 & 55.7 & 67.5 & 72.0 & 88.0 & 92.8 & 98.0 & 99.9 & 100.0 \\
 &  & Quadratic bump  & 61.4 & 82.3 & 90.2 & 95.6 & 98.9 & 99.4 & 100.0 & 100.0 & 100.0 \\
\midrule
\multirow{21}{*}{Threshold} & \multirow{7}{*}{3} & Constant  & 1.1 & 4.8 & 8.9 & 1.5 & 7.3 & 13.6 & 1.5 & 6.0 & 10.2 \\
 &  & Linear  & 1.2 & 5.8 & 11.4 & 0.9 & 5.4 & 9.8 & 1.7 & 5.4 & 10.5 \\
 &  & Cubic  & 1.9 & 9.6 & 17.0 & 3.5 & 14.3 & 23.8 & 6.8 & 25.2 & 39.3 \\
 &  & Linear threshold  & 6.1 & 16.9 & 24.6 & 9.0 & 24.9 & 38.4 & 23.2 & 46.3 & 58.8 \\
 &  & Quadratic  & 9.2 & 22.0 & 32.8 & 20.9 & 40.2 & 53.2 & 44.0 & 67.6 & 79.3 \\
 &  & Quadratic exponential  & 25.8 & 51.1 & 62.2 & 57.7 & 81.6 & 89.4 & 93.6 & 98.9 & 99.5 \\
 &  & Quadratic bump  & 54.3 & 76.0 & 85.5 & 89.4 & 98.5 & 99.5 & 99.9 & 99.9 & 100.0 \\
\addlinespace[0.20em]
 & \multirow{7}{*}{5} & Constant  & 1.5 & 6.2 & 12.1 & 1.2 & 6.4 & 13.0 & 1.5 & 5.7 & 9.8 \\
 &  & Linear  & 0.7 & 5.4 & 10.5 & 1.2 & 5.8 & 10.4 & 1.1 & 5.6 & 9.9 \\
 &  & Cubic  & 1.6 & 8.0 & 14.8 & 4.0 & 13.2 & 23.4 & 7.8 & 23.4 & 34.8 \\
 &  & Linear threshold  & 2.7 & 10.7 & 18.7 & 8.0 & 22.1 & 33.9 & 19.3 & 43.2 & 54.3 \\
 &  & Quadratic  & 6.2 & 21.0 & 32.7 & 17.8 & 39.3 & 51.3 & 43.9 & 68.4 & 79.0 \\
 &  & Quadratic exponential  & 17.8 & 37.8 & 49.9 & 56.6 & 80.6 & 87.9 & 93.3 & 98.6 & 99.4 \\
 &  & Quadratic bump  & 48.8 & 73.0 & 81.5 & 88.9 & 97.3 & 98.6 & 99.9 & 100.0 & 100.0 \\
\addlinespace[0.20em]
 & \multirow{7}{*}{10} & Constant  & 0.9 & 3.9 & 8.7 & 1.3 & 6.4 & 11.4 & 1.3 & 6.0 & 9.8 \\
 &  & Linear  & 0.9 & 4.7 & 10.5 & 0.8 & 4.2 & 8.6 & 1.4 & 6.1 & 12.5 \\
 &  & Cubic  & 2.1 & 8.3 & 16.1 & 2.4 & 12.0 & 21.1 & 8.0 & 22.4 & 37.3 \\
 &  & Linear threshold  & 4.1 & 12.8 & 20.0 & 8.0 & 22.7 & 32.8 & 19.8 & 41.7 & 53.6 \\
 &  & Quadratic  & 7.8 & 17.0 & 28.4 & 20.2 & 41.3 & 53.5 & 46.3 & 70.2 & 78.9 \\
 &  & Quadratic exponential  & 19.8 & 42.7 & 55.5 & 57.5 & 79.0 & 86.9 & 93.9 & 98.6 & 99.2 \\
 &  & Quadratic bump  & 45.0 & 68.8 & 79.4 & 89.3 & 98.0 & 98.9 & 99.6 & 100.0 & 100.0 \\
\midrule
\multirow{21}{*}{Fractional} & \multirow{7}{*}{3} & Constant  & 0.7 & 3.9 & 10.0 & 1.1 & 4.5 & 10.1 & 1.1 & 4.8 & 10.1 \\
 &  & Linear  & 0.9 & 4.9 & 11.6 & 0.9 & 4.6 & 9.9 & 0.9 & 6.0 & 10.7 \\
 &  & Cubic  & 3.8 & 11.8 & 20.1 & 5.5 & 19.7 & 31.3 & 14.6 & 38.0 & 52.2 \\
 &  & Linear threshold  & 11.0 & 25.8 & 38.2 & 22.1 & 45.5 & 58.2 & 51.2 & 75.5 & 85.2 \\
 &  & Quadratic  & 9.1 & 26.8 & 39.7 & 27.5 & 52.1 & 64.1 & 60.1 & 81.2 & 88.2 \\
 &  & Quadratic exponential  & 31.3 & 57.6 & 69.3 & 72.8 & 89.3 & 94.5 & 97.7 & 99.6 & 99.9 \\
 &  & Quadratic bump  & 98.5 & 99.6 & 99.8 & 100.0 & 100.0 & 100.0 & 100.0 & 100.0 & 100.0 \\
\addlinespace[0.20em]
 & \multirow{7}{*}{5} & Constant  & 0.9 & 5.7 & 9.8 & 0.8 & 5.5 & 10.0 & 0.9 & 4.3 & 9.8 \\
 &  & Linear  & 0.7 & 4.6 & 8.7 & 1.1 & 6.2 & 11.8 & 0.8 & 5.3 & 10.9 \\
 &  & Cubic  & 2.5 & 10.4 & 20.1 & 5.6 & 17.9 & 27.2 & 16.7 & 38.9 & 53.3 \\
 &  & Linear threshold  & 11.3 & 26.4 & 38.2 & 22.0 & 45.4 & 60.8 & 48.8 & 72.5 & 82.4 \\
 &  & Quadratic  & 9.6 & 25.5 & 36.1 & 24.3 & 46.9 & 58.7 & 59.6 & 83.0 & 90.2 \\
 &  & Quadratic exponential  & 34.6 & 59.3 & 71.3 & 73.1 & 90.1 & 94.5 & 98.2 & 99.5 & 99.9 \\
 &  & Quadratic bump  & 98.4 & 99.5 & 99.6 & 100.0 & 100.0 & 100.0 & 100.0 & 100.0 & 100.0 \\
\addlinespace[0.20em]
 & \multirow{7}{*}{10} & Constant  & 1.8 & 5.9 & 11.2 & 1.0 & 6.1 & 11.3 & 1.9 & 6.5 & 11.1 \\
 &  & Linear  & 1.4 & 5.5 & 10.4 & 1.8 & 6.4 & 11.7 & 1.8 & 5.7 & 10.6 \\
 &  & Cubic  & 3.1 & 11.2 & 19.3 & 5.8 & 17.9 & 28.0 & 13.6 & 33.2 & 47.8 \\
 &  & Linear threshold  & 10.6 & 24.5 & 36.5 & 21.6 & 43.6 & 56.5 & 46.8 & 73.6 & 83.5 \\
 &  & Quadratic  & 7.8 & 21.5 & 32.1 & 21.8 & 45.3 & 58.6 & 53.8 & 77.0 & 86.6 \\
 &  & Quadratic exponential  & 34.1 & 58.4 & 71.4 & 71.0 & 89.6 & 95.1 & 97.1 & 99.4 & 99.8 \\
 &  & Quadratic bump  & 96.9 & 99.0 & 99.5 & 100.0 & 100.0 & 100.0 & 100.0 & 100.0 & 100.0 \\
\bottomrule
\end{tabular}
\begin{tablenotes}
\footnotesize
\item Notes: Rejection frequencies are reported in percent. The nominal levels are 1\%, 5\%, and 10\%. The CATE forms are plotted in Figure \ref{fig:cate-forms-parametric}.
\end{tablenotes}
\end{threeparttable}
\end{table}

\begin{table}[!htbp]
\centering
\scriptsize
\setlength{\tabcolsep}{5.0pt}
\renewcommand{\arraystretch}{0.96}
\begin{threeparttable}
\caption{Rejection frequencies for the CvM quadratic-form test. Constant, Linear and Quadratic correspond to size, and the remaining CATE forms to power.}
\label{tab:quad-form-cvm}
\begin{tabular}{ll l *{9}{c}}
\toprule
& & & \multicolumn{3}{c}{$n=500$} & \multicolumn{3}{c}{$n=1000$} & \multicolumn{3}{c}{$n=2000$} \\
\cmidrule(lr){4-6}\cmidrule(lr){7-9}\cmidrule(lr){10-12}
Family & $d$ & CATE form & 1\% & 5\% & 10\% & 1\% & 5\% & 10\% & 1\% & 5\% & 10\% \\
\midrule
\multirow{21}{*}{Smooth logit} & \multirow{7}{*}{3} & Constant  & 0.6 & 4.2 & 9.1 & 1.2 & 5.4 & 10.1 & 1.0 & 5.9 & 10.9 \\
 &  & Linear  & 1.0 & 4.8 & 10.3 & 0.9 & 4.7 & 9.9 & 0.6 & 5.2 & 9.8 \\
 &  & Cubic  & 7.3 & 18.2 & 29.5 & 12.6 & 30.1 & 41.5 & 33.7 & 56.1 & 67.0 \\
 &  & Linear threshold  & 0.7 & 4.8 & 10.2 & 0.9 & 5.4 & 11.5 & 1.2 & 5.9 & 11.2 \\
 &  & Quadratic  & 0.5 & 4.5 & 8.6 & 0.9 & 4.9 & 10.0 & 0.7 & 4.3 & 10.2 \\
 &  & Quadratic exponential  & 3.5 & 12.7 & 19.4 & 5.6 & 16.6 & 26.8 & 14.4 & 30.6 & 41.5 \\
 &  & Quadratic bump  & 8.4 & 23.5 & 35.8 & 23.0 & 49.0 & 62.8 & 63.0 & 86.3 & 92.9 \\
\addlinespace[0.20em]
 & \multirow{7}{*}{5} & Constant  & 0.8 & 4.7 & 10.2 & 0.7 & 4.6 & 9.3 & 1.5 & 5.9 & 10.3 \\
 &  & Linear  & 0.6 & 5.4 & 9.6 & 1.0 & 5.1 & 11.5 & 0.9 & 4.7 & 8.2 \\
 &  & Cubic  & 6.1 & 15.6 & 24.0 & 14.3 & 33.2 & 44.5 & 33.7 & 57.0 & 69.0 \\
 &  & Linear threshold  & 0.7 & 5.1 & 9.4 & 1.5 & 6.1 & 11.7 & 1.9 & 8.8 & 15.5 \\
 &  & Quadratic  & 1.3 & 4.8 & 10.1 & 1.1 & 6.1 & 12.8 & 0.7 & 4.4 & 10.1 \\
 &  & Quadratic exponential  & 4.0 & 11.5 & 19.5 & 6.4 & 18.4 & 27.8 & 13.1 & 31.3 & 43.5 \\
 &  & Quadratic bump  & 8.9 & 27.3 & 39.3 & 22.7 & 49.6 & 64.0 & 62.9 & 88.1 & 93.1 \\
\addlinespace[0.20em]
 & \multirow{7}{*}{10} & Constant  & 0.4 & 4.9 & 10.0 & 0.9 & 4.1 & 9.4 & 1.4 & 5.1 & 11.6 \\
 &  & Linear  & 0.7 & 5.4 & 10.7 & 0.7 & 4.3 & 8.8 & 1.8 & 6.0 & 11.1 \\
 &  & Cubic  & 6.2 & 16.8 & 26.0 & 15.3 & 33.3 & 42.1 & 33.5 & 55.8 & 67.8 \\
 &  & Linear threshold  & 0.3 & 3.7 & 8.8 & 1.3 & 6.3 & 12.4 & 1.7 & 5.4 & 12.8 \\
 &  & Quadratic  & 1.1 & 4.4 & 10.1 & 0.5 & 4.3 & 9.6 & 0.8 & 4.8 & 9.4 \\
 &  & Quadratic exponential  & 2.6 & 10.8 & 16.9 & 7.1 & 19.0 & 28.1 & 14.8 & 32.6 & 42.7 \\
 &  & Quadratic bump  & 6.9 & 24.8 & 37.9 & 24.4 & 50.4 & 65.0 & 62.6 & 86.5 & 93.2 \\
\midrule
\multirow{21}{*}{Threshold} & \multirow{7}{*}{3} & Constant  & 0.9 & 5.1 & 9.4 & 1.4 & 6.1 & 10.9 & 1.2 & 4.5 & 9.5 \\
 &  & Linear  & 0.3 & 5.2 & 10.4 & 1.1 & 5.7 & 10.6 & 1.3 & 5.1 & 9.9 \\
 &  & Cubic  & 4.7 & 15.2 & 22.6 & 10.9 & 27.1 & 38.3 & 28.4 & 52.6 & 64.2 \\
 &  & Linear threshold  & 1.6 & 7.3 & 12.1 & 1.1 & 6.2 & 11.3 & 1.4 & 5.6 & 11.4 \\
 &  & Quadratic  & 0.9 & 5.6 & 10.7 & 1.3 & 4.8 & 9.5 & 0.6 & 4.5 & 9.7 \\
 &  & Quadratic exponential  & 2.0 & 10.2 & 18.7 & 5.5 & 14.5 & 25.0 & 8.3 & 24.1 & 34.8 \\
 &  & Quadratic bump  & 6.2 & 23.2 & 36.1 & 19.6 & 43.1 & 58.5 & 50.7 & 77.2 & 86.9 \\
\addlinespace[0.20em]
 & \multirow{7}{*}{5} & Constant  & 1.0 & 4.7 & 10.5 & 1.6 & 6.6 & 12.1 & 0.9 & 5.2 & 9.8 \\
 &  & Linear  & 0.4 & 3.8 & 8.6 & 0.9 & 4.0 & 9.0 & 0.7 & 4.2 & 9.7 \\
 &  & Cubic  & 3.9 & 11.4 & 17.9 & 11.3 & 25.7 & 36.1 & 26.1 & 48.6 & 60.6 \\
 &  & Linear threshold  & 0.7 & 4.8 & 9.8 & 1.2 & 6.0 & 11.4 & 1.3 & 6.0 & 11.5 \\
 &  & Quadratic  & 1.1 & 3.8 & 9.1 & 0.9 & 5.2 & 11.3 & 0.5 & 4.4 & 9.3 \\
 &  & Quadratic exponential  & 2.2 & 9.1 & 15.1 & 5.4 & 15.4 & 22.9 & 9.5 & 24.2 & 34.2 \\
 &  & Quadratic bump  & 5.0 & 18.2 & 30.2 & 16.9 & 40.3 & 55.0 & 52.8 & 79.7 & 89.6 \\
\addlinespace[0.20em]
 & \multirow{7}{*}{10} & Constant  & 0.4 & 3.9 & 8.8 & 0.7 & 5.0 & 10.7 & 0.7 & 5.1 & 10.0 \\
 &  & Linear  & 0.6 & 3.8 & 8.5 & 1.3 & 5.2 & 8.8 & 1.2 & 6.1 & 12.2 \\
 &  & Cubic  & 4.0 & 12.2 & 19.8 & 9.6 & 25.6 & 36.1 & 27.7 & 49.6 & 61.2 \\
 &  & Linear threshold  & 0.6 & 4.3 & 10.3 & 1.4 & 5.8 & 10.9 & 1.0 & 6.0 & 12.6 \\
 &  & Quadratic  & 0.9 & 4.4 & 9.9 & 0.8 & 3.3 & 8.4 & 0.8 & 5.2 & 10.2 \\
 &  & Quadratic exponential  & 2.7 & 8.6 & 15.7 & 5.3 & 15.1 & 21.7 & 11.6 & 24.1 & 34.3 \\
 &  & Quadratic bump  & 4.4 & 17.2 & 28.5 & 15.7 & 41.4 & 56.1 & 50.3 & 76.2 & 86.2 \\
\midrule
\multirow{21}{*}{Fractional} & \multirow{7}{*}{3} & Constant  & 0.7 & 3.3 & 8.3 & 0.8 & 4.3 & 9.8 & 1.0 & 5.6 & 10.8 \\
 &  & Linear  & 0.9 & 4.4 & 9.2 & 0.9 & 5.2 & 9.4 & 0.6 & 4.7 & 10.2 \\
 &  & Cubic  & 6.3 & 18.1 & 28.8 & 13.7 & 33.3 & 43.6 & 37.9 & 61.4 & 73.6 \\
 &  & Linear threshold  & 1.4 & 6.4 & 11.7 & 1.6 & 8.2 & 13.7 & 2.2 & 8.9 & 14.5 \\
 &  & Quadratic  & 0.5 & 6.1 & 10.7 & 1.4 & 4.8 & 9.9 & 1.3 & 5.5 & 10.1 \\
 &  & Quadratic exponential  & 3.3 & 10.1 & 17.0 & 5.5 & 17.6 & 27.0 & 13.6 & 30.5 & 40.9 \\
 &  & Quadratic bump  & 39.3 & 64.2 & 75.7 & 83.4 & 94.2 & 96.4 & 99.3 & 99.9 & 100.0 \\
\addlinespace[0.20em]
 & \multirow{7}{*}{5} & Constant  & 1.1 & 5.1 & 9.7 & 1.2 & 5.2 & 10.6 & 0.7 & 4.6 & 9.1 \\
 &  & Linear  & 0.5 & 5.2 & 10.5 & 1.4 & 4.5 & 9.3 & 0.3 & 4.5 & 9.6 \\
 &  & Cubic  & 6.6 & 17.5 & 26.8 & 14.2 & 32.8 & 43.8 & 37.2 & 60.1 & 72.0 \\
 &  & Linear threshold  & 1.4 & 6.5 & 12.6 & 0.9 & 5.9 & 12.7 & 0.9 & 8.3 & 14.9 \\
 &  & Quadratic  & 1.0 & 5.4 & 10.6 & 1.3 & 4.3 & 9.3 & 1.0 & 5.1 & 9.1 \\
 &  & Quadratic exponential  & 3.2 & 13.2 & 21.8 & 5.9 & 16.8 & 26.5 & 13.5 & 31.0 & 43.1 \\
 &  & Quadratic bump  & 35.9 & 64.7 & 75.2 & 80.1 & 93.4 & 96.6 & 99.8 & 100.0 & 100.0 \\
\addlinespace[0.20em]
 & \multirow{7}{*}{10} & Constant  & 1.2 & 4.5 & 9.4 & 0.8 & 5.4 & 10.1 & 1.4 & 5.1 & 11.1 \\
 &  & Linear  & 1.6 & 4.8 & 10.1 & 1.1 & 5.1 & 10.7 & 1.4 & 4.4 & 9.8 \\
 &  & Cubic  & 6.7 & 17.6 & 28.5 & 15.7 & 31.6 & 42.9 & 37.8 & 59.5 & 70.9 \\
 &  & Linear threshold  & 0.8 & 6.5 & 11.6 & 1.4 & 6.9 & 12.0 & 1.7 & 8.2 & 16.0 \\
 &  & Quadratic  & 0.7 & 4.6 & 10.6 & 0.5 & 5.1 & 8.9 & 1.3 & 5.4 & 10.0 \\
 &  & Quadratic exponential  & 3.3 & 10.4 & 17.5 & 6.2 & 16.2 & 27.4 & 11.7 & 29.0 & 40.4 \\
 &  & Quadratic bump  & 32.4 & 58.7 & 72.0 & 76.8 & 92.6 & 96.9 & 99.4 & 100.0 & 100.0 \\
\bottomrule
\end{tabular}
\begin{tablenotes}
\footnotesize
\item Notes: Rejection frequencies are reported in percent. The nominal levels are 1\%, 5\%, and 10\%. The CATE forms are plotted in Figure \ref{fig:cate-forms-parametric}.
\end{tablenotes}
\end{threeparttable}
\end{table}

\begin{table}[!htbp]
\centering
\scriptsize
\setlength{\tabcolsep}{5.0pt}
\renewcommand{\arraystretch}{0.96}
\begin{threeparttable}
\caption{Rejection frequencies for the KS quadratic-form test. Constant, Linear and Quadratic correspond to size, and the remaining CATE forms to power.}
\label{tab:linear-form-ks}
\begin{tabular}{ll l *{9}{c}}
\toprule
& & & \multicolumn{3}{c}{$n=500$} & \multicolumn{3}{c}{$n=1000$} & \multicolumn{3}{c}{$n=2000$} \\
\cmidrule(lr){4-6}\cmidrule(lr){7-9}\cmidrule(lr){10-12}
Family & $d$ & CATE form & 1\% & 5\% & 10\% & 1\% & 5\% & 10\% & 1\% & 5\% & 10\% \\
\midrule
\multirow{21}{*}{Smooth logit} & \multirow{7}{*}{3} & Constant  & 0.5 & 3.7 & 8.2 & 1.7 & 5.4 & 10.7 & 1.3 & 5.3 & 10.8 \\
 &  & Linear  & 1.2 & 5.6 & 10.7 & 0.8 & 4.6 & 9.7 & 1.0 & 4.7 & 9.2 \\
 &  & Cubic  & 4.6 & 14.8 & 23.7 & 8.0 & 21.9 & 34.5 & 19.7 & 43.4 & 55.1 \\
 &  & Linear threshold  & 1.4 & 5.4 & 10.5 & 0.9 & 4.2 & 9.5 & 1.3 & 4.9 & 11.2 \\
 &  & Quadratic  & 0.7 & 4.3 & 8.5 & 0.6 & 4.0 & 9.5 & 1.0 & 6.0 & 11.5 \\
 &  & Quadratic exponential  & 3.0 & 10.4 & 16.7 & 3.5 & 13.0 & 21.1 & 9.5 & 22.9 & 33.2 \\
 &  & Quadratic bump  & 8.2 & 20.7 & 33.1 & 17.9 & 40.4 & 53.2 & 51.1 & 78.4 & 88.0 \\
\addlinespace[0.20em]
 & \multirow{7}{*}{5} & Constant  & 1.2 & 4.2 & 9.8 & 0.9 & 4.7 & 9.7 & 1.0 & 5.5 & 11.2 \\
 &  & Linear  & 1.0 & 5.3 & 9.5 & 1.4 & 6.4 & 11.0 & 0.8 & 3.9 & 9.2 \\
 &  & Cubic  & 3.2 & 12.6 & 19.4 & 8.4 & 25.3 & 36.4 & 20.9 & 42.6 & 56.5 \\
 &  & Linear threshold  & 0.8 & 4.2 & 10.1 & 1.8 & 7.1 & 12.7 & 2.3 & 8.7 & 13.8 \\
 &  & Quadratic  & 0.9 & 4.8 & 10.3 & 1.1 & 4.3 & 10.6 & 1.3 & 6.6 & 11.5 \\
 &  & Quadratic exponential  & 2.2 & 8.9 & 16.6 & 4.0 & 12.2 & 21.2 & 8.1 & 21.1 & 32.5 \\
 &  & Quadratic bump  & 7.9 & 26.4 & 36.9 & 19.0 & 43.6 & 57.3 & 50.0 & 79.2 & 89.2 \\
\addlinespace[0.20em]
 & \multirow{7}{*}{10} & Constant  & 0.5 & 5.3 & 10.6 & 1.1 & 3.8 & 8.8 & 0.9 & 5.6 & 11.4 \\
 &  & Linear  & 1.3 & 5.5 & 10.5 & 1.2 & 5.0 & 9.2 & 1.0 & 5.7 & 12.1 \\
 &  & Cubic  & 3.9 & 13.4 & 22.3 & 7.9 & 26.2 & 36.6 & 19.3 & 42.8 & 55.7 \\
 &  & Linear threshold  & 0.7 & 4.7 & 9.0 & 1.3 & 6.4 & 11.5 & 1.4 & 7.8 & 13.9 \\
 &  & Quadratic  & 1.2 & 4.2 & 10.3 & 0.5 & 4.1 & 8.6 & 0.8 & 4.5 & 8.8 \\
 &  & Quadratic exponential  & 1.6 & 8.2 & 17.2 & 4.0 & 15.1 & 21.7 & 8.5 & 24.1 & 36.4 \\
 &  & Quadratic bump  & 6.8 & 20.0 & 31.8 & 20.2 & 42.1 & 55.9 & 53.9 & 78.2 & 88.7 \\
\midrule
\multirow{21}{*}{Threshold} & \multirow{7}{*}{3} & Constant  & 1.3 & 5.0 & 9.5 & 1.7 & 6.9 & 11.5 & 1.3 & 4.6 & 9.6 \\
 &  & Linear  & 0.7 & 4.9 & 10.4 & 1.1 & 5.9 & 11.4 & 1.4 & 5.1 & 10.5 \\
 &  & Cubic  & 2.7 & 12.1 & 21.5 & 7.1 & 20.2 & 30.6 & 17.1 & 38.7 & 52.8 \\
 &  & Linear threshold  & 2.1 & 6.5 & 10.9 & 1.3 & 5.8 & 13.5 & 1.3 & 5.9 & 11.4 \\
 &  & Quadratic  & 0.7 & 6.3 & 11.6 & 0.7 & 4.6 & 8.7 & 1.3 & 4.8 & 10.2 \\
 &  & Quadratic exponential  & 1.3 & 8.5 & 15.4 & 4.1 & 12.8 & 20.1 & 7.3 & 18.8 & 28.2 \\
 &  & Quadratic bump  & 6.4 & 18.7 & 31.2 & 14.4 & 36.7 & 52.0 & 40.4 & 67.9 & 80.9 \\
\addlinespace[0.20em]
 & \multirow{7}{*}{5} & Constant  & 1.2 & 4.5 & 10.1 & 1.7 & 5.2 & 11.4 & 0.4 & 5.4 & 9.9 \\
 &  & Linear  & 0.5 & 3.7 & 8.0 & 0.8 & 4.8 & 9.9 & 0.8 & 4.2 & 10.2 \\
 &  & Cubic  & 2.7 & 10.5 & 16.6 & 8.0 & 21.0 & 31.4 & 15.0 & 33.2 & 48.5 \\
 &  & Linear threshold  & 0.6 & 4.2 & 9.2 & 1.6 & 6.8 & 11.5 & 1.1 & 6.0 & 11.4 \\
 &  & Quadratic  & 0.9 & 4.4 & 9.5 & 0.7 & 5.9 & 11.4 & 0.7 & 4.8 & 9.2 \\
 &  & Quadratic exponential  & 1.5 & 7.4 & 13.6 & 3.9 & 13.3 & 21.4 & 5.9 & 18.6 & 29.8 \\
 &  & Quadratic bump  & 4.8 & 16.8 & 26.9 & 13.5 & 34.6 & 49.5 & 40.2 & 69.9 & 81.6 \\
\addlinespace[0.20em]
 & \multirow{7}{*}{10} & Constant  & 0.7 & 4.2 & 10.3 & 1.2 & 5.0 & 9.8 & 0.5 & 5.6 & 10.2 \\
 &  & Linear  & 0.7 & 5.4 & 10.4 & 1.1 & 5.8 & 10.5 & 1.0 & 6.3 & 11.8 \\
 &  & Cubic  & 3.2 & 10.9 & 17.2 & 6.0 & 19.8 & 29.3 & 14.2 & 35.7 & 50.3 \\
 &  & Linear threshold  & 0.7 & 4.3 & 9.8 & 1.2 & 5.6 & 11.7 & 0.7 & 5.8 & 11.4 \\
 &  & Quadratic  & 1.0 & 4.3 & 10.1 & 1.0 & 3.7 & 8.3 & 1.0 & 4.6 & 8.9 \\
 &  & Quadratic exponential  & 1.7 & 8.6 & 15.6 & 3.5 & 12.0 & 18.2 & 7.7 & 18.7 & 29.3 \\
 &  & Quadratic bump  & 4.1 & 16.0 & 25.8 & 12.0 & 32.5 & 48.0 & 39.0 & 68.9 & 79.9 \\
\midrule
\multirow{21}{*}{Fractional} & \multirow{7}{*}{3} & Constant  & 0.9 & 4.9 & 10.4 & 1.2 & 4.6 & 9.3 & 0.9 & 4.6 & 11.0 \\
 &  & Linear  & 1.3 & 5.8 & 11.5 & 0.6 & 4.2 & 9.8 & 1.1 & 4.7 & 9.4 \\
 &  & Cubic  & 5.1 & 15.4 & 26.0 & 10.5 & 25.9 & 36.7 & 25.9 & 48.5 & 61.6 \\
 &  & Linear threshold  & 1.7 & 6.4 & 11.6 & 1.6 & 7.2 & 14.1 & 2.1 & 7.8 & 15.6 \\
 &  & Quadratic  & 0.9 & 5.2 & 11.3 & 1.7 & 5.3 & 11.0 & 1.0 & 5.2 & 10.8 \\
 &  & Quadratic exponential  & 2.2 & 7.8 & 15.0 & 4.4 & 13.9 & 22.8 & 8.8 & 24.1 & 33.9 \\
 &  & Quadratic bump  & 24.3 & 52.0 & 66.8 & 62.7 & 86.6 & 92.6 & 95.5 & 99.4 & 100.0 \\
\addlinespace[0.20em]
 & \multirow{7}{*}{5} & Constant  & 1.2 & 5.2 & 11.0 & 0.6 & 4.9 & 8.9 & 0.5 & 4.8 & 9.2 \\
 &  & Linear  & 0.8 & 3.8 & 10.7 & 0.9 & 4.2 & 10.1 & 0.9 & 5.4 & 10.3 \\
 &  & Cubic  & 4.2 & 14.2 & 23.0 & 8.1 & 24.0 & 36.7 & 24.8 & 47.7 & 60.9 \\
 &  & Linear threshold  & 1.4 & 6.5 & 12.6 & 1.1 & 7.3 & 13.8 & 1.2 & 7.7 & 15.0 \\
 &  & Quadratic  & 1.0 & 5.8 & 11.1 & 0.9 & 5.2 & 11.2 & 0.9 & 4.0 & 9.0 \\
 &  & Quadratic exponential  & 3.3 & 11.0 & 21.5 & 3.8 & 14.4 & 24.9 & 9.1 & 24.5 & 35.6 \\
 &  & Quadratic bump  & 23.2 & 52.6 & 65.8 & 62.8 & 86.1 & 92.0 & 95.2 & 99.9 & 100.0 \\
\addlinespace[0.20em]
 & \multirow{7}{*}{10} & Constant  & 1.3 & 5.1 & 9.7 & 1.2 & 5.4 & 9.8 & 1.7 & 5.2 & 10.6 \\
 &  & Linear  & 0.7 & 4.9 & 10.3 & 1.5 & 6.0 & 10.7 & 1.3 & 5.5 & 10.9 \\
 &  & Cubic  & 4.6 & 14.8 & 22.8 & 9.2 & 25.0 & 36.2 & 20.5 & 45.6 & 58.0 \\
 &  & Linear threshold  & 0.8 & 6.5 & 11.5 & 1.7 & 7.4 & 12.6 & 1.2 & 7.8 & 15.6 \\
 &  & Quadratic  & 1.1 & 4.6 & 10.1 & 1.0 & 4.5 & 9.2 & 1.0 & 4.8 & 10.3 \\
 &  & Quadratic exponential  & 2.7 & 10.3 & 18.1 & 4.5 & 13.8 & 23.4 & 7.2 & 20.5 & 33.7 \\
 &  & Quadratic bump  & 21.4 & 46.6 & 61.9 & 58.2 & 84.3 & 91.7 & 94.4 & 99.5 & 99.9 \\
\bottomrule
\end{tabular}
\begin{tablenotes}
\footnotesize
\item Notes: Rejection frequencies are reported in percent. The nominal levels are 1\%, 5\%, and 10\%. The CATE forms are plotted in Figure \ref{fig:cate-forms-parametric}.
\end{tablenotes}
\end{threeparttable}
\end{table}

\subsection{Additional Simulations for CLATE Heterogeneity Tests}
\label{sec:app-late}

This subsection reports additional Monte Carlo evidence for the CLATE heterogeneity test discussed in Section \ref{sec:ex-clate}. The null hypothesis is that the conditional local average treatment effect
\[
\tau_L(z)=\E[Y(1)-Y(0)\mid \mathcal C,Z=z]
\]
is constant in the covariate of interest \(Z\). As in the main simulation study, we set \(Z=X_1\), consider \(r\in\{3,5,10\}\), and report empirical rejection frequencies at the 1\%, 5\%, and 10\% nominal levels. We consider sample sizes \(n\in\{500,1000,2000\}\). All results are based on \(1{,}000\) Monte Carlo replications, and bootstrap \(p\)-values are computed using \(B=2{,}000\) multiplier replications.

The DGPs are constructed to satisfy the conditional LATE assumptions while allowing treatment to be endogenous. In all designs, \(A\in\{0,1\}\) is a binary instrument generated from \(q(X)=\mP(A=1\mid X)\). To generate the potential treatment statuses, let \(U_D\sim \mathrm{Unif}[0,1]\) be independent of \((A,\varepsilon)\) conditional on \(X\). Given two functions \(\pi_0(X)\) and \(\pi_C(X)\), define
\[
\pi_1(X):=\min\{\pi_0(X)+\pi_C(X),0.95\},
\qquad
D(0):=\1\{U_D<\pi_0(X)\},
\qquad
D(1):=\1\{U_D<\pi_1(X)\}.
\]
Because \(\pi_1(X)\ge \pi_0(X)\), monotonicity \(D(1)\ge D(0)\) holds by construction. The complier group is
\[
\mathcal C=\{D(1)>D(0)\},
\]
and the observed treatment is generated as
\[
D=AD(1)+(1-A)D(0).
\]
Thus, the instrument affects treatment take-up through the shift from \(D(0)\) to \(D(1)\), and the magnitude of the first stage is governed by \(\pi_C(X)\).

Endogeneity is introduced by allowing the same latent variable \(U_D\) that determines treatment take-up to enter the untreated potential outcome:
\[
Y(0)=\mu_0(X)+\rho(U_D-0.5)+\varepsilon.
\]
Hence treatment is generally not unconfounded conditional on \(X\). The treated potential outcome is \(Y(1)=Y(0)+\tau_L(X_1)\), and the observed outcome is \(Y=DY(1)+(1-D)Y(0)\). Since \(A\) is generated independently of \((U_D,\varepsilon)\) conditional on \(X\), the conditional IV independence condition holds.

We consider the following three designs.

\paragraph{Case I: Smooth IV logit design.}
The covariates satisfy \(X_j\stackrel{iid}{\sim}\mathrm{Unif}[0,1]\), \(j=1,\ldots,r\), and \(U_j=X_j-0.5\). The instrument propensity score is
\[
q(X)=0.15+0.70\Lambda\{q_{\mathrm{index}}(X)\}.
\]
For \(r=3\),
\[
q_{\mathrm{index}}(X)
=
-0.1+0.8U_1+0.6U_2+0.4U_3+0.5U_1U_2,
\]
whereas for \(r>3\), the additional term \(0.3\bar U\) is included, with
\[
\bar U=(r-3)^{-1}\sum_{j=4}^r(X_j-0.5).
\]
The baseline outcome is \(\mu_0(X)=0.5U_2+0.3U_2U_3+0.2U_1^2\) for \(r=3\), and additionally includes \(0.15\bar U\) for \(r>3\). The first-stage components are smooth functions:
\[
\pi_0(X)=0.10+0.25\Lambda\{b_{\mathrm{index}}(X)\},
\qquad
\pi_C(X)=0.15+0.25\Lambda\{c_{\mathrm{index}}(X)\},
\]
where, for \(r=3\),
\[
b_{\mathrm{index}}(X)=-0.3+0.6U_1+0.4U_2+0.3U_3,
\qquad
c_{\mathrm{index}}(X)=0.2+0.5U_1-0.3U_2+0.4U_1^2,
\]
and both indices additionally include \(0.2\bar U\) for \(r>3\). Under the alternative,
\[
\tau_L(X_1)=\tau_0+(X_1-0.5).
\]

\paragraph{Case II: Threshold IV design.}
The covariates again satisfy \(X_j\stackrel{iid}{\sim}\mathrm{Unif}[0,1]\). The instrument propensity score is
\[
q(X)=0.10+0.80S(X),
\]
where \(S(X)=(X_1+X_2+X_3)/3\) for \(r=3\), and
\[
S(X)=0.8(X_1+X_2+X_3)/3+0.2(r-3)^{-1}\sum_{j=4}^r X_j
\]
for \(r>3\). The baseline outcome is \(\mu_0(X)=0.5(X_2-0.5)+0.5(X_3-0.5)\) for \(r=3\), with an additional term \(0.15(r-3)^{-1}\sum_{j=4}^r(X_j-0.5)\) when \(r>3\). The first-stage components are
\[
\pi_0(X)=0.10+0.20S(X),
\qquad
\pi_C(X)=0.15+0.25(0.6X_1+0.4X_2).
\]
Under the alternative,
\[
\tau_L(X_1)=\tau_0+(X_1-0.5)\1\{X_1>0.5\}.
\]

\paragraph{Case III: Fractional nonlinear IV design.}
The covariates satisfy \(X_j\stackrel{iid}{\sim}\mathrm{Beta}(2,2)\), \(j=1,\ldots,r\), and \(U_j=X_j-0.5\). The instrument propensity score is
\[
q(X)=0.15+0.70\Phi\{q_{\mathrm{index}}(X)\},
\]
where
\[
q_{\mathrm{index}}(X)
=
\frac{
-0.10+0.10\sum_{j=1}^r X_j^2+0.50U_1+0.40U_2+0.30U_1U_2
}{
\exp\{-0.10\sum_{j=1}^r X_j^2\}
}.
\]
The baseline outcome is
\[
\mu_0(X)=0.4U_1+0.3U_2+0.2U_3^2+0.10r^{-1}\sum_{j=1}^r U_j.
\]
The first-stage components are
\[
\pi_0(X)=0.10+0.25\Phi\{b_{\mathrm{index}}(X)\},
\qquad
\pi_C(X)=0.15+0.25\Phi\{c_{\mathrm{index}}(X)\},
\]
where
\[
b_{\mathrm{index}}(X)=-0.25+0.50U_1+0.30U_2+0.20U_1U_2,
\qquad
c_{\mathrm{index}}(X)=0.15+0.40U_1-0.20U_2+0.20r^{-1}\sum_{j=1}^r U_j.
\]
Under the alternative,
\[
\tau_L(X_1)=\tau_0+2(X_1-0.4)^2.
\]

In all three cases, the null is \(\tau_L(X_1)\equiv\tau_0\) with \(\tau_0=0.5\). The disturbance is either homoskedastic, \(\varepsilon\sim N(0,\sigma_u^2)\) with \(\sigma_u=0.25\), or heteroskedastic, \(\varepsilon\sim N(0,\sigma_u^2\{1+|X_1-0.5|\}^2)\). The parameter \(\rho=0.40\) controls the strength of treatment endogeneity through the latent variable \(U_D\).

Tables \ref{tab:clate_size_homo}--\ref{tab:clate_power_hetero} report the corresponding empirical rejection frequencies based on Proposition \ref{prop:clate-process-uniform} and the multiplier bootstrap procedure in \eqref{eq:bootstrap_process_clate}.

\begin{table}[!htbp]
\centering
\caption{Empirical size under homoskedastic errors for CLATE heterogeneity tests.}
\label{tab:clate_size_homo}
\begin{tabular*}{0.85\textwidth}{@{\extracolsep{\fill}} llcccccc @{}}
\toprule
& & \multicolumn{2}{c}{$n=500$} & \multicolumn{2}{c}{$n=1000$} & \multicolumn{2}{c}{$n=2000$} \\
\cmidrule(lr){3-4}\cmidrule(lr){5-6}\cmidrule(lr){7-8}
$r$ & $\alpha$ & KS & CvM & KS & CvM & KS & CvM \\
\midrule
\multicolumn{8}{l}{\textit{Case I: Smooth IV logit design}} \\
\midrule
\multirow{3}{*}{3}
& 1\%  & 1.20 & 0.80 & 1.10 & 0.50 & 1.30 & 1.30 \\

& 5\%  & 4.10 & 5.00 & 4.50 & 5.20 & 5.70 & 5.30 \\

& 10\%  & 8.40 & 8.50 & 8.80 & 9.10 & 9.70 & 9.80 \\
\addlinespace
\multirow{3}{*}{5}
& 1\%  & 0.90 & 0.80 & 0.40 & 0.50 & 1.10 & 0.80 \\

& 5\%  & 5.70 & 5.50 & 4.50 & 4.30 & 4.30 & 5.10 \\

& 10\%  & 10.60 & 9.80 & 9.30 & 9.50 & 9.60 & 10.10 \\
\addlinespace
\multirow{3}{*}{10}
& 1\%  & 1.10 & 0.90 & 1.00 & 1.20 & 1.10 & 1.10 \\

& 5\%  & 4.70 & 4.80 & 5.80 & 4.90 & 5.70 & 5.30 \\

& 10\%  & 9.40 & 9.70 & 10.00 & 10.10 & 10.40 & 10.20 \\
\midrule
\multicolumn{8}{l}{\textit{Case II: Threshold IV design}} \\
\midrule
\multirow{3}{*}{3}
& 1\%  & 0.90 & 1.00 & 0.80 & 1.00 & 1.40 & 1.40 \\

& 5\%  & 4.70 & 5.60 & 5.00 & 4.60 & 5.20 & 5.40 \\

& 10\%  & 10.50 & 10.80 & 11.10 & 10.20 & 9.50 & 8.90 \\
\addlinespace
\multirow{3}{*}{5}
& 1\%  & 1.20 & 0.80 & 1.20 & 1.40 & 0.90 & 0.80 \\

& 5\%  & 5.40 & 5.80 & 5.30 & 5.00 & 4.70 & 4.60 \\

& 10\%  & 10.50 & 10.50 & 9.80 & 10.40 & 8.90 & 10.00 \\
\addlinespace
\multirow{3}{*}{10}
& 1\%  & 0.60 & 0.90 & 1.20 & 1.10 & 0.60 & 0.80 \\

& 5\%  & 4.50 & 4.80 & 5.10 & 5.40 & 4.80 & 4.20 \\

& 10\%  & 9.10 & 9.80 & 10.70 & 10.60 & 10.00 & 9.90 \\
\midrule
\multicolumn{8}{l}{\textit{Case III: Fractional nonlinear IV design}} \\
\midrule
\multirow{3}{*}{3}
& 1\%  & 1.10 & 1.50 & 1.00 & 1.60 & 0.90 & 0.80 \\

& 5\%  & 4.90 & 5.40 & 5.90 & 5.70 & 5.00 & 5.30 \\

& 10\%  & 10.00 & 10.90 & 10.70 & 10.60 & 11.00 & 10.60 \\
\addlinespace
\multirow{3}{*}{5}
& 1\%  & 1.10 & 1.40 & 0.90 & 1.10 & 1.30 & 0.60 \\

& 5\%  & 5.70 & 4.40 & 6.50 & 5.60 & 5.80 & 5.30 \\

& 10\%  & 10.10 & 9.80 & 11.20 & 12.80 & 10.50 & 11.10 \\
\addlinespace
\multirow{3}{*}{10}
& 1\%  & 1.10 & 0.70 & 0.80 & 0.70 & 1.10 & 1.20 \\

& 5\%  & 3.90 & 3.80 & 4.70 & 4.00 & 6.00 & 6.20 \\

& 10\%  & 9.20 & 8.60 & 8.00 & 8.50 & 12.40 & 11.20 \\
\bottomrule
\end{tabular*}

\vspace{0.4em}
\begin{minipage}{0.8\textwidth}
\footnotesize
\textit{Notes.} Entries report rejection frequencies (in percent) over 1,000 Monte Carlo replications. The nominal significance levels are 1\%, 5\%, and 10\%. Bootstrap critical values are based on 2,000 multiplier bootstrap draws.
\end{minipage}
\end{table}

\begin{table}[!htbp]
\centering
\caption{Empirical size under heteroskedasticity errors for CLATE heterogeneity tests.}
\label{tab:clate_size_hetero}
\begin{tabular*}{0.85\textwidth}{@{\extracolsep{\fill}} llcccccc @{}}
\toprule
& & \multicolumn{2}{c}{$n=500$} & \multicolumn{2}{c}{$n=1000$} & \multicolumn{2}{c}{$n=2000$} \\
\cmidrule(lr){3-4}\cmidrule(lr){5-6}\cmidrule(lr){7-8}
$r$ & $\alpha$ & KS & CvM & KS & CvM & KS & CvM \\
\midrule
\multicolumn{8}{l}{\textit{Case I: Smooth IV logit design}} \\
\midrule
\multirow{3}{*}{3}
& 1\%  & 0.90 & 0.80 & 1.00 & 0.70 & 1.30 & 1.10 \\

& 5\%  & 4.20 & 5.10 & 4.50 & 5.60 & 5.70 & 5.00 \\

& 10\%  & 9.50 & 8.70 & 8.50 & 10.20 & 9.90 & 9.50 \\
\addlinespace
\multirow{3}{*}{5}
& 1\%  & 0.90 & 1.10 & 0.60 & 0.50 & 0.90 & 1.20 \\

& 5\%  & 6.30 & 5.50 & 4.90 & 4.30 & 4.60 & 4.70 \\

& 10\%  & 10.70 & 10.10 & 9.00 & 9.10 & 10.60 & 10.90 \\
\addlinespace
\multirow{3}{*}{10}
& 1\%  & 1.10 & 0.90 & 1.20 & 1.20 & 1.00 & 1.40 \\

& 5\%  & 4.40 & 4.60 & 5.20 & 5.10 & 5.50 & 5.20 \\

& 10\%  & 9.00 & 9.10 & 9.90 & 9.90 & 10.10 & 10.30 \\
\midrule
\multicolumn{8}{l}{\textit{Case II: Threshold IV design}} \\
\midrule
\multirow{3}{*}{3}
& 1\%  & 1.20 & 0.90 & 1.00 & 1.00 & 1.40 & 1.30 \\

& 5\%  & 5.10 & 6.00 & 5.80 & 4.90 & 5.80 & 5.20 \\

& 10\%  & 10.40 & 10.50 & 10.00 & 9.60 & 9.50 & 9.20 \\
\addlinespace
\multirow{3}{*}{5}
& 1\%  & 1.10 & 1.00 & 1.30 & 1.10 & 0.80 & 0.80 \\

& 5\%  & 5.90 & 5.90 & 5.20 & 5.40 & 4.50 & 4.90 \\

& 10\%  & 10.10 & 11.70 & 9.70 & 10.00 & 9.30 & 10.70 \\
\addlinespace
\multirow{3}{*}{10}
& 1\%  & 0.80 & 0.60 & 1.00 & 1.30 & 0.70 & 0.60 \\

& 5\%  & 5.20 & 5.30 & 5.20 & 5.50 & 5.40 & 4.00 \\

& 10\%  & 9.40 & 9.80 & 10.50 & 10.20 & 10.10 & 10.00 \\
\midrule
\multicolumn{8}{l}{\textit{Case III: Fractional nonlinear IV design}} \\
\midrule
\multirow{3}{*}{3}
& 1\%  & 1.00 & 1.70 & 1.20 & 1.50 & 0.90 & 0.80 \\

& 5\%  & 4.70 & 5.40 & 5.70 & 6.00 & 5.80 & 6.00 \\

& 10\%  & 10.80 & 10.40 & 10.50 & 10.50 & 11.10 & 11.10 \\
\addlinespace
\multirow{3}{*}{5}
& 1\%  & 1.10 & 1.30 & 0.90 & 1.00 & 1.10 & 0.80 \\

& 5\%  & 4.70 & 4.10 & 5.50 & 5.70 & 6.00 & 6.10 \\

& 10\%  & 9.80 & 9.90 & 11.70 & 12.50 & 10.80 & 10.40 \\
\addlinespace
\multirow{3}{*}{10}
& 1\%  & 1.00 & 0.80 & 0.60 & 0.50 & 1.30 & 1.00 \\

& 5\%  & 3.90 & 3.60 & 4.80 & 3.80 & 6.20 & 6.70 \\

& 10\%  & 8.90 & 9.40 & 8.30 & 8.40 & 12.90 & 11.20 \\
\bottomrule
\end{tabular*}

\vspace{0.4em}
\begin{minipage}{0.8\textwidth}
\footnotesize
\textit{Notes.} Entries report rejection frequencies (in percent) over 1,000 Monte Carlo replications. The nominal significance levels are 1\%, 5\%, and 10\%. Bootstrap critical values are based on 2,000 multiplier bootstrap draws.
\end{minipage}
\end{table}

\begin{table}[!htbp]
\centering
\caption{Empirical power under homoskedastic errors for CLATE heterogeneity tests.}
\label{tab:clate_power_homo}
\begin{tabular*}{0.85\textwidth}{@{\extracolsep{\fill}} llcccccc @{}}
\toprule
& & \multicolumn{2}{c}{$n=500$} & \multicolumn{2}{c}{$n=1000$} & \multicolumn{2}{c}{$n=2000$} \\
\cmidrule(lr){3-4}\cmidrule(lr){5-6}\cmidrule(lr){7-8}
$r$ & $\alpha$ & KS & CvM & KS & CvM & KS & CvM \\
\midrule
\multicolumn{8}{l}{\textit{Case I: Smooth IV logit design}} \\
\midrule
\multirow{3}{*}{3}
& 1\%  & 46.10 & 53.30 & 82.60 & 87.80 & 99.40 & 99.70 \\

& 5\%  & 68.20 & 73.90 & 94.00 & 95.50 & 100.00 & 99.90 \\

& 10\%  & 77.60 & 81.00 & 96.20 & 97.70 & 100.00 & 100.00 \\
\addlinespace
\multirow{3}{*}{5}
& 1\%  & 46.00 & 53.10 & 80.40 & 87.60 & 98.10 & 99.30 \\

& 5\%  & 67.30 & 74.00 & 92.90 & 95.10 & 99.70 & 99.90 \\

& 10\%  & 77.40 & 82.80 & 95.90 & 97.20 & 100.00 & 100.00 \\
\addlinespace
\multirow{3}{*}{10}
& 1\%  & 39.60 & 47.70 & 79.70 & 85.70 & 98.80 & 99.50 \\

& 5\%  & 62.60 & 69.20 & 91.40 & 95.50 & 99.80 & 99.80 \\

& 10\%  & 73.60 & 78.20 & 95.80 & 97.60 & 99.80 & 99.80 \\
\midrule
\multicolumn{8}{l}{\textit{Case II: Threshold IV design}} \\
\midrule
\multirow{3}{*}{3}
& 1\%  & 12.50 & 15.00 & 28.40 & 30.90 & 62.70 & 68.50 \\

& 5\%  & 29.50 & 31.90 & 50.40 & 53.80 & 84.00 & 86.40 \\

& 10\%  & 40.60 & 45.30 & 62.50 & 67.10 & 89.40 & 93.40 \\
\addlinespace
\multirow{3}{*}{5}
& 1\%  & 13.40 & 14.70 & 29.20 & 33.30 & 62.40 & 65.40 \\

& 5\%  & 29.00 & 32.60 & 53.60 & 57.70 & 81.80 & 85.10 \\

& 10\%  & 40.60 & 44.00 & 63.80 & 69.20 & 88.30 & 91.40 \\
\addlinespace
\multirow{3}{*}{10}
& 1\%  & 11.50 & 13.80 & 30.80 & 34.40 & 64.90 & 68.20 \\

& 5\%  & 26.50 & 29.20 & 51.50 & 56.40 & 84.20 & 86.90 \\

& 10\%  & 39.00 & 42.20 & 61.50 & 66.00 & 90.00 & 92.50 \\
\midrule
\multicolumn{8}{l}{\textit{Case III: Fractional nonlinear IV design}} \\
\midrule
\multirow{3}{*}{3}
& 1\%  & 5.90 & 7.60 & 13.00 & 14.90 & 37.90 & 41.00 \\

& 5\%  & 18.80 & 21.20 & 34.90 & 37.80 & 62.10 & 66.10 \\

& 10\%  & 28.10 & 31.20 & 48.60 & 52.00 & 74.70 & 80.00 \\
\addlinespace
\multirow{3}{*}{5}
& 1\%  & 6.20 & 8.40 & 15.90 & 17.40 & 36.70 & 38.40 \\

& 5\%  & 19.70 & 21.50 & 36.00 & 40.10 & 64.50 & 67.20 \\

& 10\%  & 29.50 & 32.40 & 50.00 & 52.00 & 77.20 & 79.60 \\
\addlinespace
\multirow{3}{*}{10}
& 1\%  & 6.10 & 7.60 & 14.50 & 16.00 & 41.90 & 44.80 \\

& 5\%  & 21.20 & 21.00 & 36.40 & 38.70 & 67.70 & 69.80 \\

& 10\%  & 32.20 & 33.70 & 49.20 & 51.80 & 78.10 & 82.40 \\
\bottomrule
\end{tabular*}

\vspace{0.4em}
\begin{minipage}{0.8\textwidth}
\footnotesize
\textit{Notes.} Entries report rejection frequencies (in percent) over 1,000 Monte Carlo replications. The nominal significance levels are 1\%, 5\%, and 10\%. Bootstrap critical values are based on 2,000 multiplier bootstrap draws.
\end{minipage}
\end{table}

\begin{table}[!htbp]
\centering
\caption{Empirical power under heteroskedasticity errors for CLATE heterogeneity tests.}
\label{tab:clate_power_hetero}
\begin{tabular*}{0.85\textwidth}{@{\extracolsep{\fill}} llcccccc @{}}
\toprule
& & \multicolumn{2}{c}{$n=500$} & \multicolumn{2}{c}{$n=1000$} & \multicolumn{2}{c}{$n=2000$} \\
\cmidrule(lr){3-4}\cmidrule(lr){5-6}\cmidrule(lr){7-8}
$r$ & $\alpha$ & KS & CvM & KS & CvM & KS & CvM \\
\midrule
\multicolumn{8}{l}{\textit{Case I: Smooth IV logit design}} \\
\midrule
\multirow{3}{*}{3}
& 1\%  & 31.10 & 36.40 & 63.10 & 68.80 & 96.30 & 97.90 \\

& 5\%  & 53.50 & 58.30 & 82.80 & 87.10 & 99.40 & 99.70 \\

& 10\%  & 65.20 & 69.00 & 90.10 & 92.80 & 99.90 & 99.90 \\
\addlinespace
\multirow{3}{*}{5}
& 1\%  & 32.50 & 36.10 & 64.60 & 71.10 & 93.60 & 95.90 \\

& 5\%  & 54.00 & 58.00 & 82.40 & 88.00 & 98.40 & 99.00 \\

& 10\%  & 64.90 & 69.50 & 89.70 & 92.70 & 99.10 & 99.30 \\
\addlinespace
\multirow{3}{*}{10}
& 1\%  & 26.50 & 32.00 & 64.20 & 69.50 & 94.20 & 95.90 \\

& 5\%  & 48.50 & 54.30 & 81.10 & 85.90 & 98.60 & 99.10 \\

& 10\%  & 60.10 & 65.50 & 87.70 & 91.00 & 99.40 & 99.70 \\
\midrule
\multicolumn{8}{l}{\textit{Case II: Threshold IV design}} \\
\midrule
\multirow{3}{*}{3}
& 1\%  & 7.50 & 8.70 & 15.90 & 17.20 & 41.10 & 42.60 \\

& 5\%  & 21.40 & 22.40 & 35.90 & 37.30 & 64.70 & 67.10 \\

& 10\%  & 31.30 & 32.30 & 48.90 & 49.30 & 76.10 & 78.60 \\
\addlinespace
\multirow{3}{*}{5}
& 1\%  & 8.90 & 9.50 & 16.80 & 17.90 & 41.90 & 42.30 \\

& 5\%  & 21.00 & 22.80 & 37.00 & 40.10 & 63.10 & 63.80 \\

& 10\%  & 31.90 & 34.20 & 49.80 & 52.10 & 74.40 & 75.80 \\
\addlinespace
\multirow{3}{*}{10}
& 1\%  & 6.80 & 7.80 & 18.90 & 20.10 & 43.40 & 43.60 \\

& 5\%  & 18.60 & 19.20 & 38.30 & 40.00 & 67.10 & 68.10 \\

& 10\%  & 28.10 & 29.10 & 49.20 & 51.70 & 77.00 & 78.30 \\
\midrule
\multicolumn{8}{l}{\textit{Case III: Fractional nonlinear IV design}} \\
\midrule
\multirow{3}{*}{3}
& 1\%  & 4.00 & 4.60 & 7.50 & 8.80 & 25.60 & 24.90 \\

& 5\%  & 14.40 & 17.10 & 24.70 & 24.40 & 48.90 & 49.70 \\

& 10\%  & 23.30 & 24.30 & 37.10 & 38.60 & 61.70 & 62.30 \\
\addlinespace
\multirow{3}{*}{5}
& 1\%  & 5.00 & 5.60 & 10.60 & 10.00 & 23.60 & 23.60 \\

& 5\%  & 15.00 & 16.20 & 26.90 & 27.40 & 48.40 & 47.10 \\

& 10\%  & 24.50 & 25.60 & 38.90 & 39.80 & 61.90 & 62.80 \\
\addlinespace
\multirow{3}{*}{10}
& 1\%  & 4.00 & 4.50 & 8.70 & 9.80 & 27.30 & 27.00 \\

& 5\%  & 15.40 & 15.50 & 27.60 & 27.60 & 52.40 & 51.10 \\

& 10\%  & 25.70 & 25.40 & 38.90 & 40.40 & 65.00 & 66.20 \\
\bottomrule
\end{tabular*}

\vspace{0.4em}
\begin{minipage}{0.8\textwidth}
\footnotesize
\textit{Notes.} Entries report rejection frequencies (in percent) over 1,000 Monte Carlo replications. The nominal significance levels are 1\%, 5\%, and 10\%. Bootstrap critical values are based on 2,000 multiplier bootstrap draws.
\end{minipage}
\end{table}
\clearpage
\putbib 
\end{bibunit}

\end{document}